%% file: linear_bosonic_control_ArxivV7.tex
\documentclass[
  aps,
  prl,
  reprint,
  superscriptaddress,
  longbibliography,
  nofootinbib
]{revtex4-2}

\usepackage{amsmath,amssymb,amsthm,mathtools,mathrsfs}
\usepackage{graphicx}
\usepackage{xcolor}
\usepackage[normalem]{ulem}

\usepackage{comment}
\usepackage{microtype}
\usepackage[hidelinks]{hyperref}

\newcommand{\Tr}{\operatorname{Tr}}
\newcommand{\dd}{\,\mathrm{d}}
\newcommand{\id}{\mathbb I}
\newcommand{\ket}[1]{\left|#1\right\rangle}
\newcommand{\bra}[1]{\left\langle#1\right|}

\newcommand{\picsup}[1]{\scalebox{0.9}{$\scriptstyle(\mathrm{#1})$}}
\newcommand{\assumptionlink}[2][]{%
  \hyperref[ass:#2]{\textbf{%
    \if\relax\detokenize{#1}\relax#2\else#1\fi}}}

\theoremstyle{definition}
\newtheorem{definition}{Definition}
\theoremstyle{plain}
\newtheorem{theorem}{Theorem}
\newtheorem{lemma}{Lemma}
\newtheoremstyle{plaincolon}
  {}{}
  {\itshape}{}
  {\bfseries}{:}{5pt plus 1pt minus 1pt}{}
\theoremstyle{plaincolon}
\newtheorem{colonlemma}[lemma]{Lemma}
\theoremstyle{remark}

\newcounter{bound}

\hypersetup{
  pdftitle={A Square-Root Barrier to Quantum Gate Speed under Linear Coupling},
  pdfauthor={Pablo Restrepo Gaviria and Mischa P. Woods}
}

\begin{document}

\title{
	A Square-Root Barrier to Quantum Gate Speed under Linear Coupling}

\author{Pablo Restrepo Gaviria}
\email{restrepo.pablo2@gmail.com}
\affiliation{ETH Zurich, Zurich, Switzerland}

\author{Mischa P. Woods}
\email{mischa.woods@ens-lyon.fr\quad mischa.woods@gmail.com}
\affiliation{ENS Lyon, Inria, France}

\begin{abstract}
Faster quantum gates can suppress the decoherence accumulated during a computation, but in superconducting processors, require  stronger microwave pulses. These in turn can also increase leakage, off-resonant excitation, stray-field errors, and crosstalk.  This creates a central energy--speed--error tradeoff: can quantum state engineering make a gate faster without paying proportionally more drive energy?  We address this question by treating the driving pulse as a quantum bosonic field rather than a classical waveform.  For a finite-dimensional system coupled linearly to that field, we prove that gate transition rates at fixed trace distance between the initial and final system states grow at most as the square root of the pulse energy, under pulse-uniformity conditions, without assuming an upper bound on the coupling strength.  The bound permits arbitrary pulse states, including squeezed and non-Gaussian states, as well as drive--system entanglement and back action; coherent Gaussian pulses attain its energy exponent.  Thus squeezing or other state engineering alone cannot replace the square-root energy law of a conventional linear drive by the linear scaling allowed by general quantum speed limits.  Achieving that improvement requires changing the interaction class, in addition to using a suitable nonclassical pulse, thereby identifying interaction nonlinearity as an essential resource for relaxing the practical gate-speed error tradeoff.
\end{abstract}

\maketitle
Quantum gates are commonly driven by propagating bosonic fields.  In neutral-atom and trapped-ion
processors these are optical pulses addressing atomic transitions
\cite{Levine2019,Ballance2016}; in circuit quantum electrodynamics they are microwave pulses sent
through a transmission line to one or more superconducting qubits
\cite{Chow2011,Blais2021}.  Despite their different carrier frequencies and hardware, both are
captured at leading order by a finite-dimensional system coupled linearly to quantized field
modes (which we abbreviate to {\it linear coupling} and in its absence, {\it nonlinear coupling}).  The drive is usually replaced by a classical waveform.  Retaining it as a quantum
system asks a sharper resource question: at fixed coupling, how rapidly can energy in the pulse
move the target system through a prescribed gate transition?  Figure~\ref{fig:platform-tradeoff}(a)
schematizes the common optical and microwave realizations.

The distinction is consequential.  In generic quantum systems, general quantum speed limits permit an orthogonalization rate $\omega$ 
linear in the mean energy \cite{Margolus1998,Lloyd2000},
\begin{align}
  \omega \propto E.
  \label{eq:qfc-linear-scaling}
\end{align}
Under which circumstances an implemented gate, with gate frequency $\omega_{\mathrm{gate}}$, could attain this fundamental scaling was however unclear.
Ref.~\cite{Woods2024} answered affirmatively by constructing an autonomous ``quantum
frequential computer'' whose nonclassical phase-squeezed drive pulse and localized nonlinear drive-qubit 
interaction attain Eq.~\eqref{eq:qfc-linear-scaling}.  Within the same framework, it proved
that conventional nonsqueezed, standard-quantum-limited (SQL) pulses are restricted to the exponent one half,
\begin{align}
  \omega_{\mathrm{gate}}T_0\propto\sqrt{T_0E},
  \label{eq:sql-square-root-scaling}
\end{align}
where $T_0$ is a background reference time.  Thus Ref.~\cite{Woods2024} proves that
nonclassical squeezing is necessary to obtain gate frequencies which scale linearly with pulse energy.
Moreover, its construction shows that a localized nonlinear interaction is sufficient when combined with
squeezing, but does not establish that nonlinear coupling is necessary.  This leaves their individual
roles open: can a suitably squeezed, or more exotic, state recover linear gate-frequency scaling
through the ubiquitous linear bosonic coupling alone?

Here we prove that it cannot, under assumptions that express finite
drive energy during the accepted gate, and a non-vanishing gate.  Thus the prior
linear scaling is not a state-only effect: within this broad control setting it requires leaving
the linear-interaction class as well as leaving the SQL state class.

This distinction also has immediate practical value [Fig.~\ref{fig:platform-tradeoff}(b)].  In superconducting processors, reducing
microwave-drive energy suppresses drive-induced leakage, off-resonant excitation, stray-field
errors, and crosstalk, but lengthens entangling gate times, allowing background decoherence to
accumulate within a quantum-error-correction cycle.  Increasing the drive shortens the gates but
aggravates these parasitic errors \cite{Kandala2021,Malekakhlagh2022,Wei2022,Gorshkov2026}; two-qubit gates are
consequently a leading contribution to processor error budgets \cite{GoogleQEC2025}.
Related speed--coherence and heating constraints also arise in single-qubit control, motivating
subharmonic driving schemes that exploit the transmon's native nonlinearity \cite{Xia2026}.  Replacing
the square-root law in Eq.~\eqref{eq:sql-square-root-scaling} by the linear law in
Eq.~\eqref{eq:qfc-linear-scaling} would deliver a prescribed gate rate with parametrically less
pulse energy, relaxing this speed--error tradeoff.  In this Letter, our result says that squeezing a conventional
transmission-line drive, without engineering its interaction, cannot supply that advantage.

\begin{figure*}[t]
  \centering
  \includegraphics[width=\textwidth]{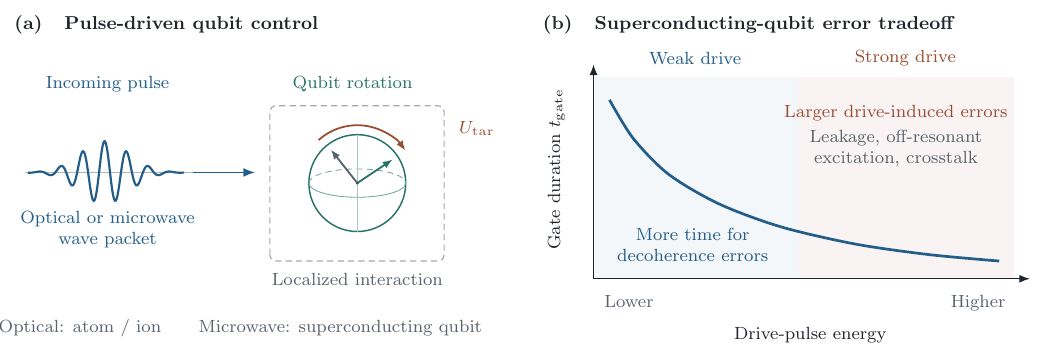}
  \caption{Pulse-driven control and its practical energy--time tradeoff.
  (a) An optical or microwave wave packet, initially with vanishing mean interaction energy,
  drives a localized qubit rotation toward a prescribed
  target state.  The same geometry describes atomic and ionic optical control and microwave
  control of superconducting qubits. Typically, the drive-pulse energy greatly exceeds
  the transition energy of the driven qubit(s).
  (b) Schematic tradeoff for superconducting-qubit gates: weaker pulses require longer gates,
  leaving more time for decoherence, whereas stronger pulses shorten gates but can increase
  leakage, off-resonant excitation, and crosstalk.}
  \label{fig:platform-tradeoff}
\end{figure*}

\textit{Model.---}
Set $\hbar=1$.  A finite-dimensional controlled system $Q$ and bosonic continuum $F$ evolve under
the time-independent Hamiltonian
\begin{equation}
 \begin{split}
 H&:=H_Q\otimes\id_F+\id_Q\otimes H_F+V,\\
 H_F&:=\int_{\mathbb R}\!\dd k\,\omega(k)a_k^\dagger a_k,\\
 V&:=\sigma\otimes B+\sigma^\dagger\otimes B^\dagger,
 \qquad B:=\int_{\mathbb R}\!\dd k\,g(k)a_k .
 \end{split}
 \label{eq:model}
\end{equation}
Here $k\in\mathbb R$ is a signed mode coordinate, such as momentum or a waveguide
propagation coordinate, whereas $\omega(k)\ge0$ is the physical mode frequency; no particular
realization is assumed. Appropriate choices of \(\sigma\) describe dipole coupling both with and without the rotating-wave approximation. We normalize $\lVert\sigma\rVert=1$ and require
$g\in L^2(\mathbb R)$ and finite energy-weighted coupling
$C_g:=\int_{g\ne0}|g(k)|^2/\omega(k)\,\dd k$.  Hence uncoupled zero-frequency modes are allowed,
but $\omega(k)>0$ almost everywhere where $g(k)\ne0$.  The full Hamiltonian is self-adjoint and
bounded below.  The initial joint state $\rho_0$ can be arbitrary, e.g. mixed, non-Gaussian, squeezed, and
correlated across $QF$; it need only have the displayed finite-energy expectations.

Let $E:=\Tr(\rho_0H)-\inf\operatorname{spec}H$ be the mean initial energy above the actual
interacting spectral bottom.  Write
$H_0:=H_Q\otimes\id_F+\id_Q\otimes H_F$,
$U^{\picsup{I}}(t):=e^{iH_0t}e^{-iHt}$, and
$\rho^{\picsup{I}}(t):=U^{\picsup{I}}(t)\rho_0U^{\picsup{I}\dagger}(t)$ for interaction-picture dynamics,
with reduced state $\rho_Q^{\picsup{I}}(t):=\Tr_F\rho^{\picsup{I}}(t)$.
An accepted \hyperref[def:gate-interval]{gate interval}
$I:=[t_{\mathrm{in}},t_{\mathrm{out}}]$ has duration
$\Delta t:=t_{\mathrm{out}}-t_{\mathrm{in}}>0$ and system states
$\rho_{\mathrm{in}}:=\rho_Q^{\picsup{I}}(t_{\mathrm{in}})$ and
$\rho_{\mathrm{out}}:=\rho_Q^{\picsup{I}}(t_{\mathrm{out}})$,
assumed to have positive trace distance $D(\rho_{\mathrm{out}},\rho_{\mathrm{in}})>0$.
The gate frequency is $\omega_{\mathrm{gate}}:=2\pi/\Delta t$.
These necessary conditions suffice for deriving upper bounds on gate speed.
With
$\rho^{\picsup{S}}(t):=e^{-iHt}\rho_0e^{iHt}$, the instantaneous free-field energy is
$E_F(t):=\Tr[\rho^{\picsup{S}}(t)(\id_Q\otimes H_F)]
=\Tr[\rho^{\picsup{I}}(t)(\id_Q\otimes H_F)]$.  We define the finite field-energy factor
$C_{\mathrm{fld}}$ via the inequality $E_F(t)\le C_{\mathrm{fld}}E$ for all $t\in I$.

\textit{Fixed-Hamiltonian bound.---}
We bound how quickly the drive changes the system state and integrate this rate over the gate interval to obtain
\begin{equation}
 \omega_{\mathrm{gate}}
 \le \frac{4\pi}{D(\rho_{\mathrm{out}},\rho_{\mathrm{in}})}\sqrt{C_gC_{\mathrm{fld}}E} .
 \label{eq:fixed-bound}
\end{equation}
Theorem~\ref{thm:exact-speed} of the S.M. gives the precise statement and proof \cite{SupplementalMaterial}.  For a fixed Hamiltonian, $C_g$ is fixed.
Consequently, at fixed $D(\rho_{\mathrm{out}},\rho_{\mathrm{in}})>0$, increasing only the initial state's energy cannot improve the exponent one half,
irrespective of squeezing or other nonclassical structure.  The conclusion is stronger than a
semiclassical estimate because it follows from the exact joint state and permits arbitrary
drive--system correlations throughout the gate.
In cases where there is a common preparation before the pulse reaches $Q$, the interaction energy vanishes, so
$E=E_F(0)+E_Q$, where $E_Q$ is the fixed initial controlled-system contribution above the
ground-state energy.  Equation~\eqref{eq:fixed-bound} is therefore also a square-root bound in the
drive energy $E_F(0)$.  We state subsequent bounds using the invariant total energy $E$, but the
same observation applies throughout.

Related fixed-coupling square-root scaling was derived for a coherent drive in the
single-mode Jaynes--Cummings model \cite{Pratapsi2024}.  Neither that result nor
Eq.~\eqref{eq:fixed-bound} alone excludes linear-in-energy rates when coupling resources
grow with energy.  The following sections exclude such scaling in reference-time units
for families satisfying uniform stability or pulse-uniformity conditions, together with
uniform energy control and a nonvanishing $D(\rho_{\mathrm{out},\lambda},\rho_{\mathrm{in},\lambda})$.

\textit{Uniformly bounded-from-below-Hamiltonian bound.---}
One might instead change the Hamiltonian as the pulse energy grows---for example, by altering the
coupling strength, bandwidth, or dispersion.  We describe such protocols by a family index
$\lambda$, keeping the finite-dimensional system Hilbert space common while allowing
$H_{Q,\lambda}$, $\omega_\lambda$, $g_\lambda$, $\sigma_\lambda$, the initial state, and the gate interval
$I_\lambda=[t_{\mathrm{in},\lambda},t_{\mathrm{out},\lambda}]$ to vary.  Applying the preceding
definitions memberwise gives $\Delta t_\lambda:=t_{\mathrm{out},\lambda}-t_{\mathrm{in},\lambda}$,
$\rho_{\mathrm{in/out},\lambda}:=\Tr_F\rho_\lambda^{\picsup{I}}(t_{\mathrm{in/out},\lambda})$, and
$\omega_{\mathrm{gate},\lambda}:=2\pi/\Delta t_\lambda$.  The field operator, interaction, and Schr\"odinger-picture state are denoted $B_\lambda$, $V_\lambda$, and $\rho_\lambda^{\picsup{S}}(t)$, respectively.  A fixed,
scale-covariant rule assigns each member a reference time $T_\lambda$
(see Sec.~\ref{sec:reference-time} of the S.M.): under a global rescaling
$H_\lambda\mapsto\alpha_\lambda H_\lambda+c_\lambda I$, it gives
$T_\lambda\mapsto T_\lambda/\alpha_\lambda$.  This supplies a background reference time and
prevents a mere fast-forwarding of the entire dynamics from being counted as an improved gate.
All assumptions \assumptionlink{A0}--\assumptionlink{A2} and \assumptionlink{B1}--\assumptionlink{B2} introduced below are invariant under this
joint rescaling, so none imposes an absolute normalization of $T_\lambda$.

Three assumptions isolate a uniform, physically comparable family.  \assumptionlink{A0} fixes a common
Fock vacuum and free-field zero and uniformly bounds, in units of $T_\lambda$, the separation
between the system spectral ceiling and the interacting ground energy.  \assumptionlink{A1} requires the same $\lambda$-independent field-energy factor
$C_{\mathrm{fld}}$ throughout the family.  This allows energy exchange, back action, and strong transient
$Q$--$F$ entanglement.  A physically natural condition implying \assumptionlink{A1} is that the initially prepared drive contributes no mean interaction energy, while the total excitation energy $E_\lambda$, measured above the interacting ground energy, does not become arbitrarily small compared with the local gate energy scale: the spectral width of the $Q$-system Hamiltonian terms involving the directly driven subsystem (Fig.~\ref{fig:platform-tradeoff}(a), Lemma~\ref{lem:preparation-A1}).  \assumptionlink{A2} requires uniform separation of $\rho_{\mathrm{out},\lambda}$ and $\rho_{\mathrm{in},\lambda}$,
$D(\rho_{\mathrm{out},\lambda},\rho_{\mathrm{in},\lambda})\ge\delta_0>0$.  Under \assumptionlink{A0}--\assumptionlink{A2}, stability itself bounds the allowed energy-weighted
coupling, and Theorem~\ref{thm:relative-speed} of the S.M. gives
\begin{equation}
 \omega_{\mathrm{gate},\lambda}T_\lambda
 \le c_{\mathrm{stab}}\sqrt{T_\lambda E_\lambda},
 \label{eq:uniform-bound}
\end{equation}
with a dimensionless $c_{\mathrm{stab}}$ independent of $\lambda$.  This gives an absolute square-root law whenever $T_\lambda$ is chosen $\lambda$-independent.

While \assumptionlink{A1}--\assumptionlink{A2} are mild and physically well-motivated assumptions, the dimensionless
lower-stability requirement in \assumptionlink{A0} is simple but not logically unavoidable: a family could
lower its interacting spectral bottom relative to its background timescale as $\lambda$ grows.
We therefore replace \assumptionlink{A0} by
two local conditions \assumptionlink{B1}--\assumptionlink{B2} which are tied to what it means operationally to deliver a drive pulse.  

\begin{figure}[t]
  \centering
  \includegraphics[width=\columnwidth]{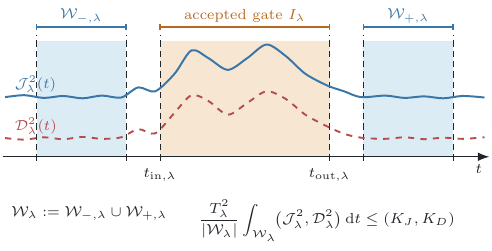}
  \vspace{-5pt}
  \caption{Assumption \assumptionlink{B1} (schematic two-window example).  The accepted gate interval
	$I_\lambda=[t_{\mathrm{in},\lambda},t_{\mathrm{out},\lambda}]$ is separated from the
	quiet windows ($\mathcal W_{-,\lambda}$ and $\mathcal W_{+,\lambda}$) by white gaps,
	which belong to neither $I_\lambda$ nor
	$\mathcal W_\lambda:=\mathcal W_{-,\lambda}\cup\mathcal W_{+,\lambda}$.
	Over $\mathcal W_\lambda$, \assumptionlink{B1} uniformly bounds the dimensionless time averages of
	the two complementary measures of system activation due to the drive,
	$\mathcal J_\lambda^2(t)$ and $\mathcal D_\lambda^2(t)$.
	The curves are illustrative: the bounds concern averages,
	not pointwise values.  Either window may be absent; the remaining windows may lie arbitrarily close to or far from the gate.
	Operational definitions are in the main text, while the S.M. gives a general definition allowing multiple
	disjoint quiet intervals.}
  \label{fig:B1-quiet-windows}
\end{figure}

\begin{figure*}[t]
  \centering
  \includegraphics[width=0.823\textwidth]
  {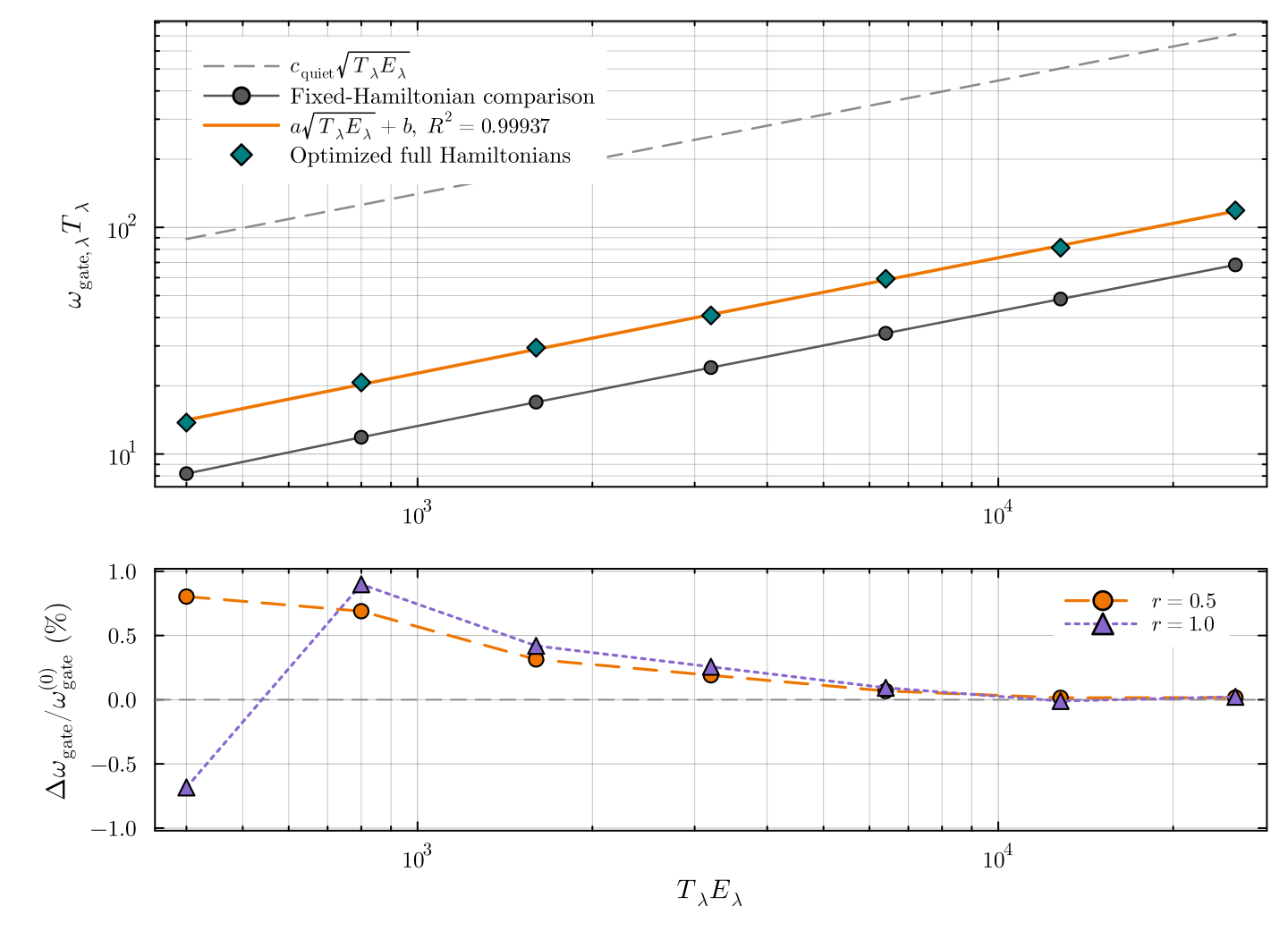}
  \caption{Full quantum dynamics with common assumption bounds.
  Upper: fastest feasible Hamiltonians found (diamonds), fixed-Hamiltonian
  comparison (circles), square-root fit (solid), and common bound
  $c_{\mathrm{quiet}}\sqrt{T_\lambda E_\lambda}$ (dashed), with
  $c_{\mathrm{quiet}}=4.44373$.  The search varies dispersion, coupling,
  qubit gap, and coherent packet.
  Lower: signed rate shifts for squeezing strengths $r=0.5,1$
  ($e^{-2r}$ is the minimum pulse-mode quadrature-variance ratio to vacuum)
  relative to coherent inputs ($r=0$), at equal total energy on the same
  Hamiltonians, with phases selected under the unchanged bounds.
  $T_\lambda=1$ throughout; search and convergence details are in the S.M.}
  \label{fig:irg-full-dynamics-extended}
\end{figure*}

\textit{Pulse-uniformity bound.---}
For \assumptionlink{B1}, we introduce two complementary measures of activation of system $Q$ due to the drive: the field activation $\mathcal J_\lambda(t)$ and the interaction activation $\mathcal D_\lambda(t)$.  Writing
$\rho_{F,\lambda}^{\picsup{S}}(t):=\Tr_Q\rho_\lambda^{\picsup{S}}(t)$ for the actual drive state, define
$\mathcal J_\lambda^2(t):=\Tr[\rho_{F,\lambda}^{\picsup{S}}(t)B_\lambda^\dagger B_\lambda]$.
Since $B_\lambda$ annihilates the collective field mode selected by the coupling,
$\mathcal J_\lambda(t)$ measures the normally ordered root-mean-square (RMS) amplitude
of the drive channel that directly interacts with $Q$ at time $t$, rather than the total
field energy.  We require `quiet' time intervals $\mathcal W_\lambda$ outside the gate interval
over which its dimensionless mean square is uniformly bounded:
$\frac{T_\lambda^2}{|\mathcal W_\lambda|}
\int_{\mathcal W_\lambda}\mathcal J_\lambda^2(t)\dd t
\le K_J$ for some finite $\lambda$-independent $K_J$.

Let $P_{+,\lambda}$ project onto the positive eigenspace of
$\rho_{\mathrm{out},\lambda}-\rho_{\mathrm{in},\lambda}$.
The binary measurement $\{P_{+,\lambda},\id_Q-P_{+,\lambda}\}$ optimally distinguishes
these states with equal priors: its outcome probability
$p_{+,\lambda}(s):=\Tr[(P_{+,\lambda}\otimes\id_F)\rho_{\lambda}^{\picsup{I}}(s)]$
characterises their trace distance through the relation
$\int_{I_\lambda}\!\!\dd s\,\frac{\dd}{\dd s}p_{+,\lambda}(s)
=D(\rho_{\mathrm{out},\lambda},\rho_{\mathrm{in},\lambda})$.
The integrand is
$\frac{\dd}{\dd s}p_{+,\lambda}(s)=-i\Tr[\rho_\lambda^{\picsup{S}}(s)C_{QF,\lambda}(s)]$.
With the freely rotated projector
$P_{Q,\lambda}(s):=e^{-iH_{Q,\lambda}s}P_{+,\lambda}e^{iH_{Q,\lambda}s}$,
the commutator $C_{QF,\lambda}(s):=[P_{Q,\lambda}(s)\otimes\id_F,V_\lambda]$
isolates the terms that transfer population between the measurement subspaces.

To test whether the off-pulse field remains quiet for the intended state change,
pair $\rho_{Q,\lambda}^{\picsup{S}}(s):=\Tr_F\rho_\lambda^{\picsup{S}}(s)$ with the field marginal at a quiet time $t$.
The product $\rho_{Q,\lambda}^{\picsup{S}}(s)\otimes\rho_{F,\lambda}^{\picsup{S}}(t)$ defines a hypothetical
test state; the actual joint dynamics may be entangled.
The interaction activation measures the RMS interaction across the two measurement subspaces,
bounding the induced population rate without cancellation between signed contributions.
Taking the supremum over gate times, define
\resizebox{0.99\columnwidth}{!}{$\mathcal D_\lambda^2(t):=\sup_{s\in I_\lambda}\Tr\bigl[(\rho_{Q,\lambda}^{\picsup{S}}(s)\otimes\rho_{F,\lambda}^{\picsup{S}}(t))C_{QF,\lambda}(s)^\dagger C_{QF,\lambda}(s)\bigr]$}.

We impose the uniform bound
$\frac{T_\lambda^2}{|\mathcal W_\lambda|}
\int_{\mathcal W_\lambda}\mathcal D_\lambda^2(t)\dd t\le K_D$
with finite $\lambda$-independent $K_D$.
Together with the bound involving $K_J$, this is
\assumptionlink{B1}; see Fig.~\ref{fig:B1-quiet-windows}.

Both $\mathcal J_\lambda^2(t)$ and $\mathcal D_\lambda^2(t)$ have complementary
operational interpretations: $\mathcal J_\lambda^2(t)$ tests the field's ability to excite
a detector, while $\mathcal D_\lambda^2(t)$ tests its ability to switch the gate's
distinguishing measurement outcomes.
Using the field state at $t$, a two-level detector initially in $\ket{g}$, with interaction
$V_{\mathrm{det},\lambda}:=\ket{e}\bra{g}\otimes B_\lambda+\mathrm{h.c.}$,
absorbs an excitation with probability $\epsilon^2\mathcal J_\lambda^2(t)+o(\epsilon^2)$
after duration $\epsilon$.
For $\mathcal D_\lambda^2(t)$, measure $\{P_{Q,\lambda}(s),\id_Q-P_{Q,\lambda}(s)\}$
on the test state defined above, apply $V_\lambda$ for duration $\epsilon$, then repeat
the measurement.  Taking the supremum over $s\in I_\lambda$, the
outcome-switching probability is $\epsilon^2\mathcal D_\lambda^2(t)+o(\epsilon^2)$
(Lemmas~\ref{lem:activation-detector} and~\ref{lem:paired-measurement-switching} of the S.M.).

 For \assumptionlink{B2}, write $\Gamma_{g_\lambda}:=\int_{\mathbb R}|g_\lambda(k)|^2\dd k$.
For some finite $\lambda$-independent $K_\omega$, the condition
$C_{g_\lambda}\le K_\omega T_\lambda\Gamma_{g_\lambda}$ requires the inverse-frequency moment
to remain uniformly comparable to this unweighted coupling norm.
Equivalently, $C_{g_\lambda}/\Gamma_{g_\lambda}$ is the $|g_\lambda|^2$-weighted mean of
$\omega_\lambda^{-1}$, so \assumptionlink{B2} prevents appreciable coupling weight from migrating to modes
whose periods become arbitrarily long relative to $T_\lambda$.  It is therefore an infrared
uniformity condition, not an ultraviolet cutoff: the overall coupling strength and bandwidth may
still grow.  It holds, for example, under a uniform dimensionless frequency gap (i.e.,
$T_\lambda\omega_\lambda(k)\geq\vartheta_{\min}>0$ for all relevant $k$), or when a
monotone dispersion is gapless but the coupling vanishes below a uniform mode cutoff whose
dimensionless frequency stays bounded away from zero.  It also allows genuinely gapless coupling
when the coupling amplitude vanishes sufficiently regularly near zero energy; these criteria are made precise in
Lemmas~\ref{lem:B2-frequency-gap} and \ref{lem:B2-energy-localization} of the S.M.

Retaining \assumptionlink{A1}--\assumptionlink{A2} and imposing \assumptionlink{B1}--\assumptionlink{B2}, Theorem~\ref{thm:quiet-relative-speed} of the S.M. proves
\begin{equation}
 \omega_{\mathrm{gate},\lambda}T_\lambda
 \le c_{\mathrm{quiet}}\sqrt{T_\lambda E_\lambda},
 \label{eq:quiet-bound}
\end{equation}
where $c_{\mathrm{quiet}}$ is independent of $\lambda$.  
\assumptionlink{B1} constrains the effective coupling on the system states visited
during the gate (Lemma~\ref{lem:quiet-controls-Gamma} of the S.M.).
Combining this constraint with \assumptionlink{B2} and the exact population-rate estimate
gives Eq.~\eqref{eq:quiet-bound}, with
$c_{\mathrm{quiet}}=\frac{4\pi}{\delta_0}\sqrt{C_{\mathrm{fld}}K_\omega}
(\sqrt{K_D}+\sqrt{K_J})$.
Thus with the other constants fixed, $c_{\mathrm{quiet}}$ is small if a `very quiet'
window $\mathcal W_\lambda$ exists with $K_J\approx K_D\approx 0$; see numerics for quantitative statements.

These requirements are not tailored to defeat the known quantum-frequential construction of~\cite{Woods2024}.  The
model there satisfies \assumptionlink{A0}--\assumptionlink{A2} and analogues of \assumptionlink{B1}--\assumptionlink{B2}
(since that model has different interaction terms, \assumptionlink{B1}--\assumptionlink{B2} as stated are not well-defined without adaptation to the nonlinear coupling); see Sec.~\ref{sec:nonlinear-clock-comparison} and Theorem~\ref{thm:nonlinear-clock-comparison} of the S.M.  Nevertheless, for any fixed $T_\lambda$ it reaches 
\begin{equation}
	\omega_{\mathrm{gate},\lambda}T_\lambda\propto (T_\lambda E_\lambda)^{1-\epsilon}
\end{equation}
 for any fixed
$\epsilon>0$ allowed by the construction, arbitrarily close to linear scaling and with arbitrarily high gate fidelity, without
contradicting Eqs.~\eqref{eq:uniform-bound} or \eqref{eq:quiet-bound}.  This comparison both
supports the physical content of the assumptions and isolates linearity of the bosonic
interaction as the decisive restriction.

\textit{Numerical illustration.---}
We search a family of waveguide Hamiltonians with a variable one-particle
dispersion, coupling function, and qubit gap, together with variable coherent
packet profiles (Appendix~\ref{app:ir-regular-gaussian} of the S.M.).
At each of seven energies spanning
$400\le T_\lambda E_\lambda\le25600$, the objective is the fastest
$0$--$0.99$ target-population transfer found under the same constants from
Assumptions~\assumptionlink{A1}, \assumptionlink{A2},
\assumptionlink{B1}, \assumptionlink{B2}, namely
$(C_{\mathrm{fld}},\delta_0,K_J,K_D,K_\omega)
=(1.03,0.99,0.02,0.12,0.5)$.
All candidates use one reference-time rule with $T_\lambda=1$;
the search box guarantees scale-covariant uniform lower stability
(\assumptionlink{A0}) with $K_0=11$.
We evolve the joint qubit--field state and evaluate \assumptionlink{B1}
directly on a common post-pulse quiet
window.  The coupling support guarantees \assumptionlink{B2} analytically;
energy conservation gives a common sufficient \assumptionlink{A1} bound,
and the accepted population transfer implies \assumptionlink{A2}.

The fastest feasible candidates found follow
$\omega_{\mathrm{gate},\lambda}T_\lambda
=0.741305\sqrt{T_\lambda E_\lambda}-0.703650$
with coefficient of determination $R^2=0.99937$
(Fig.~\ref{fig:irg-full-dynamics-extended}).
A log--log fit gives a slope of $0.5104$.
The optimized rates are approximately $1.68$--$1.74$ times those of the
fixed-Hamiltonian comparison, while the selected coupling profiles, qubit
gaps, and some dispersions change with energy.
Squeezing the incoming pulse mode to a minimum quadrature variance of
$36.8\%$ or $13.5\%$ of its vacuum value changes equal-energy rates
by less than $0.9\%$ (lower panel).

\textit{Conclusion.---}
Our bounds establish that, under the stated assumptions, no drive---including a coherent-state
drive---can surpass the square-root energy exponent within the linear bosonic interaction class.
Coherent Gaussian pulses attain this exponent.  The restriction also applies to finite-energy pulses from Heisenberg-limited
lasers \cite{Baker2021,Ostrowski2023PRL,Ostrowski2023PRA,Ostrowski2025}: their enhanced coherence and nonclassical photon statistics
cannot supply the quantum-frequential advantage through a fixed linear atom--field coupling.
Such lasers may nevertheless provide suitable nonclassical drives when combined with the
nonlinear coupling of Ref.~\cite{Woods2024}.
A natural next technological step is therefore not
further state engineering within present linear-control architectures, but engineering the
nonlinear quantum-frequential coupling of Ref.~\cite{Woods2024}.  Taken together, the two results
identify both a non-SQL drive state and an interaction outside the linear bosonic class as
necessary for the known quadratic advantage.  An intuitive distinction is the nonlinear
construction's ability to generate
large, localized transient drive--qubit entanglement: it is suppressed before and after gate realization.  The linear model permits entanglement---we never bound it---yet the gate-transition rate remains subject to the square-root bound under the stated conditions.
This is an interpretation, not an additional entanglement theorem; it nevertheless motivates nonlinear phase-dependent
couplers that preserve quiet pulse boundaries.

The conclusion fits a broader pattern.  Gaussian continuous-variable operations alone are
efficiently classically simulable, while universal continuous-variable computation requires a
nonlinear element \cite{Lloyd1999,Bartlett2002}; linear-optical universality similarly uses
measurement to supply an effective nonlinearity \cite{Knill2001}.  Our result is a dynamical,
energy-constrained counterpart: for autonomous gate drives, linear bosonic interactions erase the
putative squeezed-state advantage in the frequency--energy exponent.  Engineered nonlinearity is
therefore not merely an implementation detail but can lead to an important quantum advantage in gate performance, reducing the bottleneck in superconducting quantum computation discussed in the introduction.

\makeatletter
\let\maintextaddcontentsline\addcontentsline
\renewcommand{\addcontentsline}[3]{}
\makeatother

\begin{acknowledgments}
\textit{Acknowledgments.---}
AI use disclosure: Preliminary versions of the main results were developed by P.R.G. as part of his master's thesis under the supervision of M.P.W. During subsequent development of the results and preparation of this manuscript, M.P.W. used OpenAI Codex with GPT-5.6 Sol to assist in refining and extending theorem statements and proofs. M.P.W. directed this use by specifying the assumptions and target statements, and independently and rigorously verified every AI-assisted mathematical argument. This was a highly interactive process over several months. M.P.W. assumes primary responsibility for the correctness of the final theorem statements and proofs.
This work was supported by the French National Research Agency (ANR) under the France 2030 program, Grant No.~ANR-22-PETQ-0006, and the QuantERA II Programme that has received funding from the EU’s H2020 research and innovation programme under grant agreement number 101017733.
\end{acknowledgments}

\bibliography{linear_bosonic_control_prl}

\makeatletter
\let\addcontentsline\maintextaddcontentsline
\makeatother

\clearpage
\onecolumngrid
\appendix
\setcounter{secnumdepth}{3}
\include{self_contained_square_root_speed_boundsV7}

\end{document}

%% file: self_contained_square_root_speed_boundsV7.tex
\makeatletter
\begin{center}
  {\Large\bfseries Supplemental Material for\\[0.5ex]
  \textit{\@title}}
\end{center}
\makeatother
\renewcommand{\tocname}{Contents of the Supplemental Material}
\tableofcontents
\section{Model and exact transition-speed bound}
\label{sec:model}

Throughout, $\hbar=1$.  The controlled system has a finite-dimensional Hilbert space
$\mathcal H_Q$.  The drive is a bosonic continuum with Fock space $\mathcal H_F$ and
annihilation and creation operators satisfying
\begin{equation}
  [a_k,a_{k'}]=0,
  \qquad
  [a_k,a_{k'}^\dagger]=\delta(k-k').
  \label{eq:CCR}
\end{equation}
The mode label $k\in\mathbb R$ is a signed spectral coordinate, while $\omega(k)\ge0$ is the
physical mode frequency; no particular realization of the bosonic drive is assumed.

\subsection{Hamiltonians and regularity assumptions}

The uncoupled Hamiltonian is
\begin{equation}
  H_0
  :=H_Q\otimes\id_F+\id_Q\otimes H_F,
  \qquad
  H_F
  :=\int_{\mathbb R}\dd k\,\omega(k)a_k^\dagger a_k,
  \label{eq:free-Hamiltonian}
\end{equation}
where $H_Q=H_Q^\dagger$ and $\omega:\mathbb R\to[0,\infty)$ is measurable and finite almost
everywhere.  In particular, $H_F\ge0$ and
the Fock vacuum $\ket{\mathrm{vac}}$ satisfies $H_F\ket{\mathrm{vac}}=0$.

The interaction and total Hamiltonian are
\begin{align}
  V
  &:=\sigma\otimes B+\sigma^\dagger\otimes B^\dagger,
  \label{eq:interaction-V}\\
  B
  &:=\int_{\mathbb R}\dd k\,g(k)a_k,
  \label{eq:collective-B}\\
  H
  &:=H_0+V.
  \label{eq:total-Hamiltonian}
\end{align}
Here $\sigma$ is a nonzero operator on $\mathcal H_Q$.  We use the induced operator norm
\begin{equation}
  \lVert A\rVert
  :=\sup_{0\ne u\in\mathcal H_Q}
  \frac{\lVert Au\rVert_2}{\lVert u\rVert_2},
  \qquad
  \lVert u\rVert_2:=\sqrt{\langle u\mid u\rangle},
  \label{eq:operator-norm}
\end{equation}
and without loss of generality normalize
\begin{equation}
  \lVert\sigma\rVert=1.
  \label{eq:sigma-normalization}
\end{equation}
This entails no loss of generality: replacing $\sigma$ by
$\sigma/\lVert\sigma\rVert$ and $g(k)$ by $\lVert\sigma\rVert g(k)$ leaves $V$ unchanged.

We assume
\begin{equation}
  g\in L^2(\mathbb R),
  \qquad
  C_g
  :=\int_{\{k\in\mathbb R:\,g(k)\ne0\}} \dd k\,
  \frac{|g(k)|^2}{\omega(k)}
  <\infty.
  \label{eq:Cg-definition}
\end{equation}
The first condition defines the smeared creation and annihilation operators in
\eqref{eq:collective-B}; the second is the energy-weighted regularity condition used in the
speed estimate.  The coupled set $\{k:g(k)\ne0\}$ is understood modulo null sets.  Equivalently,
on the full mode domain $\mathbb R$, the quotient $|g|^2/\omega$ is assigned the value zero
where $g=0$, including when $\omega=0$, and the value $+\infty$ where $g\ne0$ but $\omega=0$.
Thus $C_g<\infty$ requires $\omega(k)>0$ almost everywhere on the coupled set, while permitting
uncoupled zero-frequency modes.  Quotients involving $g/\sqrt{\omega}$ or $g/\omega$ below are
likewise set to zero where $g=0$.  The same convention is used memberwise for
$C_{g_\lambda}$ introduced later.

The signed-coordinate representation used here is unitarily equivalent to a half-line
representation with a two-valued branch index: splitting a one-particle wavefunction into
$f_+(\kappa):=f(\kappa)$ and $f_-(\kappa):=f(-\kappa)$ identifies
$L^2(\mathbb R)$ with $L^2((0,\infty);\mathbb C^2)$.  Thus using the full line changes only
how opposite propagation branches are labelled; it does not add a dynamical assumption.

We also assume that the total Hamiltonian is physically well defined and stable: the coupled
system has a finite lowest possible energy.  This lowest energy need not be an eigenvalue, so a
normalizable ground state need not exist.  Mathematically, we assume that $H$, defined on the
natural common domain of $H_0$ and $V$, is self-adjoint and bounded from below.  The initial
joint state $\rho_0$ is allowed to be mixed and correlated.  It
is assumed to have finite expectations for every energy quantity displayed below.  All
commutators and traces used in the proof are assumed to be well defined on a common
finite-energy core.  Equivalently, the estimates can first be proved on that core and then
extended by finite-energy approximation; see, for example, the standard domain and core
treatment of Pauli--Fierz Hamiltonians in Ref.~\cite{Hiroshima2002}.

\subsection{Exact interaction-picture dynamics}

Define
\begin{equation}
  U^{\picsup{I}}(t):=e^{iH_0t}e^{-iHt},
  \qquad
  \rho^{\picsup{I}}(t):=U^{\picsup{I}}(t)\rho_0U^{\picsup{I}\dagger}(t).
  \label{eq:interaction-picture-state}
\end{equation}
The canonical commutation relations imply
\begin{equation}
  [H_F,a_k]=-\omega(k)a_k.
  \label{eq:HF-ak-commutator}
\end{equation}
Consequently,
\begin{align}
  \frac{\dd}{\dd t}
  \left(e^{iH_Ft}a_ke^{-iH_Ft}\right)
  &=i e^{iH_Ft}[H_F,a_k]e^{-iH_Ft}
  \notag\\
  &=-i\omega(k)e^{iH_Ft}a_ke^{-iH_Ft}.
  \label{eq:free-ak-differential-equation}
\end{align}
The initial value at $t=0$ is $a_k$, so the unique solution is
\begin{equation}
  e^{iH_Ft}a_ke^{-iH_Ft}=e^{-i\omega(k)t}a_k.
  \label{eq:free-ak-evolution}
\end{equation}
It follows that the exact interaction-picture Hamiltonian is
\begin{align}
  H_{\mathrm{int}}(t)
  &:=e^{iH_0t}Ve^{-iH_0t}
  \notag\\
  &=\sigma(t)\otimes B(t)
  +\sigma^\dagger(t)\otimes B^\dagger(t),
  \label{eq:interaction-picture-Hamiltonian}\\
  \sigma(t)
  &:=e^{iH_Qt}\sigma e^{-iH_Qt},
  \label{eq:sigma-t}\\
  B(t)
  &:=\int_{\mathbb R}\dd k\,g(k)e^{-i\omega(k)t}a_k.
  \label{eq:B-t}
\end{align}
Differentiating $U^{\picsup{I}}(t)$ gives
\begin{equation}
  i\frac{\dd}{\dd t}U^{\picsup{I}}(t)=H_{\mathrm{int}}(t)U^{\picsup{I}}(t),
  \label{eq:UI-equation}
\end{equation}
and hence the exact density operator obeys
\begin{equation}
  \frac{\dd}{\dd t}\rho^{\picsup{I}}(t)
  =-i[H_{\mathrm{int}}(t),\rho^{\picsup{I}}(t)].
  \label{eq:von-Neumann-equation}
\end{equation}

\subsection{Energy resource and gate interval}

Let
\begin{equation}
  E_{\mathrm{gs}}:=\inf\operatorname{spec}(H),
  \qquad
  E:=\Tr(\rho_0H)-E_{\mathrm{gs}}>0.
  \label{eq:energy-resource}
\end{equation}
Thus $E$ is the initial mean energy above the bottom of the spectrum of the actual total
Hamiltonian.  A normalizable ground-state vector is not required.

The free-field energy at time $t$ is
\begin{equation}
  E_F(t)
  :=\Tr\!\left[\rho^{\picsup{I}}(t)(\id_Q\otimes H_F)\right].
  \label{eq:field-energy-t}
\end{equation}
Because $[H_0,\id_Q\otimes H_F]=0$, this quantity has the same value in the interaction and
Schr\"odinger pictures.  We assume that, during the gate interval defined below, there is a
constant $C_{\mathrm{fld}}<\infty$ such that
\begin{equation}
  \sup_{t\in I}E_F(t)\le C_{\mathrm{fld}}E.
  \label{eq:field-energy-control}
\end{equation}

\begin{definition}[Gate interval and trace distance]
\label{def:gate-interval}
Let
\begin{equation}
  \rho_Q^{\picsup{I}}(t):=\Tr_F\rho^{\picsup{I}}(t)
  \label{eq:reduced-state}
\end{equation}
be the reduced state of the controlled system in the interaction picture.
An accepted gate interval is an interval
\begin{equation}
  I=[t_{\mathrm{in}},t_{\mathrm{out}}],
  \qquad
  \Delta t:=t_{\mathrm{out}}-t_{\mathrm{in}}>0,
  \label{eq:gate-interval}
\end{equation}
with system states
\begin{equation}
  \rho_{\mathrm{in}}:=\rho_Q^{\picsup{I}}(t_{\mathrm{in}}),
  \qquad
  \rho_{\mathrm{out}}:=\rho_Q^{\picsup{I}}(t_{\mathrm{out}}).
  \label{eq:endpoint-states}
\end{equation}
The states $\rho_{\mathrm{out}}$ and $\rho_{\mathrm{in}}$ may be mixed, and the joint system--field state need not factorize.
For density operators $\rho,\eta$, write
$D(\rho,\eta):=\tfrac12\lVert\rho-\eta\rVert_1$, where
$\lVert X\rVert_1:=\Tr\sqrt{X^\dagger X}$ is the trace norm.  We require
\begin{equation}
  \delta:=D(\rho_{\mathrm{out}},\rho_{\mathrm{in}})>0
  \label{eq:endpoint-distance}
\end{equation}
and absolute continuity on $I$ of every fixed system-projector expectation
$t\mapsto\Tr[P\rho_Q^{\picsup{I}}(t)]$.  The associated gate frequency is
\begin{equation}
  \omega_{\mathrm{gate}}:=\frac{2\pi}{\Delta t}.
  \label{eq:gate-frequency}
\end{equation}
\end{definition}

The separation of $\rho_{\mathrm{out}}$ and $\rho_{\mathrm{in}}$ measures net system-state change; additional operational requirements
are needed to certify the intended gate on arbitrary inputs.
However, to derive upper bounds on gate speed, it suffices to impose necessary conditions
for the intended gate; these need not be sufficient.  Further conditions are introduced
below as needed.

\subsection{Fixed-Hamiltonian square-root bound}

The following estimate is a special case of the norm formula for inner derivations
\cite{Stampfli1970}.  We reproduce its elementary projection proof for completeness.

\begin{lemma}[Commutator with a projection]
\label{lem:projection-commutator}
If $P$ is an orthogonal projection and $S$ is a bounded operator on $\mathcal H_Q$, then
\begin{equation}
  \lVert[P,S]\rVert\le\lVert S\rVert.
  \label{eq:projection-commutator-bound}
\end{equation}
\end{lemma}

\begin{proof}
Put $Q:=\id_Q-P$.  Since $P^2=P$,
\begin{align}
  [P,S]
  &=PS-SP
  \notag\\
  &=PS(P+Q)-(P+Q)SP
  \notag\\
  &=PSQ-QSP.
  \label{eq:projection-block-form}
\end{align}
For any $u\in\mathcal H_Q$, the vectors $PSQ\,u$ and $QSP\,u$ lie in the orthogonal subspaces
$P\mathcal H_Q$ and $Q\mathcal H_Q$.  Therefore
\begin{equation}
  \lVert[P,S]u\rVert_2^2
  =\lVert PSQ\,u\rVert_2^2+\lVert QSP\,u\rVert_2^2.
  \label{eq:projection-Pythagoras-first}
\end{equation}
Orthogonal projections are contractions and the induced operator norm is submultiplicative, so
\begin{align}
  \lVert PSQ\,u\rVert_2
  &\le\lVert P\rVert\lVert S\rVert\lVert Qu\rVert_2
  \le\lVert S\rVert\lVert Qu\rVert_2,
  \notag\\
  \lVert QSP\,u\rVert_2
  &\le\lVert Q\rVert\lVert S\rVert\lVert Pu\rVert_2
  \le\lVert S\rVert\lVert Pu\rVert_2.
  \label{eq:projection-contractions}
\end{align}
Because $Pu\perp Qu$ and $u=Pu+Qu$,
\begin{equation}
  \lVert Pu\rVert_2^2+\lVert Qu\rVert_2^2=\lVert u\rVert_2^2.
  \label{eq:projection-Pythagoras-second}
\end{equation}
Substituting \eqref{eq:projection-contractions} and
\eqref{eq:projection-Pythagoras-second} into
\eqref{eq:projection-Pythagoras-first} gives
\begin{equation}
  \lVert[P,S]u\rVert_2^2
  \le\lVert S\rVert^2\lVert u\rVert_2^2.
  \label{eq:projection-vector-bound}
\end{equation}
Taking the supremum over nonzero $u$ proves
\eqref{eq:projection-commutator-bound}.
\end{proof}

\begin{lemma}[Energy control of the collective annihilation operator]
\label{lem:B-energy}
For every joint density operator $\rho$ with finite free-field energy,
\begin{equation}
  \Tr\!\left[\rho(\id_Q\otimes B^\dagger(t)B(t))\right]
  \le
  C_g\Tr\!\left[\rho(\id_Q\otimes H_F)\right].
  \label{eq:B-energy-state-bound}
\end{equation}
\end{lemma}

\begin{proof}
Let $\ket{\Psi}$ be a joint vector in the form domain of $\id_Q\otimes H_F$, i.e.,
$\langle\Psi|\id_Q\otimes H_F|\Psi\rangle<\infty$.  Write
\begin{equation}
  B(t)\ket{\Psi}
  =\int_{\mathbb R}\dd k\,
  \frac{g(k)}{\sqrt{\omega(k)}}
  \left(\sqrt{\omega(k)}e^{-i\omega(k)t}a_k\ket{\Psi}\right).
  \label{eq:B-weighted-vector-integral}
\end{equation}
To see explicitly how Cauchy--Schwarz applies to this vector-valued integral, let
$\ket{\Phi}$ be any unit vector.  Then
\begin{align}
  \left|
    \left\langle\Phi\middle|B(t)\middle|\Psi\right\rangle
  \right|
  &\le
  \int_{\mathbb R}\dd k\,
  \frac{|g(k)|}{\sqrt{\omega(k)}}
  \sqrt{\omega(k)}
  \left|
    \left\langle\Phi\middle|a_k\middle|\Psi\right\rangle
  \right|
  \notag\\
  &\le
  \left(
    \int_{\mathbb R}\dd k\,
    \frac{|g(k)|^2}{\omega(k)}
  \right)^{1/2}
  \left(
    \int_{\mathbb R}\dd k\,
    \omega(k)
    \left|
      \left\langle\Phi\middle|a_k\middle|\Psi\right\rangle
    \right|^2
  \right)^{1/2}
  \notag\\
  &\le
  \sqrt{C_g}
  \left(
    \int_{\mathbb R}\dd k\,
    \omega(k)\lVert a_k\ket{\Psi}\rVert_2^2
  \right)^{1/2}.
  \label{eq:B-duality-Cauchy-Schwarz}
\end{align}
The last inequality uses
$|\langle\Phi|a_k|\Psi\rangle|\le\lVert a_k\ket{\Psi}\rVert_2$.
Taking the supremum over unit $\ket{\Phi}$ gives
\begin{align}
  \lVert B(t)\ket{\Psi}\rVert_2^2
  &\le
  C_g\int_{\mathbb R}\dd k\,
  \omega(k)\lVert a_k\ket{\Psi}\rVert_2^2
  \notag\\
  &=C_g
  \left\langle\Psi\middle|
  \id_Q\otimes H_F
  \middle|\Psi\right\rangle.
  \label{eq:B-energy-vector-bound}
\end{align}
Now take a spectral decomposition
$\rho=\sum_jp_j\ket{\Psi_j}\!\!\mkern1mu\bra{\Psi_j}$, with $p_j\ge0$ and $\sum_jp_j=1$.
Applying \eqref{eq:B-energy-vector-bound} to each vector and summing gives
\begin{align}
  \Tr[\rho(\id_Q\otimes B^\dagger(t)B(t))]
  &=\sum_jp_j\lVert B(t)\ket{\Psi_j}\rVert_2^2
  \notag\\
  &\le C_g\sum_jp_j
  \left\langle\Psi_j\middle|
  \id_Q\otimes H_F
  \middle|\Psi_j\right\rangle,
  \label{eq:B-energy-mixed-step}
\end{align}
which is \eqref{eq:B-energy-state-bound}.
\end{proof}

\setcounter{theorem}{7}
\begin{theorem}[Exact full-dynamics speed bound]
\label{thm:exact-speed}
Under the assumptions of Section~\ref{sec:model}, every accepted gate interval in
Definition~\ref{def:gate-interval} satisfying the field-energy condition
\eqref{eq:field-energy-control} obeys
\begin{equation}
  \omega_{\mathrm{gate}}
  \le
  \frac{4\pi}{D(\rho_{\mathrm{out}},\rho_{\mathrm{in}})}\sqrt{C_gC_{\mathrm{fld}}E}.
  \label{eq:theorem8-frequency}
\end{equation}
\end{theorem}

\begin{proof}
Let $P_+$ be the spectral projection
onto the strictly positive eigenspace of $\rho_{\mathrm{out}}-\rho_{\mathrm{in}}$, and put
\begin{equation}
  p_+(t):=\Tr[P_+\rho_Q^{\picsup{I}}(t)].
  \label{eq:endpoint-witness-population}
\end{equation}
Because $\rho_{\mathrm{out}}-\rho_{\mathrm{in}}$ is Hermitian and traceless, the sum of its positive
eigenvalues is half its trace norm.  Consequently,
\begin{align}
  \delta
  &=\Tr[P_+(\rho_{\mathrm{out}}-\rho_{\mathrm{in}})]
   =p_+(t_{\mathrm{out}})-p_+(t_{\mathrm{in}})
   \notag\\
  &=\max_{P=P^\dagger=P^2}\Tr[P(\rho_{\mathrm{out}}-\rho_{\mathrm{in}})].
  \label{eq:endpoint-witness}
\end{align}
The projection $P_+$ may have arbitrary rank and is held fixed throughout $I$.
Since
$p_+(t)=\Tr[\rho^{\picsup{I}}(t)(P_+\otimes\id_F)]$,
the von Neumann equation and cyclicity of the trace give
\begin{align}
  \frac{\dd}{\dd t}p_+(t)
  &=-i\Tr\!\left[
    (P_+\otimes\id_F)
    [H_{\mathrm{int}}(t),\rho^{\picsup{I}}(t)]
  \right]
  \notag\\
  &=-i\Tr\!\left[
    \rho^{\picsup{I}}(t)
    [P_+\otimes\id_F,H_{\mathrm{int}}(t)]
  \right].
  \label{eq:theorem8-population-derivative-first}
\end{align}
For brevity, define
\begin{equation}
  A(t):=[P_+,\sigma(t)],
  \qquad
  X(t):=A(t)\otimes B(t).
  \label{eq:theorem8-A-X}
\end{equation}
Because $P_+=P_+^\dagger$,
\begin{align}
  [P_+,\sigma^\dagger(t)]
  &=P_+\sigma^\dagger(t)
    -\sigma^\dagger(t)P_+
  \notag\\
  &=-\left(
    P_+\sigma(t)-\sigma(t)P_+
  \right)^\dagger
  =-A^\dagger(t).
  \label{eq:theorem8-adjoint-commutator}
\end{align}
Substitution of \eqref{eq:interaction-picture-Hamiltonian} therefore yields
\begin{align}
  [P_+\otimes\id_F,H_{\mathrm{int}}(t)]
  &=A(t)\otimes B(t)-A^\dagger(t)\otimes B^\dagger(t)
  \notag\\
  &=X(t)-X^\dagger(t).
  \label{eq:theorem8-X-minus-adjoint}
\end{align}
Set
\begin{equation}
  z(t):=\Tr[\rho^{\picsup{I}}(t)X(t)].
  \label{eq:theorem8-z}
\end{equation}
Since $\rho^{\picsup{I}}(t)=\rho^{\picsup{I}\dagger}(t)$,
$\Tr[\rho^{\picsup{I}}(t)X^\dagger(t)]=z(t)^*$.  Equations
\eqref{eq:theorem8-population-derivative-first} and
\eqref{eq:theorem8-X-minus-adjoint} consequently give
\begin{equation}
  \frac{\dd}{\dd t}p_+(t)
  =-i\bigl(z(t)-z(t)^*\bigr)
  =2\operatorname{Im}z(t),
  \label{eq:theorem8-population-derivative-z}
\end{equation}
and hence
\begin{equation}
  \left|\frac{\dd}{\dd t}p_+(t)\right|
  \le2|z(t)|.
  \label{eq:theorem8-derivative-by-z}
\end{equation}

We next bound $z(t)$.  The Hilbert--Schmidt Cauchy--Schwarz inequality, applied to
$\rho^{\picsup{I}}(t)^{1/2}$ and $X(t)\rho^{\picsup{I}}(t)^{1/2}$, gives
\begin{align}
  |z(t)|^2
  &=\left|
    \Tr\!\left[
      \rho^{\picsup{I}}(t)^{1/2}X(t)\rho^{\picsup{I}}(t)^{1/2}
    \right]
  \right|^2
  \notag\\
  &\le
  \Tr[\rho^{\picsup{I}}(t)]
  \Tr[\rho^{\picsup{I}}(t)X^\dagger(t)X(t)]
  \notag\\
  &=\Tr[\rho^{\picsup{I}}(t)X^\dagger(t)X(t)],
  \label{eq:theorem8-state-CS}
\end{align}
where $\Tr\rho^{\picsup{I}}(t)=1$ was used in the last line.  Because the two factors of $X(t)$ act on
different tensor components,
\begin{equation}
  X^\dagger(t)X(t)
  =A^\dagger(t)A(t)\otimes B^\dagger(t)B(t).
  \label{eq:theorem8-XdagX}
\end{equation}
The positive-operator inequality
$A^\dagger(t)A(t)\le\lVert A(t)\rVert^2\id_Q$
\cite[Chapter~5, Exercise~45(a)]{MacCluer2009} can be tensored with the positive operator
$B^\dagger(t)B(t)\ge0$ without reversing the order.  Hence
\[
  A^\dagger(t)A(t)\otimes B^\dagger(t)B(t)
  \le
  \lVert A(t)\rVert^2\id_Q\otimes B^\dagger(t)B(t).
\]

Taking the trace in a basis that diagonalises $\rho^{\picsup{I}}(t)$ and using
\eqref{eq:theorem8-state-CS} and \eqref{eq:theorem8-XdagX} gives
\begin{equation}
  |z(t)|^2
  \le\lVert A(t)\rVert^2
  \Tr\!\left[
    \rho^{\picsup{I}}(t)(\id_Q\otimes B^\dagger(t)B(t))
  \right].
  \label{eq:theorem8-z-second-moment}
\end{equation}
Lemma~\ref{lem:projection-commutator}, unitary invariance of the operator norm, and
\eqref{eq:sigma-normalization} give
\begin{equation}
  \lVert A(t)\rVert
  \le\lVert\sigma(t)\rVert
  =\lVert\sigma\rVert
  =1.
  \label{eq:theorem8-A-norm}
\end{equation}
Lemma~\ref{lem:B-energy} and the definition \eqref{eq:field-energy-t} then yield
\begin{equation}
  |z(t)|^2\le C_gE_F(t).
  \label{eq:theorem8-z-energy}
\end{equation}
Combining \eqref{eq:theorem8-derivative-by-z},
\eqref{eq:theorem8-z-energy}, and the uniform field-energy estimate
\eqref{eq:field-energy-control}, we obtain
\begin{equation}
  \left|\frac{\dd}{\dd t}p_+(t)\right|
  \le2\sqrt{C_gC_{\mathrm{fld}}E}
  \qquad(t\in I).
  \label{eq:theorem8-uniform-rate}
\end{equation}

Finally, absolute continuity and the fundamental theorem of calculus imply
\begin{align}
  \delta
  &=
  p_+(t_{\mathrm{out}})
  -p_+(t_{\mathrm{in}})
  \notag\\
  &=\int_{t_{\mathrm{in}}}^{t_{\mathrm{out}}}
  \frac{\dd}{\dd t}p_+(t)\dd t
  \notag\\
  &\le\int_{t_{\mathrm{in}}}^{t_{\mathrm{out}}}
  \left|\frac{\dd}{\dd t}p_+(t)\right|\dd t
  \notag\\
  &\le2\Delta t\sqrt{C_gC_{\mathrm{fld}}E}.
  \label{eq:theorem8-integration}
\end{align}
Rearranging and substitution of
$\omega_{\mathrm{gate}}=2\pi/\Delta t$ proves
\eqref{eq:theorem8-frequency}.
\end{proof}

\section{Uniform families and dimensionless square-root bounds}
\label{sec:families}

We first formulate the uniform family and its reference time, then prove the square-root speed bound under
Assumptions~\assumptionlink{A0}--\assumptionlink{A2}.  We then replace the uniform
lower-stability condition \assumptionlink{A0} by the quiet-interaction conditions
\assumptionlink{B1}--\assumptionlink{B2}, obtaining a corresponding bound without that
dimensionless lower-stability requirement.

We now consider a family of control protocols indexed by $\lambda$.  Every quantity that may
vary with the family is assigned a subscript $\lambda$.  Thus
\begin{align}
  H_{F,\lambda}
  &:=\int_{\mathbb R}\dd k\,
  \omega_\lambda(k)a_k^\dagger a_k,
  \label{eq:family-HF}\\
  B_\lambda
  &:=\int_{\mathbb R}\dd k\,g_\lambda(k)a_k,
  \label{eq:family-B}\\
  V_\lambda
  &:=\sigma_\lambda\otimes B_\lambda
  +\sigma_\lambda^\dagger\otimes B_\lambda^\dagger,
  \label{eq:family-V}
\end{align}
where every $H_{Q,\lambda}=H_{Q,\lambda}^\dagger$ acts on the same finite-dimensional
controlled-system Hilbert space $\mathcal H_Q$.  Without loss of generality, we use the convention
$\lVert\sigma_\lambda\rVert=1$ for every $\lambda$, as in
\eqref{eq:sigma-normalization}.  We further define
\begin{align}
  H_\lambda
  &:=H_{Q,\lambda}\otimes\id_F+\id_Q\otimes H_{F,\lambda}+V_\lambda,
  \label{eq:family-H}\\
  C_{g_\lambda}
  &:=\int_{\{k\in\mathbb R:\,g_\lambda(k)\ne0\}}\dd k\,
  \frac{|g_\lambda(k)|^2}{\omega_\lambda(k)},
  \label{eq:family-Cg}\\
  E_{\mathrm{gs},\lambda}
  &:=\inf\operatorname{spec}(H_\lambda),
  \qquad
  E_\lambda
  :=\Tr(\rho_{0,\lambda}H_\lambda)-E_{\mathrm{gs},\lambda}.
  \label{eq:family-energy}
\end{align}
The definition of $C_{g_\lambda}$ uses the quotient convention following
\eqref{eq:Cg-definition}; in particular, its finiteness requires
$\omega_\lambda(k)>0$ almost everywhere on the coupled set
$\{k:g_\lambda(k)\ne0\}$.

For each member, define the interaction-picture joint state by
\begin{equation}
  \rho_\lambda^{\picsup{I}}(t)
  :=U_\lambda^{\picsup{I}}(t)\rho_{0,\lambda}U_\lambda^{\picsup{I}\dagger}(t),
  \qquad
  U_\lambda^{\picsup{I}}(t)
  :=e^{iH_{0,\lambda}t}e^{-iH_\lambda t},
  \label{eq:family-interaction-picture-state}
\end{equation}
where $H_{0,\lambda}:=H_{Q,\lambda}\otimes\id_F+\id_Q\otimes H_{F,\lambda}$; this is the
$\lambda$-indexed version of \eqref{eq:interaction-picture-state}.  In particular,
$\sigma_\lambda(t):=e^{iH_{Q,\lambda}t}\sigma_\lambda e^{-iH_{Q,\lambda}t}$.

Define the reduced state
$\rho_{Q,\lambda}^{\picsup{I}}(t):=\Tr_F\rho_\lambda^{\picsup{I}}(t)$.
The accepted gate interval is
\begin{equation}
  I_\lambda
  :=[t_{\mathrm{in},\lambda},t_{\mathrm{out},\lambda}],
  \qquad
  \Delta t_\lambda
  :=t_{\mathrm{out},\lambda}-t_{\mathrm{in},\lambda}>0.
  \label{eq:family-gate-interval}
\end{equation}
The system states $\rho_{\mathrm{in},\lambda}$ and $\rho_{\mathrm{out},\lambda}$ and their trace distance are
\begin{align}
  \rho_{\mathrm{in},\lambda}
  &:=\rho_{Q,\lambda}^{\picsup{I}}(t_{\mathrm{in},\lambda}),
  \qquad
  \rho_{\mathrm{out},\lambda}
  :=\rho_{Q,\lambda}^{\picsup{I}}(t_{\mathrm{out},\lambda}),
  \label{eq:family-endpoint-states}\\
  \delta_\lambda
  &:=D(\rho_{\mathrm{out},\lambda},\rho_{\mathrm{in},\lambda})>0.
  \label{eq:family-endpoint-distance}
\end{align}
The corresponding gate frequency is
\begin{equation}
  \omega_{\mathrm{gate},\lambda}:=\frac{2\pi}{\Delta t_\lambda}.
  \label{eq:family-gate-frequency}
\end{equation}
The interval, states $\rho_{\mathrm{in},\lambda}$ and $\rho_{\mathrm{out},\lambda}$, trace distance, and gate frequency may all depend on $\lambda$.

\subsection{Reference time and scale covariance}
\label{sec:reference-time}

The bounds on $\omega_{\mathrm{gate},\lambda}T_\lambda$ established in
Theorems~\ref{thm:relative-speed} and~\ref{thm:quiet-relative-speed} hold for any
assignment of reference times $T_\lambda>0$ satisfying the relevant assumptions.
The values of $T_\lambda$ are constrained through these assumptions, introduced below:
\assumptionlink{A0}, \assumptionlink{A1}, \assumptionlink{A2},
\assumptionlink{B1}, and \assumptionlink{B2}.
Because these assumptions involve $T_\lambda$, its choice partly determines
which families satisfy the hypotheses and the resulting bound on
$\omega_{\mathrm{gate},\lambda}T_\lambda$.
The near-linear scaling result for nonlinear interactions
(Theorem~\ref{thm:nonlinear-clock-comparison}) uses the harmonic oscillator's period as
its reference time.

A scale-covariant assignment additionally obeys
\begin{equation}
  T_\lambda
  :=\tau(H_{\mathrm{ref},\lambda},\rho_{\mathrm{ref},\lambda}),
  \qquad
  \tau(\alpha H+c\id,\rho)
  =\alpha^{-1}\tau(H,\rho)
  \quad(\alpha>0,\ c\in\mathbb R),
  \label{eq:reference-time-rule}
\end{equation}
where $\tau$ is one fixed rule for the whole family.
Under $H_\lambda\mapsto\alpha_\lambda H_\lambda+c_\lambda\id_{QF}$, assign the scalar
shift to $H_{Q,\lambda}$ and divide all gate and quiet times by $\alpha_\lambda$.
Under Eq.~\eqref{eq:reference-time-rule}, $T_\lambda\mapsto T_\lambda/\alpha_\lambda$, and
\assumptionlink{A0}--\assumptionlink{A2} and \assumptionlink{B1}--\assumptionlink{B2}
remain invariant with the same constants, as do
$\omega_{\mathrm{gate},\lambda}T_\lambda$ and $T_\lambda E_\lambda$.
Thus uniformly speeding up the entire dynamics does not improve gate performance
relative to the reference dynamics.  An absolute bound in a common laboratory time unit
follows when the $T_\lambda$ are uniformly comparable to that unit.

\medskip
\noindent The following examples give choices of the reference-time rule $\tau$ in
Eq.~\eqref{eq:reference-time-rule}, based on either the full dynamics or the free field.

\medskip
\noindent\textit{Bures-angle example.}
Let
\begin{equation}
  \mathcal L_{\mathrm B}(\rho,\eta)
  :=\arccos\Tr\!\sqrt{\sqrt\rho\,\eta\sqrt\rho}
  \label{eq:Bures-angle}
\end{equation}
be the Bures angle.  Whenever the orbit reaches an orthogonal state, its full
orthogonalization time
\begin{equation}
  T_\lambda^{\mathrm{orth}}
  :=\inf\!\left\{t>0:
    \mathcal L_{\mathrm B}\!\left(
      \rho_{0,\lambda},
      e^{-iH_\lambda t}\rho_{0,\lambda}e^{iH_\lambda t}
    \right)=\frac{\pi}{2}
  \right\}
  \label{eq:Bures-orthogonalization-time}
\end{equation}
provides such a rule.  Here $\pi/2$ is the Bures angle between orthogonal states; the associated
angular reference frequency is written as $2\pi/T_\lambda^{\mathrm{orth}}$, in the same
full-cycle convention as \eqref{eq:family-gate-frequency}.

\medskip
\noindent\textit{Quantum Fisher information example.}
An infinitesimal alternative is
\begin{equation}
  T_\lambda^{\mathrm{QFI}}
  :=\frac{2}{\sqrt{\mathcal F_Q(\rho_{0,\lambda},H_\lambda)}},
  \label{eq:QFI-reference-time}
\end{equation}
when the quantum Fisher information is finite and nonzero.  For a pure initial state this is
$1/\Delta_{\rho_{0,\lambda}}H_\lambda$.  Both examples obey
\eqref{eq:reference-time-rule}; the first compares the gate frequency with the full
orthogonalization time of the dynamics, whereas the second uses its initial statistical speed.

\medskip
\noindent\textit{Propagation-time example.}
For a nondispersive waveguide with $\omega_\lambda(k)=v_\lambda|k|$ and $v_\lambda>0$,
fix a reference length $\ell>0$ independently of the gate optimization and choose the
background propagation time
\begin{equation}
  T_\lambda^{\mathrm{prop}}:=\frac{\ell}{v_\lambda}.
  \label{eq:propagation-reference-time}
\end{equation}
This is a reference rule based on the free field rather than the full interacting Hamiltonian.
Under an overall rescaling of the Hamiltonian by $\alpha_\lambda>0$, the propagation speed
becomes $\alpha_\lambda v_\lambda$, so
$T_\lambda^{\mathrm{prop}}\mapsto T_\lambda^{\mathrm{prop}}/\alpha_\lambda$.
Choosing $T_\lambda=T_\lambda^{\mathrm{prop}}$ makes
$\omega_{\mathrm{gate},\lambda}T_\lambda$ the gate frequency measured in units of pulse
propagation across the fixed length $\ell$.

The same reference has a classical interpretation.  For an ideal lossless transmission
line, the telegrapher equations give
$v_\lambda=1/\sqrt{L'_\lambda C'_\lambda}$, where $L'_\lambda$ and $C'_\lambda$ are the
inductance and capacitance per unit length~\cite{FitzpatrickTransmissionLines}.
The propagation time across $\ell$ is independent of the pulse amplitude.  Quantizing the
linear field preserves this free propagation law, providing the same background reference
for classical and quantum microwave pulses in superconducting-qubit control~\cite{Blais2021}.
The resulting bound compares the local gate rate with this propagation reference; it does
not identify the gate duration with the pulse delivery time.

\subsection{\texorpdfstring{Assumptions~\assumptionlink{A0}--\assumptionlink{A2}
and their interpretation}%
{Assumptions A0--A2 and their interpretation}}

The following are the additional uniform assumptions needed for the square-root bound.

\begin{enumerate}
  \item[\textbf{A0.}] \phantomsection\label{ass:A0}\textbf{Scale-covariant uniform lower stability.}
  The modes $\{a_k\}_{k\in\mathbb R}$ are independent of $\lambda$ and act on a common Fock space,
  so the vacuum $\ket{\mathrm{vac}}$, defined by $a_k\ket{\mathrm{vac}}=0$ for every $k$, is
  also independent of $\lambda$.
  The uncoupled bosonic Hamiltonians use the common convention
  \begin{equation}
    H_{F,\lambda}\ket{\mathrm{vac}}=0,
    \qquad
    \inf\operatorname{spec}(H_{F,\lambda})=0,
    \label{eq:common-free-energy-zero}
  \end{equation}
  In the member-dependent reference units, the separation between the system spectral ceiling
  and the interacting ground energy is uniformly bounded: there exists a dimensionless
  $K_0<\infty$ such that
  \begin{equation}
    T_\lambda\bigl[\lambda_{\max}(H_{Q,\lambda})-E_{\mathrm{gs},\lambda}\bigr]\le K_0
    \qquad\text{for every }\lambda.
    \label{eq:A0-reference-time-bound}
  \end{equation}

  \item[\textbf{A1.}] \phantomsection\label{ass:A1}\textbf{Uniform field-energy control.}
  Every member satisfies
  \begin{equation}
    \sup_{t\in I_\lambda}
    \Tr\!\left[
      \rho_\lambda^{\picsup{I}}(t)(\id_Q\otimes H_{F,\lambda})
    \right]
    \le C_{\mathrm{fld}}E_\lambda
    \label{eq:family-field-energy-control}
  \end{equation}
  with the same finite constant $C_{\mathrm{fld}}$.
  The left-hand side is the supremum of the drive-field energy during the gate interval.
\end{enumerate}

The following preparation condition is sufficient for \assumptionlink{A1}; it is not an
additional requirement when \assumptionlink{A1} is established by other means.
It is physically natural for conventional pulse-driven gates: the drive can be prepared
with vanishing mean interaction energy before reaching the system, and faster gates are
typically obtained by increasing the pulse energy relative to the local gate energy scale.
The condition below merely excludes the opposite limit, in which the excitation energy
becomes arbitrarily small relative to that scale.

\begin{lemma}[Physical preparation condition for field-energy control]
\label{lem:preparation-A1}
Under the model assumptions of Section~\ref{sec:model}, suppose
$\mathcal H_Q=\mathcal H_G\otimes\mathcal H_{\mathrm{sp}}$, where $G$ is the subsystem
directly coupled to the drive, and write
\begin{equation}
  \sigma_\lambda=\sigma_{G,\lambda}\otimes\id_{\mathrm{sp}},\qquad
  H_{Q,\lambda}=\id_G\otimes H_{\mathrm{sp},\lambda}+H_{\mathrm{loc},\lambda}.
  \label{eq:preparation-local-decomposition}
\end{equation}
Here the self-adjoint operators $H_{\mathrm{sp},\lambda}$ and $H_{\mathrm{loc},\lambda}$
contain, respectively, the Hamiltonian terms acting exclusively on the remaining system
and the driven subsystem's
Hamiltonian and its interactions with that system.  Define the local gate energy scale
$W_{\mathrm{loc},\lambda}:=\lambda_{\max}(H_{\mathrm{loc},\lambda})
-\lambda_{\min}(H_{\mathrm{loc},\lambda})$.
Use the energy $E_\lambda=\Tr(\rho_{0,\lambda}H_\lambda)-E_{\mathrm{gs},\lambda}$
above the interacting ground energy, as in \eqref{eq:family-energy}.
If $\Tr(\rho_{0,\lambda}V_\lambda)=0$ and
$E_\lambda\ge\eta W_{\mathrm{loc},\lambda}$ for a $\lambda$-independent $\eta>0$,
then \assumptionlink{A1} holds for all times with
\begin{equation}
  C_{\mathrm{fld}}=\bigl(2\sqrt{M}+\sqrt{1+4M}\bigr)^2,
  \qquad M:=1+\eta^{-1}.
  \label{eq:preparation-field-constant}
\end{equation}
(Neither \assumptionlink{A0} nor \assumptionlink{A2} is needed.)
\end{lemma}

\begin{proof}
This proof uses the trial-state construction of Lemma~\ref{lem:memberwise-variational-bound}
from the next subsection.  We prove that lemma there to avoid interrupting the flow of
the text here.
Put $\Delta_{\mathrm{gs},\lambda}
:=\lambda_{\min}(H_{Q,\lambda})-E_{\mathrm{gs},\lambda}\ge0$.
Since the initial interaction expectation vanishes, \eqref{eq:family-energy} gives
\[
  E_\lambda
  =\Tr\!\left[\rho_{0,\lambda}
    (H_{Q,\lambda}\otimes\id_F+\id_Q\otimes H_{F,\lambda})\right]
    -E_{\mathrm{gs},\lambda}
  \ge\Delta_{\mathrm{gs},\lambda},
\]
using $H_{Q,\lambda}\ge\lambda_{\min}(H_{Q,\lambda})\id_Q$ and $H_{F,\lambda}\ge0$.
Choose a normalized $\ket{\psi_{G,\lambda}}$ with
$a_\lambda:=|\bra{\psi_{G,\lambda}}\sigma_{G,\lambda}\ket{\psi_{G,\lambda}}|\ge1/2$;
this is possible by \eqref{eq:numerical-radius-operator-norm} and
$\lVert\sigma_{G,\lambda}\rVert=\lVert\sigma_\lambda\rVert=1$.
Let $\ket{\chi_{\mathrm{sp},\lambda}}$ be a normalized ground vector of
$H_{\mathrm{sp},\lambda}$ and put
$\ket{\psi_\lambda}:=\ket{\psi_{G,\lambda}}\otimes\ket{\chi_{\mathrm{sp},\lambda}}$.
Equation~\eqref{eq:preparation-local-decomposition} gives
\begin{align*}
  \bra{\psi_\lambda}H_{Q,\lambda}\ket{\psi_\lambda}
  &\le\lambda_{\min}(H_{\mathrm{sp},\lambda})
       +\lambda_{\max}(H_{\mathrm{loc},\lambda}),\\
  \lambda_{\min}(H_{Q,\lambda})
  &\ge\lambda_{\min}(H_{\mathrm{sp},\lambda})
       +\lambda_{\min}(H_{\mathrm{loc},\lambda}).
\end{align*}
Thus $\bra{\psi_\lambda}H_{Q,\lambda}\ket{\psi_\lambda}
-\lambda_{\min}(H_{Q,\lambda})\le W_{\mathrm{loc},\lambda}$.
The coherent-field trial construction in the proof of
Lemma~\ref{lem:memberwise-variational-bound}, applied to this $\ket{\psi_\lambda}$,
gives, as in \eqref{eq:Cg-before-common-bound},
$a_\lambda^2C_{g_\lambda}\le\bra{\psi_\lambda}H_{Q,\lambda}\ket{\psi_\lambda}
-E_{\mathrm{gs},\lambda}\le W_{\mathrm{loc},\lambda}+\Delta_{\mathrm{gs},\lambda}$.
Together with $E_\lambda\ge\eta W_{\mathrm{loc},\lambda}$ and
$\Delta_{\mathrm{gs},\lambda}\le E_\lambda$, this implies
\begin{equation}
  C_{g_\lambda}
  \le4(W_{\mathrm{loc},\lambda}+\Delta_{\mathrm{gs},\lambda})
  \le4M E_\lambda.
  \label{eq:preparation-Cg-bound}
\end{equation}
Let $h_\lambda(t):=\Tr[\rho_\lambda^{\picsup{I}}(t)
(\id_Q\otimes H_{F,\lambda})]$.
Taking expectations in the Schr\"odinger-picture joint state, Cauchy--Schwarz,
$\sigma_\lambda^\dagger\sigma_\lambda\le\id_Q$, and Lemma~\ref{lem:B-energy} give
\begin{equation}
  \begin{aligned}
    |\langle V_\lambda\rangle_t|
    &=2\bigl|\operatorname{Re}\langle\sigma_\lambda\otimes B_\lambda\rangle_t\bigr|
    \le2|\langle\sigma_\lambda\otimes B_\lambda\rangle_t|
    \\
    &\le2\sqrt{\langle\sigma_\lambda^\dagger\sigma_\lambda\otimes B_\lambda^\dagger B_\lambda\rangle_t}
    \le2\sqrt{\langle\id_Q\otimes B_\lambda^\dagger B_\lambda\rangle_t}
    \le2\sqrt{C_{g_\lambda}h_\lambda(t)}.
  \end{aligned}
  \label{eq:preparation-interaction-bound}
\end{equation}
Using \eqref{eq:family-H} and \eqref{eq:family-energy}, energy conservation,
$H_{Q,\lambda}\ge\lambda_{\min}(H_{Q,\lambda})\id_Q$, and
\eqref{eq:preparation-interaction-bound} therefore yield
\begin{equation}
  E_\lambda\ge h_\lambda(t)+\Delta_{\mathrm{gs},\lambda}
    -2\sqrt{C_{g_\lambda}h_\lambda(t)}.
  \label{eq:preparation-energy-conservation}
\end{equation}
Solving this quadratic inequality in $\sqrt{h_\lambda(t)}$ gives
\begin{equation}
  h_\lambda(t)
  \le\bigl(\sqrt{C_{g_\lambda}}
      +\sqrt{C_{g_\lambda}+E_\lambda-\Delta_{\mathrm{gs},\lambda}}\bigr)^2
  \le C_{\mathrm{fld}}E_\lambda,
  \label{eq:preparation-field-bound}
\end{equation}
using \eqref{eq:preparation-Cg-bound} and $\Delta_{\mathrm{gs},\lambda}\ge0$.
\end{proof}


\begin{enumerate}
  \item[\textbf{A2.}] \phantomsection\label{ass:A2}\textbf{Uniformly nontrivial state separation.}
  The system states $\rho_{\mathrm{out},\lambda}$ and $\rho_{\mathrm{in},\lambda}$ satisfy a uniform lower bound on their trace-distance separation (recall definition \eqref{eq:family-endpoint-distance})
  \begin{equation}
    \delta_\lambda\ge \delta_0>0
    \label{eq:uniform-distance}
  \end{equation}
  with one constant $\delta_0$ independent of $\lambda$.
\end{enumerate}

Each member is also assumed to satisfy the self-adjointness, regularity, finite-energy, and
absolute-continuity assumptions stated in Section~\ref{sec:model}.

\subsection{Proof of the square-root speed bound under uniform lower stability}
\label{sec:relative-speed-proof}

In this subsection we prove Theorem~\ref{thm:relative-speed}.  We first need a technical
lemma, whose trial-state construction is also used in Lemma~\ref{lem:preparation-A1}.

\begin{lemma}[Memberwise variational bound]
\label{lem:memberwise-variational-bound}
Under the model assumptions of Section~\ref{sec:model}, every member satisfies
\[
  C_{g_\lambda}\le4[\lambda_{\max}(H_{Q,\lambda})-E_{\mathrm{gs},\lambda}].
\]
No uniform assumption \assumptionlink{A0}--\assumptionlink{A2} is required.
\end{lemma}

\begin{proof}
Define the numerical radius
\begin{equation}
  w(\sigma_\lambda)
  :=\max_{\lVert\psi\rVert_2=1}
  \left|
    \left\langle\psi\middle|\sigma_\lambda\middle|\psi\right\rangle
  \right|.
  \label{eq:numerical-radius}
\end{equation}
In finite dimension the maximum exists.  We shall use the standard numerical-radius
inequality \cite[Section~1.3]{GustafsonRao1997}
\begin{equation}
  w(A)\ge\frac12\lVert A\rVert.
  \label{eq:numerical-radius-operator-norm}
\end{equation}
Applying \eqref{eq:numerical-radius-operator-norm} to $A=\sigma_\lambda$ and using the
normalization $\lVert\sigma_\lambda\rVert=1$ from \eqref{eq:sigma-normalization} gives
\begin{equation}
  w(\sigma_\lambda)\ge\frac12.
  \label{eq:numerical-radius-half}
\end{equation}

Choose a unit vector $\ket{\psi_\lambda}$ attaining
$w(\sigma_\lambda)$, and put
\begin{equation}
  s_\lambda
  :=\left\langle\psi_\lambda\middle|
    \sigma_\lambda
  \middle|\psi_\lambda\right\rangle.
  \label{eq:s-lambda}
\end{equation}
To construct admissible coherent trial vectors without assuming additional infrared
regularity, define
\begin{equation}
  D_n
  :=\left\{
    k\in[-n,n]:
    \frac1n\le\omega_\lambda(k)\le n,
    \ |g_\lambda(k)|\le n
  \right\}.
  \label{eq:Dn-definition}
\end{equation}
These measurable sets exhaust the modes with $\omega_\lambda(k)>0$ up to null sets.  Their
complement contributes nothing to $C_{g_\lambda}$ because its finiteness requires
$g_\lambda(k)=0$ almost everywhere where $\omega_\lambda(k)=0$.  On every $D_n$,
$|k|$, $\omega_\lambda(k)$, $1/\omega_\lambda(k)$, and $|g_\lambda(k)|$ are bounded.  Let
$\mathbf 1_{D_n}$ denote the indicator function of $D_n$: thus
$\mathbf 1_{D_n}(k)=1$ when $k\in D_n$ and $\mathbf 1_{D_n}(k)=0$ otherwise.  Define
\begin{equation}
  \alpha_{\lambda,n}(k)
  :=-
  \frac{s_\lambda^*g_\lambda^*(k)}{\omega_\lambda(k)}
  \mathbf 1_{D_n}(k).
  \label{eq:coherent-trial-amplitude}
\end{equation}
Then $\alpha_{\lambda,n}\in L^2(\mathbb R)$ and hence defines a normalized coherent state
$\ket{\alpha_{\lambda,n}}$, i.e., a state satisfying
$a_k\ket{\alpha_{\lambda,n}}=\alpha_{\lambda,n}(k)\ket{\alpha_{\lambda,n}}$ for almost every
$k$.  Introduce
\begin{equation}
  C_{\lambda,n}
  :=\int_{D_n}\dd k\,
  \frac{|g_\lambda(k)|^2}{\omega_\lambda(k)}.
  \label{eq:C-lambda-n}
\end{equation}
The coherent-state identities
$\langle a_k\rangle=\alpha_{\lambda,n}(k)$ and
$\langle a_k^\dagger a_k\rangle=|\alpha_{\lambda,n}(k)|^2$ give
\begin{align}
  \left\langle\alpha_{\lambda,n}\middle|
    H_{F,\lambda}
  \middle|\alpha_{\lambda,n}\right\rangle
  &=|s_\lambda|^2C_{\lambda,n},
  \label{eq:trial-field-energy}\\
  \left\langle\alpha_{\lambda,n}\middle|
    B_\lambda
  \middle|\alpha_{\lambda,n}\right\rangle
  &=-s_\lambda^*C_{\lambda,n}.
  \label{eq:trial-B-mean}
\end{align}
For the normalized product trial vector
$\ket{\psi_\lambda}\otimes\ket{\alpha_{\lambda,n}}$, the interaction expectation is therefore
\begin{align}
  \langle V_\lambda\rangle
  &=s_\lambda
    \left(-s_\lambda^*C_{\lambda,n}\right)
  +s_\lambda^*
    \left(-s_\lambda C_{\lambda,n}\right)
  \notag\\
  &=-2|s_\lambda|^2C_{\lambda,n}.
  \label{eq:trial-interaction-energy}
\end{align}
Adding the system, field, and interaction contributions in the same product state
$\ket{\psi_\lambda}\otimes\ket{\alpha_{\lambda,n}}$ gives
\begin{equation}
  \left\langle H_\lambda\right\rangle
  =
  \left\langle\psi_\lambda\middle|H_{Q,\lambda}\middle|\psi_\lambda\right\rangle
  -|s_\lambda|^2C_{\lambda,n}.
  \label{eq:trial-total-energy}
\end{equation}
By the definition $E_{\mathrm{gs},\lambda}=\inf\operatorname{spec}(H_\lambda)$ in
\eqref{eq:family-energy}, the expectation of $H_\lambda$ in every normalized admissible state
is at least $E_{\mathrm{gs},\lambda}$.  Applying this fact to the product state used in
\eqref{eq:trial-total-energy} gives
\begin{equation}
  E_{\mathrm{gs},\lambda}
  \le
  \left\langle\psi_\lambda\middle|H_{Q,\lambda}\middle|\psi_\lambda\right\rangle
  -|s_\lambda|^2C_{\lambda,n}.
  \label{eq:variational-bound-n}
\end{equation}
Since the integrand in \eqref{eq:C-lambda-n} is nonnegative, monotone convergence gives
$C_{\lambda,n}\to C_{g_\lambda}$.  Taking $n\to\infty$ in
\eqref{eq:variational-bound-n} yields
\begin{equation}
  |s_\lambda|^2C_{g_\lambda}
  \le
  \left\langle\psi_\lambda\middle|H_{Q,\lambda}\middle|\psi_\lambda\right\rangle
  -E_{\mathrm{gs},\lambda}.
  \label{eq:Cg-before-common-bound}
\end{equation}
Using
$\langle\psi_\lambda|H_{Q,\lambda}|\psi_\lambda\rangle
\le\lambda_{\max}(H_{Q,\lambda})$ and
$|s_\lambda|=w(\sigma_\lambda)\ge1/2$, we obtain the memberwise estimate
\begin{equation}
  C_{g_\lambda}
  \le4\bigl[\lambda_{\max}(H_{Q,\lambda})-E_{\mathrm{gs},\lambda}\bigr].
  \label{eq:uniform-Cg-bound}
\end{equation}
The quantity in square brackets is nonnegative.  Indeed, testing $H_\lambda$ on a system ground
vector tensored with the Fock vacuum gives
$E_{\mathrm{gs},\lambda}\le\lambda_{\min}(H_{Q,\lambda})
\le\lambda_{\max}(H_{Q,\lambda})$.

\end{proof}

\begin{theorem}[Square-root speed bound under uniform lower stability]
\label{thm:relative-speed}
Under the model assumptions of Section~\ref{sec:model} and the additional uniform assumptions
\assumptionlink{A0}--\assumptionlink{A2}, every member of the family satisfies
\begin{equation}
  \omega_{\mathrm{gate},\lambda}T_\lambda
  \le c_{\mathrm{stab}}\sqrt{T_\lambda E_\lambda},
  \label{eq:theorem9-result}
\end{equation}
where the dimensionless constant
\begin{equation}
  c_{\mathrm{stab}}
  :=\frac{8\pi}{\delta_0}
  \sqrt{C_{\mathrm{fld}}K_0}
  \label{eq:theorem9-c-stab}
\end{equation}
is $\lambda$-independent and defined in terms of the constants appearing in
Assumptions \assumptionlink{A0}--\assumptionlink{A2}.
\end{theorem}

This gives an absolute square-root law whenever $T_\lambda$ is chosen $\lambda$-independent.

\begin{proof}
We derive and integrate the exact projector-expectation-rate estimate for the member $\lambda$,
and then apply Assumption~\assumptionlink{A0}.
This proceeds similarly to the proof of~Theorem~\ref{thm:exact-speed}.
Let $P_{+,\lambda}$ project onto the strictly positive eigenspace of
$\rho_{\mathrm{out},\lambda}-\rho_{\mathrm{in},\lambda}$ and put
$p_{+,\lambda}(t):=\Tr[P_{+,\lambda}\rho_{Q,\lambda}^{\picsup{I}}(t)]$.
By \eqref{eq:endpoint-witness}, this fixed projector obeys
\begin{equation}
  \delta_\lambda
  =p_{+,\lambda}(t_{\mathrm{out},\lambda})
   -p_{+,\lambda}(t_{\mathrm{in},\lambda}).
  \label{eq:family-endpoint-witness}
\end{equation}
For the member $\lambda$, let
\begin{align}
  A_\lambda(t)
  &:=[P_{+,\lambda},\sigma_\lambda(t)],
  \label{eq:theorem9-A}\\
  X_\lambda(t)
  &:=A_\lambda(t)\otimes B_\lambda(t),
  \label{eq:theorem9-X}\\
  z_\lambda(t)
  &:=\Tr[\rho_\lambda^{\picsup{I}}(t)X_\lambda(t)].
  \label{eq:theorem9-z}
\end{align}
Since $P_{+,\lambda}$ is self-adjoint,
\begin{equation}
  [P_{+,\lambda},\sigma_\lambda^\dagger(t)]
  =-A_\lambda^\dagger(t).
  \label{eq:theorem9-adjoint-commutator}
\end{equation}
Thus the exact interaction-picture von Neumann equation gives
\begin{align}
  \frac{\dd}{\dd t}p_{+,\lambda}(t)
  &=-i\Tr\!\left[
    \rho_\lambda^{\picsup{I}}(t)
    \bigl(X_\lambda(t)-X_\lambda^\dagger(t)\bigr)
  \right]
  \notag\\
  &=2\operatorname{Im}z_\lambda(t).
  \label{eq:theorem9-population-derivative}
\end{align}
The Hilbert--Schmidt Cauchy--Schwarz inequality gives
\begin{align}
  |z_\lambda(t)|^2
  &\le\Tr[\rho_\lambda^{\picsup{I}}(t)X_\lambda^\dagger(t)X_\lambda(t)]
  \notag\\
  &\le\lVert A_\lambda(t)\rVert^2
  \Tr\!\left[
    \rho_\lambda^{\picsup{I}}(t)
    (\id_Q\otimes B_\lambda^\dagger(t)B_\lambda(t))
  \right].
  \label{eq:theorem9-z-bound-first}
\end{align}
Lemma~\ref{lem:projection-commutator} and
$\lVert\sigma_\lambda(t)\rVert=\lVert\sigma_\lambda\rVert=1$ imply
\begin{equation}
  \lVert A_\lambda(t)\rVert\le1.
  \label{eq:theorem9-A-bound}
\end{equation}
Lemma~\ref{lem:B-energy}, applied with $g_\lambda$ and $\omega_\lambda$, gives
\begin{equation}
  \Tr\!\left[
    \rho_\lambda^{\picsup{I}}(t)
    (\id_Q\otimes B_\lambda^\dagger(t)B_\lambda(t))
  \right]
  \le
  C_{g_\lambda}
  \Tr\!\left[
    \rho_\lambda^{\picsup{I}}(t)(\id_Q\otimes H_{F,\lambda})
  \right].
  \label{eq:theorem9-B-energy}
\end{equation}
Using Assumption \assumptionlink{A1}, equations
\eqref{eq:theorem9-population-derivative}--\eqref{eq:theorem9-B-energy} yield
\begin{equation}
  \left|
    \frac{\dd}{\dd t}p_{+,\lambda}(t)
  \right|
  \le2\sqrt{C_{g_\lambda}C_{\mathrm{fld}}E_\lambda}
  \qquad(t\in I_\lambda).
  \label{eq:theorem9-rate}
\end{equation}
Integration over $I_\lambda$, \eqref{eq:family-endpoint-witness}, and
Assumption \assumptionlink{A2} give
\begin{align}
  \delta_0
  &\le \delta_\lambda
  =
  p_{+,\lambda}(t_{\mathrm{out},\lambda})
  -p_{+,\lambda}(t_{\mathrm{in},\lambda})
  \notag\\
  &\le
  2\Delta t_\lambda
  \sqrt{C_{g_\lambda}C_{\mathrm{fld}}E_\lambda}.
  \label{eq:theorem9-integrated-rate}
\end{align}
Therefore
\begin{equation}
  \omega_{\mathrm{gate},\lambda}
  =\frac{2\pi}{\Delta t_\lambda}
  \le
  \frac{4\pi}{\delta_0}
  \sqrt{C_{g_\lambda}C_{\mathrm{fld}}E_\lambda}.
  \label{eq:theorem9-frequency-before-Cg}
\end{equation}
Multiplying by $T_\lambda$ and using \eqref{eq:uniform-Cg-bound} from
Lemma~\ref{lem:memberwise-variational-bound} together with
Assumption~\assumptionlink{A0}, we find
\begin{align}
  \omega_{\mathrm{gate},\lambda}T_\lambda
  &\le
  \frac{4\pi}{\delta_0}
  \sqrt{
    (T_\lambda C_{g_\lambda})C_{\mathrm{fld}}(T_\lambda E_\lambda)
  }
  \notag\\
  &\le
  \frac{8\pi}{\delta_0}
  \sqrt{C_{\mathrm{fld}}K_0}
  \sqrt{T_\lambda E_\lambda}
  \notag\\
  &=c_{\mathrm{stab}}\sqrt{T_\lambda E_\lambda}.
  \label{eq:theorem9-final}
\end{align}
This proves \eqref{eq:theorem9-result} with the constant
\eqref{eq:theorem9-c-stab}.
\end{proof}

\subsection{\texorpdfstring{Assumptions~\assumptionlink{B1}--\assumptionlink{B2}
and their interpretation}%
{Assumptions B1--B2 and their interpretation}}
\label{sec:quiet-interaction-bound}

Theorem~\ref{thm:relative-speed} obtains a uniform dimensionless bound on $C_{g_\lambda}$ from
Assumption~\assumptionlink{A0}.  We now give a different sufficient route.  It imposes no uniform
control on $T_\lambda[\lambda_{\max}(H_{Q,\lambda})-E_{\mathrm{gs},\lambda}]$, but instead
quantifies the physical expectation that there are intervals outside the gate during which
interaction-induced disturbance of the system populations relevant to the gate does not grow
with $\lambda$.  We express this through uniform bounds on dimensionless mean-square
interaction activation over selected quiet intervals.  Such intervals are natural for
pulse-driven gate implementations: before the pulse arrives or after it has passed, the
system should be comparatively quiet and the drive should have little influence on those
populations.  The intervals need not adjoin the gate.  The result permits
$\omega_\lambda$ to depend on $\lambda$; a $\lambda$-independent dispersion is the special case
$\omega_\lambda=\omega$.

Small net motion in a system population does not, by itself, bound the instantaneous
interaction: a rapidly oscillating generator can act strongly while its contributions cancel
over time.  The condition below instead bounds the mean-square interaction activation across
the partition relevant to the chosen state change, tested on the system states actually
visited during the gate and the field state at each quiet time.  The projector and system
state are paired at the same gate time.  These are product-state tests of the interaction;
the actual joint state may remain correlated.  The condition neither asserts exact
inactivity nor certifies the full gate channel.
We first define all of its ingredients.

For member $\lambda$, write
\begin{align}
  \Gamma_{g_\lambda}
  &:=\int_{\mathbb R}\dd k\,|g_\lambda(k)|^2,
  \label{eq:quiet-Gamma}\\
  H_{\mathrm{int},\lambda}(t)
  &:=\sigma_\lambda(t)\otimes B_\lambda(t)
   +\sigma_\lambda^\dagger(t)\otimes B_\lambda^\dagger(t),
  \label{eq:quiet-Hint}\\
  B_\lambda(t)
  &:=\int_{\mathbb R}\dd k\,
  g_\lambda(k)e^{-i\omega_\lambda(k)t}a_k.
  \label{eq:quiet-Bt}
\end{align}
The canonical commutation relations give the time-independent identity
\begin{equation}
  [B_\lambda(t),B_\lambda^\dagger(t)]
  =\Gamma_{g_\lambda}\id_F.
  \label{eq:quiet-B-commutator}
\end{equation}

Using the fixed projector $P_{+,\lambda}$ and population $p_{+,\lambda}$ from
\eqref{eq:family-endpoint-witness}, for a gate time $s\in I_\lambda$ define the paired
Schr\"odinger-picture projector and actual system state
\begin{align}
  P_{Q,\lambda}(s)
  &:=e^{-iH_{Q,\lambda}s}P_{+,\lambda}e^{iH_{Q,\lambda}s},
  \label{eq:paired-gate-projector}\\
  \rho_{Q,\lambda}^{\picsup{S}}(s)
  &:=\Tr_F\rho_\lambda^{\picsup{S}}(s)
   =e^{-iH_{Q,\lambda}s}\rho_{Q,\lambda}^{\picsup{I}}(s)e^{iH_{Q,\lambda}s}.
  \label{eq:paired-gate-state}
\end{align}
Here $\rho_\lambda^{\picsup{S}}(s)=e^{-iH_\lambda s}\rho_{0,\lambda}e^{iH_\lambda s}$.
The binary POVM $\{P_{+,\lambda},\id_Q-P_{+,\lambda}\}$ distinguishes
$\rho_{\mathrm{out},\lambda}$ and $\rho_{\mathrm{in},\lambda}$ optimally
(for equal priors, it achieves the maximal success probability $(1+\delta_\lambda)/2$).
Its freely rotated projector satisfies
\begin{equation}
  p_{+,\lambda}(s)
  =\Tr[P_{+,\lambda}\rho_{Q,\lambda}^{\picsup{I}}(s)]
  =\Tr[P_{Q,\lambda}(s)\rho_{Q,\lambda}^{\picsup{S}}(s)].
  \label{eq:paired-population}
\end{equation}
Only the pairs $(P_{Q,\lambda}(s),\rho_{Q,\lambda}^{\picsup{S}}(s))$ with $s\in I_\lambda$ will be tested;
neither entry is independently varied.

We retain Assumptions~\assumptionlink{A1} and
\assumptionlink{A2}.  Assumption
\assumptionlink{A0} is replaced by the following two conditions.

\begin{enumerate}
  \item[\textbf{B1.}] \phantomsection\label{ass:B1}\textbf{Uniform quiet-time activation.}
  For each $\lambda$, let $\mathcal A_\lambda$ be a nonempty finite or countable index set
  and choose a family of intervals
  \begin{equation}
    \mathcal W_{j,\lambda}:=[a_{j,\lambda},b_{j,\lambda}],
    \qquad -\infty<a_{j,\lambda}<b_{j,\lambda}<\infty,
    \quad j\in\mathcal A_\lambda,
    \label{eq:B1-quiet-intervals}
  \end{equation}
  pairwise disjoint and disjoint from the gate interval $I_\lambda$, except possibly at
  shared boundary times.  Define
  \begin{equation}
    \mathcal W_\lambda:=\bigcup_{j\in\mathcal A_\lambda}\mathcal W_{j,\lambda},
    \qquad
    0<|\mathcal W_\lambda|
    =\sum_{j\in\mathcal A_\lambda}(b_{j,\lambda}-a_{j,\lambda})<\infty.
    \label{eq:B1-quiet-set}
  \end{equation}
  The windows need not adjoin the gate interval: they may lie arbitrarily far before or
  after it, with gaps between them.  Their number, positions, and lengths may depend on
  $\lambda$; no uniform lower bound on their total length is imposed.
  More generally, the same argument permits any measurable set of finite positive measure
  disjoint from $I_\lambda$ up to a set of measure zero; we use intervals for simplicity.

  The autonomous unitary evolution
  $\rho_\lambda^{\picsup{S}}(t)=e^{-iH_\lambda t}\rho_{0,\lambda}e^{iH_\lambda t}$
  is defined for all $t\in\mathbb R$.  Negative-time windows therefore refer to its
  mathematical backward extension and may describe fictitious pre-evolution time intervals
  rather than intervals in the laboratory preparation history.  The bound remains valid even if that preparation
  used a different Hamiltonian.  When one chooses all windows to lie in $t\ge0$, they are physical quiet
  zones of the modeled post-preparation evolution.  Assumption~\assumptionlink{A1}
  is required only on $I_\lambda$, not on the quiet windows or in the intervening gaps.

  Let $\ket{\Phi_\lambda^{\picsup{I}}(t)}\in\mathcal H_F\otimes\mathcal H_R$ be any
  purification of the field marginal $\Tr_Q\rho_\lambda^{\picsup{I}}(t)$, where
  $\rho_\lambda^{\picsup{I}}(t)$ is the interaction-picture joint state defined in
  \eqref{eq:family-interaction-picture-state}.
  Denote the corresponding Schr\"odinger-picture joint state by
  $\rho_\lambda^{\picsup{S}}(t):=e^{-iH_\lambda t}\rho_{0,\lambda}e^{iH_\lambda t}$, and set
  $\ket{\Phi_\lambda^{\picsup{S}}(t)}
  :=(e^{-iH_{F,\lambda}t}\otimes\id_R)\ket{\Phi_\lambda^{\picsup{I}}(t)}$; this vector purifies
  $\Tr_Q\rho_\lambda^{\picsup{S}}(t)$.
  Define the field activation
  \begin{align}
    \mathcal J_\lambda(t)
    &:=\left\|
      (B_\lambda(t)\otimes\id_R)\ket{\Phi_\lambda^{\picsup{I}}(t)}
    \right\|_2
    \notag\\
    &=\sqrt{\Tr\!\left[\rho_\lambda^{\picsup{I}}(t)
      \bigl(\id_Q\otimes B_\lambda^\dagger(t)B_\lambda(t)\bigr)\right]}
    \notag\\
    &=\sqrt{\Tr\!\left[
      \rho_\lambda^{\picsup{S}}(t)
      \bigl(\id_Q\otimes B_\lambda^\dagger B_\lambda\bigr)
    \right]}.
    \label{eq:B1-J}
  \end{align}
  The last equality follows from the unitary change of picture and
  $B_\lambda(t)=e^{iH_{F,\lambda}t}B_\lambda e^{-iH_{F,\lambda}t}$.
  Put $\rho_{F,\lambda}^{\picsup{S}}(t):=\Tr_Q\rho_\lambda^{\picsup{S}}(t)$ and choose any
  purification $\ket{\Xi_\lambda^{\picsup{I}}(s)}\in\mathcal H_Q\otimes\mathcal H_{R_Q}$
  of $\rho_{Q,\lambda}^{\picsup{I}}(s)$, with
  $\dim\mathcal H_{R_Q}=\dim\mathcal H_Q$.  Define
  \begin{align}
    \ket{\Xi_\lambda^{\picsup{S}}(s)}
    &:=\bigl(e^{-iH_{Q,\lambda}s}\otimes\id_{R_Q}\bigr)
       \ket{\Xi_\lambda^{\picsup{I}}(s)},
    \notag\\
    \ket{\widehat\Phi_\lambda^{\picsup{I}}(s;t)}
    &:=\bigl(e^{iH_{F,\lambda}(s-t)}\otimes\id_R\bigr)
       \ket{\Phi_\lambda^{\picsup{I}}(t)}.
    \label{eq:paired-purification-pictures}
  \end{align}
  The first vector purifies $\rho_{Q,\lambda}^{\picsup{S}}(s)$; the second expresses
  the field state at $t$ in the free frame at $s$.
  With tensor factors reordered as $QF R_QR$, put
  \begin{equation}
    C_{QF,\lambda}(s)
    :=[P_{Q,\lambda}(s)\otimes\id_F,V_\lambda]
    \label{eq:paired-commutator}
  \end{equation}
  and define
  \begin{align}
    \mathcal D_\lambda(t)
    &:=\sup_{s\in I_\lambda}\left\|
      \bigl([P_{+,\lambda}\otimes\id_F,H_{\mathrm{int},\lambda}(s)]
        \otimes\id_{R_Q}\otimes\id_R\bigr)
      \bigl(\ket{\Xi_\lambda^{\picsup{I}}(s)}
        \otimes\ket{\widehat\Phi_\lambda^{\picsup{I}}(s;t)}\bigr)
    \right\|_2
    \notag\\
    &=\sup_{s\in I_\lambda}\left\|
      \bigl(C_{QF,\lambda}(s)\otimes\id_{R_Q}\otimes\id_R\bigr)
      \bigl(\ket{\Xi_\lambda^{\picsup{S}}(s)}
        \otimes\ket{\Phi_\lambda^{\picsup{S}}(t)}\bigr)
    \right\|_2
    \notag\\
    &=\sup_{s\in I_\lambda}\sqrt{\Tr\!\left[
      \bigl(\rho_{Q,\lambda}^{\picsup{S}}(s)\otimes\rho_{F,\lambda}^{\picsup{S}}(t)\bigr)
      C_{QF,\lambda}(s)^\dagger C_{QF,\lambda}(s)
    \right]}.
    \label{eq:B1-disturbance-action}
  \end{align}
  The first equality follows by unitary invariance of the vector norm; the phase factor
  in $\ket{\widehat\Phi_\lambda^{\picsup{I}}(s;t)}$ ensures that the norm tests the
  field state at $t$, although the interaction Hamiltonian is evaluated at $s$.
  The gate time $s$ and quiet time $t$ are independent, but $P_{Q,\lambda}(s)$ and
  $\rho_{Q,\lambda}^{\picsup{S}}(s)$ always use the same time $s$.  The tensor product in
  \eqref{eq:B1-disturbance-action} defines a test state assembled from two actual marginals at
  different times; it does not assert that the joint dynamical state factorizes.

  These expressions are quadratic-form norms on the domains specified in
  Section~\ref{sec:model}.  The trace identity in \eqref{eq:B1-disturbance-action} shows
  independence of the purifications.  For finite $\mathcal J_\lambda(t)$,
  the field-vector norms below are finite and the expression is continuous in $s$.
  The supremum can therefore be taken over a countable dense subset of the gate interval,
  ensuring measurability.  The integral bounds below require finiteness almost everywhere
  on the quiet set.

  The commutator also has a direct block interpretation.  For a fixed gate time $s$,
  use the shorthand $P:=P_{Q,\lambda}(s)$ and put
  $\Pi:=P\otimes\id_F$, $\Pi^\perp:=(\id_Q-P)\otimes\id_F$, and
  $S_P:=(2P-\id_Q)\otimes\id_F$.  Every interaction decomposes as
  \begin{align*}
    V_\lambda
    &=V_{\lambda,\parallel}^{(P)}+V_{\lambda,\perp}^{(P)},\\
    V_{\lambda,\parallel}^{(P)}
    &:=\Pi V_\lambda\Pi+\Pi^\perp V_\lambda\Pi^\perp,\\
    V_{\lambda,\perp}^{(P)}
    &:=\Pi V_\lambda\Pi^\perp+\Pi^\perp V_\lambda\Pi.
  \end{align*}
  The block-diagonal part commutes with $S_P$, while the block-off-diagonal part
  anticommutes with $S_P$; the latter need not anticommute with $P$ itself.  Moreover,
  \begin{equation}
    [P\otimes\id_F,V_\lambda]
    =S_PV_{\lambda,\perp}^{(P)}.
    \label{eq:B1-off-diagonal-commutator}
  \end{equation}
  Since $S_P$ is unitary, the trace under the square root in
  \eqref{eq:B1-disturbance-action} measures the mean-square contribution of the
  part of $V_\lambda$ coupling the $P_{Q,\lambda}(s)$ and $\id_Q-P_{Q,\lambda}(s)$ sectors, weighted
  by the actual gate state $\rho_{Q,\lambda}^{\picsup{S}}(s)$.  Thus \eqref{eq:B1-disturbance-action} tests
  the largest interaction activation among the paired gate states.

  We assume $\lambda$-independent constants $K_J<\infty$ and $K_D<\infty$ such that
  \begin{align}
    \frac{T_\lambda^2}{|\mathcal W_\lambda|}
    \int_{\mathcal W_\lambda}\mathcal J_\lambda^2(t)\dd t
    &\le K_J,
    \label{eq:B1-J-bound}\\
    \frac{T_\lambda^2}{|\mathcal W_\lambda|}
    \int_{\mathcal W_\lambda}\mathcal D_\lambda^2(t)\dd t
    &\le K_D.
    \label{eq:B1-D-bound}
  \end{align}
  Bounds \eqref{eq:B1-J-bound} and \eqref{eq:B1-D-bound} control both the field activation
  and the interaction activation on average over $\mathcal W_\lambda$, with no
  separate bound on its component intervals.  Together, \eqref{eq:B1-J-bound} and
  \eqref{eq:B1-D-bound} constitute Assumption~\assumptionlink{B1}.

  \begin{lemma}[Operational meaning of the field activation $\mathcal J_\lambda(t)$]
  \label{lem:activation-detector}
  Fix $\lambda$ and a time $t$ with $\mathcal J_\lambda(t)<\infty$.
  Prepare an auxiliary two-level detector independently in $\ket{g}$, with the field
  in its actual marginal $\rho_{F,\lambda}^{\picsup{S}}(t)$, and evolve under the
  matched absorption interaction
  \begin{equation}
    V_{\mathrm{det},\lambda}
    :=\ket{e}\bra{g}\otimes B_\lambda
      +\ket{g}\bra{e}\otimes B_\lambda^\dagger.
    \label{eq:activation-detector-H}
  \end{equation}
  Let $U_{\mathrm{det},\lambda}(\epsilon):=e^{-i\epsilon V_{\mathrm{det},\lambda}}$.
  The probability of finding the detector in $\ket{e}$ after the interaction is
  \begin{align}
    p_{\mathrm{exc},\lambda}(t;\epsilon)
    &:=\Tr\!\left[
      (\ket{e}\bra{e}\otimes\id_F)U_{\mathrm{det},\lambda}(\epsilon)
      (\ket{g}\bra{g}\otimes\rho_{F,\lambda}^{\picsup{S}}(t))
      U_{\mathrm{det},\lambda}(\epsilon)^\dagger
    \right]
    \notag\\
    &=\epsilon^2\mathcal J_\lambda^2(t)+o(\epsilon^2)
    \qquad(\epsilon\to0).
    \label{eq:activation-detector-probability}
  \end{align}
  This is an auxiliary test of the selected field channel, not an assumption on the
  actual gate dynamics.  The evolution is generated by the interaction alone, and
  the remainder is stated for fixed $\lambda,t$.
  \end{lemma}

  \begin{proof}
  If $\Gamma_{g_\lambda}=0$, then $B_\lambda=0$ and the assertion is immediate.
  Otherwise, $b_\lambda:=B_\lambda/\sqrt{\Gamma_{g_\lambda}}$ is the annihilation
  operator of a normalized collective mode, with
  $[b_\lambda,b_\lambda^\dagger]=\id_F$.  In the corresponding Fock factorization,
  the interaction acts trivially on the other field modes.  For each collective-mode
  occupation $n\ge1$, it couples $\ket{g,n}$ only to $\ket{e,n-1}$, with matrix element
  $\sqrt{\Gamma_{g_\lambda}n}$; $\ket{g,0}$ is invariant.
  These two-dimensional blocks define the canonical self-adjoint interaction and give
  \begin{equation}
    p_{\mathrm{exc},\lambda}(t;\epsilon)
    =\Tr\!\left[
      \rho_{F,\lambda}^{\picsup{S}}(t)
      \sin^2\!\left(\epsilon\sqrt{B_\lambda^\dagger B_\lambda}\right)
    \right].
    \label{eq:activation-detector-exact}
  \end{equation}
  The identity holds for arbitrary $\rho_{F,\lambda}^{\picsup{S}}(t)$, i.e.\ any density
  operator on $\mathcal H_F$.

  Let $\mu_{\lambda,t}$ denote the probability distribution of the outcomes obtained by
  measuring $B_\lambda^\dagger B_\lambda$ in the field state
  $\rho_{F,\lambda}^{\picsup{S}}(t)$.
  (Since $B_\lambda^\dagger B_\lambda=\Gamma_{g_\lambda}b_\lambda^\dagger b_\lambda$
  has pure point spectrum, $\mu_{\lambda,t}$ is a purely atomic measure: a weighted sum
  of Dirac point masses at $\Gamma_{g_\lambda}n$, $n=0,1,\ldots$, with weights given
  by the collective-mode occupation probabilities.)
  Integration against $\mu_{\lambda,t}$ therefore averages over these measurement outcomes.
  To justify the short-duration limit, put
  \[
    f_\epsilon(x):=\frac{\sin^2(\epsilon\sqrt{x})}{\epsilon^2},
    \qquad x\ge0,\quad \epsilon\in\mathbb R\setminus\{0\}.
  \]
  The conditions for dominated convergence hold:
  \begin{itemize}
    \item Measurability: each $f_\epsilon$ is continuous in $x$, hence measurable.
    \item Pointwise convergence: $f_\epsilon(x)\to x$ for every $x\ge0$ as $\epsilon\to0$.
    \item Integrable domination: $|\sin y|\le|y|$ gives
      $0\le f_\epsilon(x)\le x$, and the dominating function $x$ is integrable because
      \[
        \int_0^\infty x\,\dd\mu_{\lambda,t}(x)
        =\Tr[\rho_{F,\lambda}^{\picsup{S}}(t)B_\lambda^\dagger B_\lambda]
        =\mathcal J_\lambda^2(t)<\infty
      \]
      by the lemma's assumption.
  \end{itemize}
  Applying dominated convergence along any sequence of nonzero $\epsilon\to0$ therefore yields
  \[
    \frac{p_{\mathrm{exc},\lambda}(t;\epsilon)}{\epsilon^2}
    =\int_0^\infty f_\epsilon(x)\,\dd\mu_{\lambda,t}(x)
    \longrightarrow\mathcal J_\lambda^2(t),
  \]
  proving
  \eqref{eq:activation-detector-probability} without any assumption on higher field moments.
  In particular, vacuum produces no detector excitation, which explains the normally
  ordered product $B_\lambda^\dagger B_\lambda$.
  \end{proof}

  \begin{lemma}[Operational meaning of the interaction activation $\mathcal D_\lambda(t)$]
  \label{lem:paired-measurement-switching}
  Fix $\lambda$ and a time $t$ with $\mathcal J_\lambda(t)<\infty$.
  For each gate time $s\in I_\lambda$, prepare the test state
  \[
    \tau_\lambda(s,t)
    :=\rho_{Q,\lambda}^{\picsup{S}}(s)\otimes\rho_{F,\lambda}^{\picsup{S}}(t).
  \]
  Perform the ideal projective measurement
  $\{P_{Q,\lambda}(s),\id_Q-P_{Q,\lambda}(s)\}$ on the system, with the L\"uders state update,
  evolve for duration $\epsilon$ under the interaction alone,
  $U_\lambda(\epsilon):=e^{-i\epsilon V_\lambda}$, and repeat the same measurement.
  Let $p_{\mathrm{flip},\lambda}(s,t;\epsilon)$ be the probability that the two outcomes differ,
  averaged over the first outcome.  Then
  \begin{equation}
    \sup_{s\in I_\lambda}p_{\mathrm{flip},\lambda}(s,t;\epsilon)
    =\epsilon^2\mathcal D_\lambda^2(t)+o(\epsilon^2)
    \qquad(\epsilon\to0).
    \label{eq:paired-measurement-switching}
  \end{equation}
  \end{lemma}

  \begin{proof}
  Suppress $\lambda,t$ where unambiguous, and put
  $\Pi_s:=P_{Q,\lambda}(s)\otimes\id_F$ and $\Pi_s^\perp:=\id_{QF}-\Pi_s$.
  Use the purifications defined above and abbreviate
  $\ket{\Phi_t}:=\ket{\Phi_\lambda^{\picsup{S}}(t)}$ and
  $\ket{\Xi_s}:=\ket{\Xi_\lambda^{\picsup{S}}(s)}$.  Write
  $\ket{\Omega_s}:=\ket{\Xi_s}\otimes\ket{\Phi_t}$, with tensor factors reordered as needed.
  All operators below act trivially on the reference systems.
  Since $\ket{\Omega_s}$ purifies $\tau_\lambda(s,t)$, the measurement rule gives
  \begin{align}
    p_{\mathrm{flip},\lambda}(s,t;\epsilon)
    &=\Tr[\Pi_s^\perp U_\lambda(\epsilon)\Pi_s
       \tau_\lambda(s,t)\Pi_s U_\lambda(\epsilon)^\dagger\Pi_s^\perp]
       \notag\\
    &\quad+\Tr[\Pi_s U_\lambda(\epsilon)\Pi_s^\perp
       \tau_\lambda(s,t)\Pi_s^\perp U_\lambda(\epsilon)^\dagger\Pi_s]
       \notag\\
    &=\bigl\lVert\Pi_s^\perp U_\lambda(\epsilon)\Pi_s\ket{\Omega_s}\bigr\rVert_2^2
      +\bigl\lVert\Pi_s U_\lambda(\epsilon)\Pi_s^\perp\ket{\Omega_s}\bigr\rVert_2^2.
    \label{eq:paired-flip-probability}
  \end{align}
  This formula includes the disturbance caused by the first measurement.

  For completeness, the interaction unitary is generated by the self-adjoint closure of
  $V_\lambda$ initially defined on the dense finite-particle domain.  For the present
  finite-dimensional system with $g_\lambda\in L^2$, the annihilation and creation
  bounds give, for every $\ket{\psi}\in\mathcal H_Q\otimes\mathcal H_F$ with at most
  $m$ field particles and every integer $n\ge1$,
  \begin{align*}
    \bigl\lVert V_\lambda^n\ket{\psi}\bigr\rVert_2
    &\le (2\sqrt{\Gamma_{g_\lambda}})^n
      \left(\prod_{j=1}^n\sqrt{m+j}\right)\lVert\ket{\psi}\rVert_2\\
    &=(2\sqrt{\Gamma_{g_\lambda}})^n
      \sqrt{\frac{(m+n)!}{m!}}\,\lVert\ket{\psi}\rVert_2.
  \end{align*}
  Finite-particle vectors are therefore analytic for $V_\lambda$.  Nelson's analytic-vector
  theorem~\cite[Theorem~X.39]{ReedSimon1975} implies essential self-adjointness on that
  domain.

  Finite $\mathcal J_\lambda(t)$ and $\Gamma_{g_\lambda}$ imply
  \[
    \bigl\lVert B_\lambda\ket{\Phi_t}\bigr\rVert_2=\mathcal J_\lambda(t),\qquad
    \bigl\lVert B_\lambda^\dagger\ket{\Phi_t}\bigr\rVert_2^2
    =\mathcal J_\lambda^2(t)+\Gamma_{g_\lambda}.
  \]
  Hence the finite-dimensional space
  $\mathcal E_t:=(\mathcal H_Q\otimes\mathcal H_{R_Q})
  \otimes\operatorname{span}\{\ket{\Phi_t}\}$ lies in the domain of $V_\lambda$.
  It contains $\ket{\Omega_s}$, $\Pi_s\ket{\Omega_s}$, and $\Pi_s^\perp\ket{\Omega_s}$ for every $s$.

  On the domain of $V_\lambda$, put
  $R_\epsilon:=U_\lambda(\epsilon)-\id+i\epsilon V_\lambda$.
  The first-order expansion of the unitary, applied to a finite basis of $\mathcal E_t$, yields
  \[
    \sup_{\substack{\ket{v}\in\mathcal E_t\\\lVert\ket{v}\rVert_2\le1}}
    \bigl\lVert R_\epsilon\ket{v}\bigr\rVert_2
    =o(|\epsilon|).
  \]
  To apply this bound directly to the norms in \eqref{eq:paired-flip-probability},
  substitute $U_\lambda(\epsilon)=\id-i\epsilon V_\lambda+R_\epsilon$.
  Using $\Pi_s^\perp\Pi_s=0$ and the reverse triangle inequality gives
  \begin{align*}
    &\left|
      \bigl\lVert\Pi_s^\perp U_\lambda(\epsilon)\Pi_s\ket{\Omega_s}\bigr\rVert_2
      -|\epsilon|\bigl\lVert\Pi_s^\perp V_\lambda\Pi_s\ket{\Omega_s}\bigr\rVert_2
    \right|\\
    &\qquad\le\bigl\lVert\Pi_s^\perp R_\epsilon\Pi_s\ket{\Omega_s}\bigr\rVert_2
    \le\bigl\lVert R_\epsilon\Pi_s\ket{\Omega_s}\bigr\rVert_2.
  \end{align*}
  The second inequality uses that $\Pi_s^\perp$ cannot increase norms.
  Likewise, $\Pi_s$ cannot increase norms, so
  $\lVert\Pi_s\ket{\Omega_s}\rVert_2\le\lVert\ket{\Omega_s}\rVert_2=1$.
  Since $\Pi_s\ket{\Omega_s}\in\mathcal E_t$, the preceding uniform bound therefore gives
  \[
    \sup_{s\in I_\lambda}
    \bigl\lVert\Pi_s^\perp R_\epsilon\Pi_s\ket{\Omega_s}\bigr\rVert_2
    \le\sup_{\substack{\ket{v}\in\mathcal E_t\\\lVert\ket{v}\rVert_2\le1}}
    \bigl\lVert R_\epsilon\ket{v}\bigr\rVert_2=o(|\epsilon|).
  \]
  The same argument applies with $\Pi_s$ and $\Pi_s^\perp$ interchanged.
  The projected $V_\lambda$ norms are bounded uniformly in $s$ because $V_\lambda$ is bounded
  as a map from the finite-dimensional space $\mathcal E_t$ into the full Hilbert space,
  and both $V_\lambda$ and $\mathcal E_t$ are independent of $s$.
  The difference-of-squares identity thus shows that each squared norm in
  \eqref{eq:paired-flip-probability} differs from its leading term by $o(\epsilon^2)$,
  uniformly in $s$.  Consequently,
  \begin{align}
    p_{\mathrm{flip},\lambda}(s,t;\epsilon)
    &=\epsilon^2\bigl(
      \bigl\lVert\Pi_s^\perp V_\lambda\Pi_s\ket{\Omega_s}\bigr\rVert_2^2
      +\bigl\lVert\Pi_s V_\lambda\Pi_s^\perp\ket{\Omega_s}\bigr\rVert_2^2\bigr)+o(\epsilon^2)
      \notag\\
    &=\epsilon^2\Tr[\tau_\lambda(s,t)
      C_{QF,\lambda}(s)^\dagger C_{QF,\lambda}(s)]+o(\epsilon^2).
    \label{eq:paired-flip-expansion}
  \end{align}
  The last equality uses
  $C_{QF,\lambda}(s)=\Pi_s V_\lambda\Pi_s^\perp-\Pi_s^\perp V_\lambda\Pi_s$:
  its two output vectors occupy orthogonal system sectors.
  The remainder remains uniform in $s$, so taking the supremum proves
  \eqref{eq:paired-measurement-switching} directly with
  $\mathcal D_\lambda^2(t)$ as defined in \eqref{eq:B1-disturbance-action}.
  \end{proof}

  \item[\textbf{B2.}] \phantomsection\label{ass:B2}\textbf{Uniform inverse-frequency compatibility.}
  There is a dimensionless constant $K_\omega<\infty$ such that
  \begin{equation}
    C_{g_\lambda}
    \le K_\omega T_\lambda \Gamma_{g_\lambda}
    \qquad\text{for every }\lambda.
    \label{eq:B2-inverse-frequency}
  \end{equation}
  Equivalently, for $g_\lambda\ne0$,
  \begin{equation}
    \frac{
      \displaystyle\int_{\{k\in\mathbb R:\,g_\lambda(k)\ne0\}}
      |g_\lambda(k)|^2/\omega_\lambda(k)\dd k
    }{
      \displaystyle T_\lambda\int_{\{k\in\mathbb R:\,g_\lambda(k)\ne0\}}
      |g_\lambda(k)|^2\dd k
    }
    \le K_\omega.
    \label{eq:B2-ratio-form}
  \end{equation}
  This is an infrared condition on the coupling relative to the free dispersion.  It does not
  bound the peak of $g_\lambda$, its bandwidth, or $\Gamma_{g_\lambda}$ separately.  When
  both $T_\lambda$ and the dispersion are $\lambda$-independent, only uniformity over the
  allowed coupling shapes remains.  A positive dispersion alone does not guarantee this
  dimensionless uniformity if it approaches zero and
  the coupling is allowed to concentrate there.
\end{enumerate}

Assumption~\assumptionlink{B2} imposes only a uniform infrared condition on the coupling relative to
the free dispersion; it is not an ultraviolet cutoff.  The following two lemmas make this
clear.  The first collects two straightforward gap criteria, one for the dispersion and one
for the coupling, while the second is a more nuanced gapless criterion based only on
localization near zero energy.

\begin{colonlemma}[Uniform dimensionless frequency- or coupling-gap sufficiency for \assumptionlink{B2}]
\label{lem:B2-frequency-gap}
\leavevmode\par\nobreak
\noindent\emph{Case 1) Uniform frequency gap.}
Suppose there exists a dimensionless constant $\vartheta_{\min}>0$, independent of $\lambda$,
such that
\begin{equation}
  T_\lambda\omega_\lambda(k)\ge\vartheta_{\min}
  \qquad\text{for almost every }k\text{ and every }\lambda.
  \label{eq:B2-gapped-sufficient}
\end{equation}
Then Assumption~\assumptionlink{B2} holds with
$K_\omega=\vartheta_{\min}^{-1}$.

\noindent\emph{Case 2) Uniform coupling gap.}
Alternatively, suppose there exist positive frequencies
$\omega_{\mathrm{gap},\lambda}$ and a dimensionless constant
$\vartheta_{\mathrm{gap}}>0$, independent of $\lambda$, such that
\begin{equation}
 \begin{gathered}
  g_\lambda(k)=0
  \quad\text{for almost every }k\text{ with }
  \omega_\lambda(k)<\omega_{\mathrm{gap},\lambda},\\
  T_\lambda\omega_{\mathrm{gap},\lambda}\ge\vartheta_{\mathrm{gap}}
  \quad\text{for every }\lambda.
 \end{gathered}
  \label{eq:B2-coupling-gap-sufficient}
\end{equation}
Thus the coupling excludes the frequency interval
$[0,\omega_{\mathrm{gap},\lambda})$.
Then Assumption~\assumptionlink{B2} holds with
$K_\omega=\vartheta_{\mathrm{gap}}^{-1}$.
\end{colonlemma}

\begin{proof}
For the frequency-gap case, the gap assumption gives
$\omega_\lambda(k)^{-1}\le T_\lambda/\vartheta_{\min}$ almost everywhere,
and therefore
\begin{equation}
  C_{g_\lambda}
  =\int_{\{k\in\mathbb R:\,g_\lambda(k)\ne0\}}\dd k\,
    \frac{|g_\lambda(k)|^2}{\omega_\lambda(k)}
  \le\frac{T_\lambda}{\vartheta_{\min}}\Gamma_{g_\lambda}.
  \label{eq:B2-gap-proof}
\end{equation}
This is precisely \eqref{eq:B2-inverse-frequency} with
$K_\omega=\vartheta_{\min}^{-1}$.

For the coupling-gap case, \eqref{eq:B2-coupling-gap-sufficient} gives
$\omega_\lambda(k)\ge\omega_{\mathrm{gap},\lambda}
\ge\vartheta_{\mathrm{gap}}/T_\lambda$ almost everywhere on the coupled set
$\{k:g_\lambda(k)\ne0\}$.  Thus the estimate
\eqref{eq:B2-gap-proof} applies with $\vartheta_{\min}$ replaced by
$\vartheta_{\mathrm{gap}}$, proving \assumptionlink{B2} with
$K_\omega=\vartheta_{\mathrm{gap}}^{-1}$.
\end{proof}

\begin{lemma}[Gapless infrared-localization sufficiency for \assumptionlink{B2}]
\label{lem:B2-energy-localization}
Suppose that the one-particle frequency operator admits an absolutely continuous energy
representation in which, for every $\lambda$, the coupling is a possibly vector-valued function
$\gamma_\lambda(y)$ taking values in a Hilbert space $\mathcal K_\lambda$ and satisfying
\begin{align}
  \Gamma_{g_\lambda}
  &=\int_0^\infty\lVert\gamma_\lambda(y)\rVert_2^2\dd y,
  \notag\\
  C_{g_\lambda}
  &=\int_0^\infty
    \frac{\lVert\gamma_\lambda(y)\rVert_2^2}{y}\dd y.
  \label{eq:energy-representation-identities}
\end{align}
Assume further that every member has an infrared cutoff $y_{\mathrm{IR},\lambda}>0$ and a
time $\tau_{\mathrm{loc},\lambda}<\infty$ such that
\begin{align}
  \gamma_\lambda|_{(0,y_{\mathrm{IR},\lambda})}
  &\in H^1((0,y_{\mathrm{IR},\lambda});\mathcal K_\lambda),
  \qquad \gamma_\lambda(0)=0,
  \notag\\
  \lVert\partial_y\gamma_\lambda\rVert_{L^2(0,y_{\mathrm{IR},\lambda})}
  &\le\tau_{\mathrm{loc},\lambda}
  \lVert\gamma_\lambda\rVert_{L^2(0,y_{\mathrm{IR},\lambda})}
  \label{eq:energy-localization-condition}
\end{align}
and that there is one dimensionless $K_{\mathrm{loc}}<\infty$ satisfying
\begin{equation}
  \frac{2\tau_{\mathrm{loc},\lambda}+y_{\mathrm{IR},\lambda}^{-1}}
       {T_\lambda}
  \le K_{\mathrm{loc}}
  \qquad\text{for every }\lambda.
  \label{eq:energy-localization-Komega}
\end{equation}
Then Assumption~\assumptionlink{B2} holds with $K_\omega=K_{\mathrm{loc}}$.
\end{lemma}

\begin{proof}
The vector-valued Hardy inequality \cite[Theorem~330]{HardyLittlewoodPolya1952}, localized to
$(0,y_{\mathrm{IR},\lambda})$, is
\begin{equation}
  \int_0^{y_{\mathrm{IR},\lambda}}
  \frac{\lVert\gamma_\lambda(y)\rVert_2^2}{y^2}\dd y
  \le4\int_0^{y_{\mathrm{IR},\lambda}}
  \lVert\partial_y\gamma_\lambda(y)\rVert_2^2\dd y.
  \label{eq:energy-Hardy}
\end{equation}
For completeness, this form follows first for smooth $\gamma_\lambda$ vanishing at $y=0$ by
integration by parts:
\begin{align}
  \int_0^{y_{\mathrm{IR},\lambda}}
    \frac{\lVert\gamma_\lambda(y)\rVert_2^2}{y^2}\dd y
  &=-\frac{\lVert\gamma_\lambda(y_{\mathrm{IR},\lambda})\rVert_2^2}{y_{\mathrm{IR},\lambda}}
    +2\operatorname{Re}\int_0^{y_{\mathrm{IR},\lambda}}
    \frac{\langle\gamma_\lambda(y),
    \partial_y\gamma_\lambda(y)\rangle}{y}\dd y
  \notag\\
  &\le2
  \left(\int_0^{y_{\mathrm{IR},\lambda}}
  \frac{\lVert\gamma_\lambda(y)\rVert_2^2}{y^2}\dd y\right)^{1/2}
  \left(\int_0^{y_{\mathrm{IR},\lambda}}
  \lVert\partial_y\gamma_\lambda(y)\rVert_2^2\dd y\right)^{1/2}.
  \label{eq:energy-Hardy-proof}
\end{align}
Thus, rearranging gives \eqref{eq:energy-Hardy}.  Splitting the defining
integral for $C_{g_\lambda}$ at $y_{\mathrm{IR},\lambda}$, applying Cauchy--Schwarz and
\eqref{eq:energy-Hardy} below the cutoff, and using
$y^{-1}\le y_{\mathrm{IR},\lambda}^{-1}$ above it,
we obtain
\begin{align}
  C_{g_\lambda}
  &\le
  \left(\int_0^{y_{\mathrm{IR},\lambda}}
    \lVert\gamma_\lambda(y)\rVert_2^2\dd y\right)^{1/2}
  \left(\int_0^{y_{\mathrm{IR},\lambda}}
  \frac{\lVert\gamma_\lambda(y)\rVert_2^2}{y^2}\dd y\right)^{1/2}
  +\frac{1}{y_{\mathrm{IR},\lambda}}
  \int_{y_{\mathrm{IR},\lambda}}^\infty
    \lVert\gamma_\lambda(y)\rVert_2^2\dd y
  \notag\\
  &\le
  2\tau_{\mathrm{loc},\lambda}
  \int_0^{y_{\mathrm{IR},\lambda}}\lVert\gamma_\lambda(y)\rVert_2^2\dd y
  +\frac{1}{y_{\mathrm{IR},\lambda}}
  \int_{y_{\mathrm{IR},\lambda}}^\infty
    \lVert\gamma_\lambda(y)\rVert_2^2\dd y
  \notag\\
  &\le
  \left(2\tau_{\mathrm{loc},\lambda}+\frac{1}{y_{\mathrm{IR},\lambda}}\right)
  \Gamma_{g_\lambda}.
  \label{eq:energy-localization-implies-B2}
\end{align}
Dividing the coefficient by $T_\lambda$ and using
\eqref{eq:energy-localization-Komega} proves \eqref{eq:B2-inverse-frequency} with
$K_\omega=K_{\mathrm{loc}}$.  Thus no regularity or localization hypothesis is needed above
the assigned infrared cutoff.
\end{proof}

The cutoff hypothesis can be expressed directly in the original mode variable.  Suppose that,
over the infrared energy interval $(0,y_{\mathrm{IR},\lambda})$, the preimage of the dispersion
decomposes into monotone differentiable branches indexed by $j$, with inverse functions
$k_{\lambda,j}(y)$.  Then the coupling in the energy representation is the branch vector
\begin{equation}
  \gamma_\lambda(y)
  =\left(
    \frac{g_\lambda(k_{\lambda,j}(y))}
    {\sqrt{\left|\left.\frac{\dd\omega_\lambda(k)}{\dd k}
      \right|_{k=k_{\lambda,j}(y)}\right|}}
  \right)_j,
  \qquad 0<y<y_{\mathrm{IR},\lambda}.
  \label{eq:energy-representation-gamma}
\end{equation}
The absolute value of the Jacobian covers both increasing and decreasing branches; the
resulting vector-valued spectral representation is standard
\cite[Chap.~VII]{ReedSimon1980}.
Thus Lemma~\ref{lem:B2-energy-localization} controls only those amplitudes below
$y_{\mathrm{IR},\lambda}$.  The coupling at higher energies is unrestricted; uniformity is
required only through the dimensionless bound \eqref{eq:energy-localization-Komega}.

\medskip
\noindent\textbf{{Example: the linear-dispersion case.}}
For the full-line linear dispersion $\omega_\lambda(k)=\nu_\lambda|k|$, with $\nu_\lambda>0$, put
$k_{\mathrm{IR},\lambda}:=y_{\mathrm{IR},\lambda}/\nu_\lambda$ and define the two-branch
coupling $\mathbf g_\lambda(r):=(g_\lambda(r),g_\lambda(-r))$ for $r>0$.  Equation
\eqref{eq:energy-representation-gamma} gives
\begin{align}
  \gamma_\lambda(y)
  &=\nu_\lambda^{-1/2}\mathbf g_\lambda(y/\nu_\lambda),
  \qquad 0<y<y_{\mathrm{IR},\lambda},
  \notag\\
  \frac{
    \lVert\partial_y\gamma_\lambda\rVert_{L^2(0,y_{\mathrm{IR},\lambda})}
  }{
    \lVert\gamma_\lambda\rVert_{L^2(0,y_{\mathrm{IR},\lambda})}
  }
  &=\frac1{\nu_\lambda}
  \frac{
    \lVert\partial_k g_\lambda\rVert_{L^2(-k_{\mathrm{IR},\lambda},k_{\mathrm{IR},\lambda})}
  }{
    \lVert g_\lambda\rVert_{L^2(-k_{\mathrm{IR},\lambda},k_{\mathrm{IR},\lambda})}
  }.
  \label{eq:linear-energy-gamma}
\end{align}
The ratio is used when the infrared restriction of $g_\lambda$ is nonzero; if that restriction
vanishes, its contribution to $C_{g_\lambda}$ is zero and the same conclusion is immediate.
Consequently, if
$g_\lambda|_{(-k_{\mathrm{IR},\lambda},k_{\mathrm{IR},\lambda})}
\in H^1((-k_{\mathrm{IR},\lambda},k_{\mathrm{IR},\lambda}))$,
$g_\lambda(0)=0$, and
\begin{equation}
  \lVert\partial_k g_\lambda\rVert_{L^2(-k_{\mathrm{IR},\lambda},k_{\mathrm{IR},\lambda})}
  \le L_{\mathrm{IR},\lambda}
  \lVert g_\lambda\rVert_{L^2(-k_{\mathrm{IR},\lambda},k_{\mathrm{IR},\lambda})}
  \label{eq:linear-fixed-interaction-region}
\end{equation}
then \eqref{eq:energy-localization-condition} holds with
\begin{equation}
  \tau_{\mathrm{loc},\lambda}
  =\frac{L_{\mathrm{IR},\lambda}}{\nu_\lambda}.
  \label{eq:linear-B2-constant}
\end{equation}
Consequently, this example satisfies \assumptionlink{B2} whenever
\begin{equation}
  \sup_\lambda
  \frac{
    2L_{\mathrm{IR},\lambda}/\nu_\lambda
    +y_{\mathrm{IR},\lambda}^{-1}
  }{T_\lambda}
  <\infty.
  \label{eq:linear-B2-uniformity}
\end{equation}
No regularity or localization condition is imposed on $g_\lambda(k)$ for
$|k|\ge k_{\mathrm{IR},\lambda}$.

There is also a spatial interpretation of a useful stronger version of
\eqref{eq:linear-fixed-interaction-region}.  Define the infrared piece
$g_{\lambda,<}(k):=\boldsymbol{1}_{(-k_{\mathrm{IR},\lambda},k_{\mathrm{IR},\lambda})}(k)g_\lambda(k)$, where
$\boldsymbol{1}_{(-k_{\mathrm{IR},\lambda},k_{\mathrm{IR},\lambda})}$ denotes the indicator function of the displayed interval.
Suppose that this infrared piece belongs to $H^1(\mathbb R)$, which in
particular excludes hard jumps at $\pm k_{\mathrm{IR},\lambda}$.  If
\begin{equation}
  G_{\lambda,<}(x)
  :=\frac1{\sqrt{2\pi}}
  \int_{-k_{\mathrm{IR},\lambda}}^{k_{\mathrm{IR},\lambda}}g_\lambda(k)e^{-ikx}\dd k,
  \label{eq:linear-infrared-spatial-profile}
\end{equation}
then Plancherel's theorem gives
\begin{equation}
  \frac{\lVert xG_{\lambda,<}\rVert_2}{\lVert G_{\lambda,<}\rVert_2}
  =
  \frac{
    \lVert\partial_k g_\lambda\rVert_{L^2(-k_{\mathrm{IR},\lambda},k_{\mathrm{IR},\lambda})}
  }{
    \lVert g_\lambda\rVert_{L^2(-k_{\mathrm{IR},\lambda},k_{\mathrm{IR},\lambda})}
  }.
  \label{eq:linear-infrared-Plancherel}
\end{equation}
Hence a bound by $L_{\mathrm{IR},\lambda}$ on the left-hand side is a spatial-extent condition
for the infrared part of each transmission-line coupling.  Its required family uniformity is
the dimensionless condition \eqref{eq:linear-B2-uniformity}; the spatial profile of the
higher-frequency part remains unrestricted.

\subsection{Quiet-interval square-root bound}
\label{sec:quiet-interval-bound}

We now establish the quiet-interval square-root bound of Theorem~\ref{thm:quiet-relative-speed}.
Two preliminary lemmas bound the population-change rate during the gate and show how
Assumption~\assumptionlink{B1} controls the coupling entering that rate.

  For the population-rate estimate, use $P_{Q,\lambda}(s)$ from
  \eqref{eq:paired-gate-projector} and define
  \begin{align}
    A_\lambda^{\picsup{S}}(s)
    &:=[P_{Q,\lambda}(s),\sigma_\lambda],
    \label{eq:paired-system-commutator}\\
    r_\lambda(s)^2
    &:=\Tr\!\left[\rho_{Q,\lambda}^{\picsup{S}}(s)
      A_\lambda^{\picsup{S}}(s)A_\lambda^{\picsup{S}}(s)^\dagger\right],
    \qquad 0\le r_\lambda(s)\le1.
    \label{eq:paired-sensitivity}
  \end{align}
  The upper bound follows from Lemma~\ref{lem:projection-commutator} and
  $\|\sigma_\lambda\|=1$.

  \begin{lemma}[State-weighted population-rate estimate]
  \label{lem:target-transition-bounds-chi}
  Under the model assumptions and \assumptionlink{A1},
  \begin{equation}
    \left|\frac{\dd}{\dd s}p_{+,\lambda}(s)\right|
    \le2r_\lambda(s)\mathcal J_\lambda(s)
    \qquad\text{for almost every }s\in I_\lambda.
    \label{eq:paired-gate-rate}
  \end{equation}
  If \assumptionlink{A2} also holds, then
  \begin{equation}
    0<\delta_0\le\delta_\lambda
    \le2\int_{I_\lambda}r_\lambda(s)\mathcal J_\lambda(s)\dd s.
    \label{eq:paired-integrated-rate}
  \end{equation}
  \end{lemma}

  \begin{proof}
  Using the memberwise interaction-picture evolution \eqref{eq:sigma-t} for
  $\sigma_\lambda(s)$, define the commutator
  \begin{equation}
    A_\lambda^{\picsup{I}}(s):=[P_{+,\lambda},\sigma_\lambda(s)].
    \label{eq:paired-interaction-commutator}
  \end{equation}
  Unitary conjugation of \eqref{eq:paired-system-commutator} and
  \eqref{eq:paired-sensitivity} gives
  \begin{align}
    A_\lambda^{\picsup{S}}(s)
    &=e^{-iH_{Q,\lambda}s}A_\lambda^{\picsup{I}}(s)e^{iH_{Q,\lambda}s},
    \label{eq:paired-A-picture}\\
    r_\lambda(s)^2
    &=\Tr\!\left[\rho_{Q,\lambda}^{\picsup{I}}(s)
      A_\lambda^{\picsup{I}}(s)A_\lambda^{\picsup{I}}(s)^\dagger\right].
    \label{eq:paired-sensitivity-picture}
  \end{align}
  Put $X_\lambda(s):=A_\lambda^{\picsup{I}}(s)\otimes B_\lambda(s)$ and
  $z_\lambda(s):=\Tr[\rho_\lambda^{\picsup{I}}(s)X_\lambda(s)]$.
  Write $\rho_{F,\lambda}^{\picsup{I}}(s):=\Tr_Q\rho_\lambda^{\picsup{I}}(s)$.
  Self-adjointness of $P_{+,\lambda}$ gives
  \[
    [P_{+,\lambda}\otimes\id_F,H_{\mathrm{int},\lambda}(s)]
    =X_\lambda(s)-X_\lambda(s)^\dagger.
  \]
  For the population $p_{+,\lambda}(s)$ defined in \eqref{eq:paired-population},
  the exact von Neumann equation therefore yields
  $\frac{\dd}{\dd s}p_{+,\lambda}(s)=2\operatorname{Im}z_\lambda(s)$ almost everywhere.
  To bound $z_\lambda$, use Cauchy--Schwarz in the state
  $\rho=\rho_\lambda^{\picsup{I}}(s)$:
  $|\Tr[\rho M^\dagger N]|^2\le\Tr[\rho M^\dagger M]\Tr[\rho N^\dagger N]$.
  Choosing $M=A_\lambda^{\picsup{I}}(s)^\dagger\otimes\id_F$ and
  $N=\id_Q\otimes B_\lambda(s)$ gives
  \begin{align}
    |z_\lambda(s)|^2
    &\le
    \Tr\!\left[\rho_{Q,\lambda}^{\picsup{I}}(s)
      A_\lambda^{\picsup{I}}(s)A_\lambda^{\picsup{I}}(s)^\dagger\right]
    \Tr\!\left[\rho_{F,\lambda}^{\picsup{I}}(s)
      B_\lambda^\dagger(s)B_\lambda(s)\right]
    \notag\\
    &=r_\lambda(s)^2\mathcal J_\lambda^2(s).
    \label{eq:paired-state-CS}
  \end{align}
  The right-hand side is finite since $r_\lambda(s)\le1$ and
  $\mathcal J_\lambda^2(s)\le C_{g_\lambda}C_{\mathrm{fld}}E_\lambda<\infty$
  by Lemma~\ref{lem:B-energy} and \assumptionlink{A1}.
  Since $|\operatorname{Im}z_\lambda(s)|\le|z_\lambda(s)|$, this gives
  $\left|\frac{\dd}{\dd s}p_{+,\lambda}(s)\right|
  \le2r_\lambda(s)\mathcal J_\lambda(s)$.
  This concludes the proof of \eqref{eq:paired-gate-rate}.

  For \eqref{eq:paired-integrated-rate}, \assumptionlink{A2} gives
  $\delta_0\le\delta_\lambda$, and \eqref{eq:family-endpoint-witness} identifies
  $\delta_\lambda$ with the net increase of $p_{+,\lambda}$ during the gate.
  The time regularity assumed in Section~\ref{sec:model} ensures that this increase
  equals the integral of its derivative.  Therefore,
  \begin{align*}
    \delta_0\le\delta_\lambda
    &=p_{+,\lambda}(t_{\mathrm{out},\lambda})
      -p_{+,\lambda}(t_{\mathrm{in},\lambda})
      =\int_{I_\lambda}\frac{\dd}{\dd s}p_{+,\lambda}(s)\dd s\\
    &\le\int_{I_\lambda}\left|\frac{\dd}{\dd s}p_{+,\lambda}(s)\right|\dd s
      \le2\int_{I_\lambda}r_\lambda(s)\mathcal J_\lambda(s)\dd s,
  \end{align*}
  where the final inequality is \eqref{eq:paired-gate-rate}.
  The last integral is finite: \eqref{eq:paired-sensitivity} gives $r_\lambda(s)\le1$,
  and the second-moment bound above gives
  $\mathcal J_\lambda(s)\le\sqrt{C_{g_\lambda}C_{\mathrm{fld}}E_\lambda}$
  throughout the finite gate interval.
  \end{proof}

\begin{lemma}[Quiet-time activations control the gate-state coupling]
\label{lem:quiet-controls-Gamma}
Under Assumption~\assumptionlink{B1}, for every $\lambda$ and every $s\in I_\lambda$,
\begin{equation}
  T_\lambda^2\Gamma_{g_\lambda}r_\lambda(s)^2
  \le K_{\mathrm q}
  :=\bigl(\sqrt{K_D}+\sqrt{K_J}\bigr)^2.
  \label{eq:quiet-Gamma-bound}
\end{equation}
\end{lemma}

\begin{proof}
Fix a member $\lambda$, a gate time $s\in I_\lambda$, and a quiet time
$t\in\mathcal W_\lambda$ at which $\mathcal J_\lambda(t)$ and $\mathcal D_\lambda(t)$
are finite.  Using $A_\lambda^{\picsup{S}}(s)$ and $r_\lambda(s)$ from
\eqref{eq:paired-system-commutator} and \eqref{eq:paired-sensitivity}, abbreviate
\[
  A:=A_\lambda^{\picsup{S}}(s),\qquad B:=B_\lambda,\qquad
  \mathcal J:=\mathcal J_\lambda(t),\qquad
  \Gamma:=\Gamma_{g_\lambda},\qquad r:=r_\lambda(s).
\]
Also abbreviate $\ket{\Xi_s}:=\ket{\Xi_\lambda^{\picsup{S}}(s)}$ and
$\ket{\Phi_t}:=\ket{\Phi_\lambda^{\picsup{S}}(t)}$.
Operators below act on their designated system or
field factor and as identities on the reference systems; tensor factors may be reordered.
The definitions give
\begin{equation}
  \|B\ket{\Phi_t}\|_2=\mathcal J.
  \label{eq:quiet-annihilation-vector-norm}
\end{equation}
Using $[B,B^\dagger]=\Gamma\id_F$ and normalization,
\begin{align}
  \|B^\dagger\ket{\Phi_t}\|_2^2
  &=\bigl\lVert B\ket{\Phi_t}\bigr\rVert_2^2+\Gamma
   =\mathcal J^2+\Gamma.
  \label{eq:quiet-creation-vector-norm}
\end{align}
On the system factor,
\begin{equation}
  \|A^\dagger\ket{\Xi_s}\|_2=r,\qquad
  \|A\ket{\Xi_s}\|_2\le\|A\|\le1,
  \label{eq:paired-system-vector-norms}
\end{equation}
where the last inequality is Lemma~\ref{lem:projection-commutator}.
Using the definitions \eqref{eq:paired-commutator}, \eqref{eq:family-V}, and
\eqref{eq:paired-system-commutator}, and writing $P:=P_{Q,\lambda}(s)=P^\dagger$, we have
\begin{align}
  C_{QF,\lambda}(s)
  &=[P\otimes\id_F,V_\lambda]\notag\\
  &=[P\otimes\id_F,\sigma_\lambda\otimes B
     +\sigma_\lambda^\dagger\otimes B^\dagger]\notag\\
  &=[P,\sigma_\lambda]\otimes B
     +[P,\sigma_\lambda^\dagger]\otimes B^\dagger\notag\\
  &=[P,\sigma_\lambda]\otimes B
     -[P,\sigma_\lambda]^\dagger\otimes B^\dagger\notag\\
  &=A\otimes B-A^\dagger\otimes B^\dagger.
  \label{eq:paired-adjoint-commutator}
\end{align}
Using \eqref{eq:quiet-annihilation-vector-norm}--\eqref{eq:paired-system-vector-norms},
we find that the two terms in \eqref{eq:paired-adjoint-commutator}, acting on
$\ket{\Xi_s}\otimes\ket{\Phi_t}$, have norms
$\|A\ket{\Xi_s}\|_2\mathcal J$ and $r\sqrt{\mathcal J^2+\Gamma}$.
By the reverse triangle inequality and the paired supremum in \eqref{eq:B1-disturbance-action},
\begin{align}
  \mathcal D_\lambda(t)
  &\ge\|C_{QF,\lambda}(s)(\ket{\Xi_s}\otimes\ket{\Phi_t})\|_2
   \notag\\
  &\ge r\sqrt{\mathcal J^2+\Gamma}
        -\|A\ket{\Xi_s}\|_2\mathcal J
   \notag\\
  &\ge r\sqrt{\mathcal J^2+\Gamma}-\mathcal J.
  \label{eq:paired-reverse-triangle}
\end{align}
Since $\mathcal J^2\ge0$, this implies
\begin{equation}
  \sqrt{\Gamma_{g_\lambda}}\,r_\lambda(s)
  \le\mathcal D_\lambda(t)+\mathcal J_\lambda(t).
  \label{eq:paired-pointwise-coupling}
\end{equation}
No division by $r_\lambda(s)$ is used, so the estimate includes $r_\lambda(s)=0$.

Keep $s$ fixed and square \eqref{eq:paired-pointwise-coupling}; its left-hand side
is independent of the quiet time $t$.  Averaging over $\mathcal W_\lambda$ yields
\begin{align}
  T_\lambda^2\Gamma_{g_\lambda}r_\lambda(s)^2
  &\le\frac{T_\lambda^2}{|\mathcal W_\lambda|}
    \int_{\mathcal W_\lambda}
    \bigl(\mathcal D_\lambda(t)+\mathcal J_\lambda(t)\bigr)^2\dd t
    \notag\\
  &\le K_D+K_J
   +\frac{2T_\lambda^2}{|\mathcal W_\lambda|}
    \left(\int_{\mathcal W_\lambda}\mathcal D_\lambda^2(t)\dd t\right)^{1/2}
    \left(\int_{\mathcal W_\lambda}\mathcal J_\lambda^2(t)\dd t\right)^{1/2}
    \notag\\
  &\le K_D+K_J+2\sqrt{K_DK_J}
   =K_{\mathrm q}.
  \label{eq:quiet-Gamma-proof-finish}
\end{align}
The second line uses \eqref{eq:B1-J-bound}, \eqref{eq:B1-D-bound}, and
Cauchy--Schwarz for the mixed term.  The constants do not depend on $s$,
which proves \eqref{eq:quiet-Gamma-bound} for every gate time.
\end{proof}

The key ingredients of the proof have the following roles.
The creation norm
$\bigl\lVert B_\lambda^\dagger\ket{\Phi_\lambda^{\picsup{S}}(t)}\bigr\rVert_2
=\sqrt{\mathcal J_\lambda^2(t)+\Gamma_{g_\lambda}}$
in \eqref{eq:quiet-creation-vector-norm} supplies the term needed to control the coupling
weighted by the actual gate state,
$\Gamma_{g_\lambda}r_\lambda(s)^2
=\Gamma_{g_\lambda}\Tr[\rho_{Q,\lambda}^{\picsup{S}}(s)
A_\lambda^{\picsup{S}}(s)A_\lambda^{\picsup{S}}(s)^\dagger]$.
Retaining this state dependence is important: the population-rate bound
\eqref{eq:paired-gate-rate}, combined with \assumptionlink{B2} and \assumptionlink{A1},
depends on $\sqrt{\Gamma_{g_\lambda}}\,r_\lambda(s)$, precisely the quantity controlled by
\eqref{eq:paired-pointwise-coupling}.
A bound on $\Gamma_{g_\lambda}$ alone would impose an unnecessarily stronger restriction.
The annihilation norm
$\bigl\lVert B_\lambda\ket{\Phi_\lambda^{\picsup{S}}(t)}\bigr\rVert_2
=\mathcal J_\lambda(t)$ cannot provide this control by itself, since it vanishes on the
vacuum regardless of the coupling strength.
Finally, the system norms
$\bigl\lVert A_\lambda^{\picsup{S}}(s)^\dagger
\ket{\Xi_\lambda^{\picsup{S}}(s)}\bigr\rVert_2=r_\lambda(s)$ and
$\bigl\lVert A_\lambda^{\picsup{S}}(s)
\ket{\Xi_\lambda^{\picsup{S}}(s)}\bigr\rVert_2\le1$
need not coincide; using the latter bound produces the additional
$\mathcal J_\lambda(t)$ term in \eqref{eq:paired-pointwise-coupling}, and hence the
$K_J$ contribution to $K_{\mathrm q}$.

\begin{theorem}[Quiet-interval square-root bound]
\label{thm:quiet-relative-speed}
Consider a family satisfying the model assumptions of Section~\ref{sec:model}.
Assume \assumptionlink{A1}, \assumptionlink{A2},
\assumptionlink{B1} and \assumptionlink{B2}. 
Then
\begin{equation}
  \omega_{\mathrm{gate},\lambda}T_\lambda
  \le c_{\mathrm{quiet}}\sqrt{T_\lambda E_\lambda},
  \label{eq:quiet-relative-result}
\end{equation}
where
\begin{equation}
  c_{\mathrm{quiet}}
  :=\frac{4\pi}{\delta_0}
    \sqrt{C_{\mathrm{fld}}K_\omega}\,
    \bigl(\sqrt{K_D}+\sqrt{K_J}\bigr).
  \label{eq:quiet-relative-constant}
\end{equation}
The dimensionless constant $c_{\mathrm{quiet}}$ is $\lambda$-independent and defined
in terms of the constants appearing in Assumptions \assumptionlink{A1},
\assumptionlink{A2}, \assumptionlink{B1}, and \assumptionlink{B2}.
\end{theorem}

The statement is covariant under a global fast-forwarding: both
$\omega_{\mathrm{gate},\lambda}T_\lambda$ and $T_\lambda E_\lambda$ are unchanged by
$H_\lambda\mapsto\alpha_\lambda H_\lambda$ when the reference rule
\eqref{eq:reference-time-rule} is used.  The paired states and projectors are evaluated at
the correspondingly rescaled gate times; the quiet-time activations scale by $\alpha_\lambda$,
so their dimensionless mean-square bounds are unchanged.
The theorem implies an absolute square-root law whenever
the times $T_\lambda$ are chosen to be $\lambda$-independent.

\begin{proof}
The proof connects the quiet-field tests to the actual correlated gate dynamics through
the same $r_\lambda(s)$ defined in \eqref{eq:paired-sensitivity}.

\medskip
\noindent\textit{Step 1: quiet intervals control the coupling on every paired gate state.}
Lemma~\ref{lem:quiet-controls-Gamma} gives, for every $s\in I_\lambda$,
\begin{equation}
  T_\lambda^2\Gamma_{g_\lambda}r_\lambda(s)^2\le K_{\mathrm q}.
  \label{eq:quiet-theorem-step-one}
\end{equation}

\medskip
\noindent\textit{Step 2: inverse-frequency compatibility controls the energy-weighted coupling.}
By \assumptionlink{B2}, for the same gate time $s$,
\begin{align}
  T_\lambda C_{g_\lambda}r_\lambda(s)^2
  &\le K_\omega T_\lambda^2\Gamma_{g_\lambda}r_\lambda(s)^2
   \notag\\
  &\le K_\omega K_{\mathrm q}.
  \label{eq:quiet-uniform-Cg}
\end{align}

\medskip
\noindent\textit{Step 3: bound the exact population rate and integrate.}
Lemma~\ref{lem:target-transition-bounds-chi} gives the state-weighted rate estimate
\eqref{eq:paired-gate-rate}.  By Lemma~\ref{lem:B-energy} and \assumptionlink{A1},
\begin{equation}
  \mathcal J_\lambda^2(s)
  \le C_{g_\lambda}\Tr[\rho_{F,\lambda}^{\picsup{S}}(s)H_{F,\lambda}]
  \le C_{g_\lambda}C_{\mathrm{fld}}E_\lambda.
  \label{eq:quiet-theorem-energy-bound}
\end{equation}
Combining Lemma~\ref{lem:target-transition-bounds-chi} with \eqref{eq:quiet-theorem-energy-bound}, followed by applying \eqref{eq:quiet-uniform-Cg} therefore gives
\begin{align}
  \left|\frac{\dd}{\dd s}p_{+,\lambda}(s)\right|
  &\le2r_\lambda(s)\sqrt{C_{g_\lambda}C_{\mathrm{fld}}E_\lambda}
   \notag\\
  &\le2\sqrt{\frac{C_{\mathrm{fld}}K_\omega K_{\mathrm q}E_\lambda}{T_\lambda}}
  \quad\text{for almost every }s\in I_\lambda.
  \label{eq:quiet-theorem-rate}
\end{align}

The time regularity assumed in Section~\ref{sec:model},
\eqref{eq:family-endpoint-witness}, and \assumptionlink{A2} now give
\begin{align}
  \delta_0
  &\le\delta_\lambda
   =p_{+,\lambda}(t_{\mathrm{out},\lambda})
    -p_{+,\lambda}(t_{\mathrm{in},\lambda})
   \notag\\
  &\le\int_{I_\lambda}\left|\frac{\dd}{\dd s}p_{+,\lambda}(s)\right|\dd s
   \notag\\
  &\le2\Delta t_\lambda
    \sqrt{\frac{C_{\mathrm{fld}}K_\omega K_{\mathrm q}E_\lambda}{T_\lambda}}.
  \label{eq:quiet-theorem-integrated-rate}
\end{align}
Using $\omega_{\mathrm{gate},\lambda}=2\pi/\Delta t_\lambda$ yields
\begin{align}
  \omega_{\mathrm{gate},\lambda}T_\lambda
  &\le\frac{4\pi}{\delta_0}
    \sqrt{C_{\mathrm{fld}}K_\omega K_{\mathrm q}}\sqrt{T_\lambda E_\lambda}
    \notag\\
  &=c_{\mathrm{quiet}}\sqrt{T_\lambda E_\lambda}.
  \label{eq:quiet-theorem-finish}
\end{align}
This proves \eqref{eq:quiet-relative-result}--\eqref{eq:quiet-relative-constant}.
\end{proof}

\section{Uniform lower stability and near-linear speed with nonlinear interactions}
\label{sec:nonlinear-clock-comparison}

This section presents a family of finite-dimensional harmonic oscillators with nonlinear interactions that satisfies \assumptionlink{A0}--\assumptionlink{A2}
and the stronger all-projector quietness condition \assumptionlink[B1$'$]{B1prime} together
with the inverse-frequency analogue \assumptionlink[B2$'$]{B2prime}, while
achieving near-linear gate-frequency scaling with drive energy.  It thereby shows that the
square-root scaling is not fundamental beyond the linear bosonic control setting.

\subsection{\texorpdfstring{Model and fulfilment of Assumptions
\assumptionlink{A0}--\assumptionlink{A2}}{Model and fulfilment of Assumptions A0--A2}}
\label{sec:nonlinear-clock-A-assumptions}

The linear bosonic form of the interaction in Section~\ref{sec:model} is essential to
Theorem~\ref{thm:relative-speed}.  Uniform lower stability, memberwise background periods
defined by one common rule,
uniform control of the drive energy, and a nontrivial separation of $\rho_{\mathrm{out},\lambda}$ and $\rho_{\mathrm{in},\lambda}$ do not by
themselves imply a square-root speed limit for arbitrary interactions.  We demonstrate this
by briefly recalling the finite-dimensional harmonic-oscillator construction in Theorem~2 of
Ref.~\cite{Woods2024}.  In that construction the family parameter is the oscillator truncation
dimension $d$; thus $d$ serves as the counterpart of the family index $\lambda$ used for the
linear bosonic drive family.  The $d$-dependence of many quantities there was left
implicit.  Here we make it explicit by adding a subscript
$d$ or an argument $d$; for example, we write $N_g(d)$ for the number of gates.

Assign a background period $T_d>0$, denoted $T_0$ in Ref.~\cite{Woods2024}, and write
$\omega_{0,d}:=2\pi/T_d$.  Throughout the following, as in that work,
$T_d$ is held constant as $d$ varies.  The harmonic oscillator restricted to $d$ energy levels
has energy basis $\{\ket{E_n}\}_{n=0}^{d-1}$ and free Hamiltonian
\begin{equation}
  H_{C,d}
  :=\sum_{n=0}^{d-1}n\omega_{0,d}\ket{E_n}\!\!\mkern1mu\bra{E_n}.
  \label{eq:finite-clock-free-H}
\end{equation}
Its discrete Fourier basis is
\begin{equation}
  \ket{\theta_k}
  :=\frac{1}{\sqrt d}\sum_{n=0}^{d-1}
  e^{-2\pi i nk/d}\ket{E_n}.
  \label{eq:finite-clock-time-basis}
\end{equation}
For a finite gate set independent of $d$, the time-independent total Hamiltonian has the form
\begin{equation}
  H_d
  :=H_{C,d}
  +\sum_{l=1}^{N_g(d)}
  I_{M_0L}^{(l)}\otimes I_{C,d}^{(l)},
  \label{eq:finite-clock-total-H}
\end{equation}
where
\begin{equation}
  I_{C,d}^{(l)}
  :=\omega_{0,d}\sum_k
  \overline V_{0,d}^{(l)}\!\left(\frac{2\pi k}{d}\right)
  \ket{\theta_k}\!\!\mkern1mu\bra{\theta_k}.
  \label{eq:finite-clock-interaction}
\end{equation}
Here $\overline V_{0,d}^{(l)}\ge0$ is a smooth, localized, periodic potential normalized by
\begin{equation}
  \int_0^{2\pi}\overline V_{0,d}^{(l)}(x)\,\dd x=1.
  \label{eq:finite-clock-area}
\end{equation}
The operator $I_{M_0L}^{(l)}$ reads the $l$th memory cell and applies the corresponding
logical gate generator.  The logarithm of each gate is chosen with spectrum in $(0,2\pi]$;
consequently,
\begin{equation}
  0\le I_{M_0L}^{(l)}\le2\pi\id_{M_0L},
  \qquad
  I_{C,d}^{(l)}\ge0.
  \label{eq:finite-clock-positive-interactions}
\end{equation}

\begin{theorem}[Uniform stability is compatible with near-linear speed for nonlinear control]
\label{thm:nonlinear-clock-comparison}
For every fixed $\bar\varepsilon\in(0,1/6)$, the family constructed in Theorem~2 of
Ref.~\cite{Woods2024} can be chosen so that its background period is $T_d$, its
total Hamiltonians have the common ground energy
\begin{equation}
  E_{\mathrm{gs},d}=0,
  \label{eq:finite-clock-ground-energy}
\end{equation}
and its initial mean energy above the ground satisfies
\begin{equation}
  E_d
  =\frac{2\pi\widetilde n_0}{T_d}(d-1)+r_d,
  \qquad T_d r_d\longrightarrow0,
  \label{eq:finite-clock-energy}
\end{equation}
where $\widetilde n_0\in(0,1)$ is independent of $d$.  With the choice
$\widetilde n_0=1/(2\pi)$ and
\begin{equation}
  N_g(d)=\left\lfloor d^{1-\bar\varepsilon}\right\rfloor,
  \label{eq:finite-clock-number-gates}
\end{equation}
the gate frequency $f_d:=N_g(d)/T_d$ obeys
\begin{equation}
  f_d
  =\frac{1}{T_d}(T_d E_d)^{1-\bar\varepsilon}
  +O\!\left(\frac{1}{T_d}\right).
  \label{eq:finite-clock-near-linear-frequency}
\end{equation}
The construction satisfies the reference-time definition and the uniform stability,
drive-energy, and nontrivial-transition requirements in assumptions
\assumptionlink{A0}--\assumptionlink{A2}.
\end{theorem}

\begin{proof}
We verify the asserted uniform properties and then derive the frequency-energy relation.

\medskip
\noindent\textit{Uniform logical structure and regularity.}
The logical Hilbert space, gate alphabet, and logical gate set are independent of $d$.  The
free logical Hamiltonian is zero, and \eqref{eq:finite-clock-positive-interactions} gives the
$d$-independent bound $\lVert I_{M_0L}^{(l)}\rVert\le2\pi$ for every local gate generator.
This verifies the common-system-Hilbert-space and uniform-normalization conventions of the family setup.  Each Hamiltonian
is finite dimensional before the harmless common-Hilbert-space embedding, so it is bounded
and self-adjoint and all state trajectories and fixed-projector expectations are analytic in time.  The free Hamiltonian
$H_{C,d}$ is the restriction of a harmonic-oscillator Hamiltonian, with its ground energy set to zero,
to its first $d$ energy levels.

\medskip
\noindent\textit{Scale-covariant background time.}
Equation~\eqref{eq:finite-clock-free-H} gives
\begin{equation}
  e^{-iH_{C,d}T_d}
  =\sum_{n=0}^{d-1}e^{-2\pi i n}\ket{E_n}\!\!\mkern1mu\bra{E_n}
  =\id_{C,d}
  \label{eq:finite-clock-period}
\end{equation}
for every $d$.  Thus $T_d$ is the recurrence period assigned to member $d$.  The same
recurrence-period rule is used throughout the family, and
$H_d\mapsto\alpha_dH_d$ sends $T_d\mapsto T_d/\alpha_d$; a global fast-forward therefore
leaves all dimensionless comparisons below invariant.

Importantly, the achievability construction of Ref.~\cite{Woods2024} does not
use this rescaling freedom.  Instead, $E_d$ is increased through the $d$-dependent quasi-ideal
initial oscillator state.  The resulting near-linear energy--frequency scaling is
therefore not produced by the trivial global rescaling
$H_d\mapsto\alpha_dH_d$ and the accompanying contraction
$T_d\mapsto T_d/\alpha_d$.

\medskip
\noindent\textit{Scale-covariant uniform lower stability.}
Equations~\eqref{eq:finite-clock-free-H} and
\eqref{eq:finite-clock-positive-interactions} imply
\begin{equation}
  H_d\ge0
  \qquad\text{for every }d.
  \label{eq:finite-clock-uniform-lower-bound}
\end{equation}
If the computational block has no zero eigenvector, the construction in
Ref.~\cite{Woods2024} adjoins one decoupled vector
$\ket{\mathrm{ground}}$ and defines $H_d\ket{\mathrm{ground}}=0$.  This leaves the
computational dynamics unchanged and proves \eqref{eq:finite-clock-ground-energy}.  Hence the
family has $E_{\mathrm{gs},d}=0$.  Since its logical free Hamiltonian is also zero,
$T_d[\lambda_{\max}(H_{Q,d})-E_{\mathrm{gs},d}]=0$; hence the scale-covariant
lower-stability condition \eqref{eq:A0-reference-time-bound} holds with $K_0=0$.

\medskip
\noindent\textit{Uniform drive-energy control.}
Positivity of every interaction term gives the operator inequality
\begin{equation}
  0\le H_{C,d}\le H_d.
  \label{eq:finite-clock-free-below-total}
\end{equation}
Let $\rho_d(t):=e^{-iH_dt}\rho_{0,d}e^{iH_dt}$.  Taking expectations in
\eqref{eq:finite-clock-free-below-total} and using conservation of total energy yields
\begin{align}
  \Tr[\rho_d(t)H_{C,d}]
  &\le\Tr[\rho_d(t)H_d]
  \notag\\
  &=\Tr[\rho_{0,d}H_d]
  =E_d.
  \label{eq:finite-clock-energy-control}
\end{align}
The last equality uses \eqref{eq:finite-clock-ground-energy}.  Thus the analogue of
\assumptionlink{A1} holds for all times with the uniform constant $C_{\mathrm{fld}}=1$.

\medskip
\noindent\textit{Uniformly nontrivial transitions.}
Choose from the $d$-independent gate set a gate and an input state whose ideal initial and final system states have
trace distance $\delta_{\mathrm{tr,id}}>0$; an orthogonalizing gate gives
$\delta_{\mathrm{tr,id}}=1$.
Theorem~2 of Ref.~\cite{Woods2024} bounds the trace-distance error at the gate
before/after times $t_j$ by a quantity $\epsilon_d$ that tends to zero faster than any inverse
polynomial in the energy.  Contractivity under partial trace and the triangle inequality
therefore give, for the actual trace distance
$\delta_{\mathrm{tr},d}:=D(\rho_{\mathrm{out},d},\rho_{\mathrm{in},d})$,
\begin{equation}
  \delta_{\mathrm{tr},d}\ge\delta_{\mathrm{tr,id}}-2\epsilon_d.
  \label{eq:finite-clock-endpoint-distance}
\end{equation}
For all sufficiently large $d$, this gives
$\delta_{\mathrm{tr},d}\ge\delta_{\mathrm{tr,id}}/2=:\delta_0>0$,
the uniform separation of $\rho_{\mathrm{out},d}$ and $\rho_{\mathrm{in},d}$ required by \assumptionlink{A2}.
If these values are used in
Definition~\ref{def:gate-interval}, then its duration is $T_d/N_g(d)$ and
\begin{equation}
  \omega_{\mathrm{gate},d}=2\pi f_d.
  \label{eq:finite-clock-two-frequency-conventions}
\end{equation}

\medskip
\noindent\textit{Energy and speed.}
For the quasi-ideal initial oscillator state used in the construction, Appendix~C of
Ref.~\cite{Woods2024} proves \eqref{eq:finite-clock-energy}; the interaction-energy
contribution $T_d r_d$ vanishes faster than every inverse polynomial.  Setting
$\widetilde n_0=1/(2\pi)$ gives
\begin{equation}
  T_d E_d=d-1+o(1).
  \label{eq:finite-clock-dimension-energy}
\end{equation}
Combining this identity with \eqref{eq:finite-clock-number-gates} gives
\begin{align}
  f_d
  &=\frac{\lfloor d^{1-\bar\varepsilon}\rfloor}{T_d}
  \notag\\
  &=\frac{1}{T_d}(T_d E_d)^{1-\bar\varepsilon}
  +O\!\left(\frac{1}{T_d}\right),
  \label{eq:finite-clock-frequency-calculation}
\end{align}
which is \eqref{eq:finite-clock-near-linear-frequency}.  For every fixed
$\bar\varepsilon>0$ the proven exponent is $1-\bar\varepsilon$, rather than exactly one, but
it can be chosen arbitrarily close to one.

There is no contradiction with Theorem~\ref{thm:relative-speed}.  The decisive model
hypothesis used there is the linear interaction
$\sigma\otimes B_\lambda+\sigma^\dagger\otimes B^\dagger_\lambda$, with $B_\lambda$ linear in annihilation
operators.  By contrast, $I_{C,d}^{(l)}$ in \eqref{eq:finite-clock-interaction} is a localized
nonlinear function of the phase operator conjugate to $H_{C,d}$.  Its normalization
\eqref{eq:finite-clock-area} fixes the integrated rotation, while its profile may become
narrower as $d$ increases.  Positivity then preserves the uniform lower bound even as the
gate is localized into a shorter part of the common background cycle.  The square-root proof
therefore does not extend to this nonlinear interaction, and the near-linear family lies
outside the theorem's stated class.
\end{proof}

\subsection{\texorpdfstring{Analogues of Assumptions~\assumptionlink{B1} and
\assumptionlink{B2} for the nonlinear interactions of Ref.~\cite{Woods2024}}%
{Analogues of Assumptions B1 and B2 for the nonlinear interactions of the cited work}}
\label{sec:nonlinear-clock-B-prime}

The bosonic assumptions \assumptionlink{B1} and \assumptionlink{B2} cannot be transferred
literally to the truncated harmonic oscillator: 	\eqref{eq:finite-clock-interaction} is not a
linear creation-annihilation interaction of the form \eqref{eq:family-V} and there is no annihilation-field coupling $g_d$. 

The following analogues retain
the two pieces of physical content that remain meaningful.  The first tests the activation induced by the
\emph{full} nonlinear interaction in quiet windows, retaining a stronger all-projector
condition that implies the counterpart of \assumptionlink{B1}.  The second tests the inverse-frequency
moment of the part detected by free oscillator evolution.  In particular, the latter does not
remove the zero-frequency part from the Hamiltonian; that part remains present in the first
test.

Write
\begin{equation}
  \mathbb V_d
  :=H_d-H_{C,d}
  =\sum_{l=1}^{N_g(d)}I_{M_0L}^{(l)}\otimes I_{C,d}^{(l)},
  \qquad
  \varrho_{C,d}(t):=\Tr_{M_0L}[\rho_d(t)].
  \label{eq:nonlinear-total-interaction-and-clock-state}
\end{equation}
For the $j$th gate, let $t_{j,d}:=jT_d/N_g(d)$, set
\begin{equation}
  a_d:=\frac{T_d}{8N_g(d)},
  \qquad
  \mathcal W'_{j,d}
  :=[t_{j-1,d}-a_d,t_{j-1,d}]
    \cup[t_{j,d},t_{j,d}+a_d],
  \label{eq:nonlinear-quiet-windows}
\end{equation}
The autonomous trajectory is defined for all real $t$; only the oscillator phase used in the
localization argument below is read modulo $T_d$.
Define the nonlinear analog of the field activation
\begin{equation}
  \mathcal J_d^{\mathrm{nl}}(t)
  :=\sum_{l=1}^{N_g(d)}
  \sqrt{\Tr\!\left[
    \varrho_{C,d}(t)\bigl(I_{C,d}^{(l)}\bigr)^2
  \right]},
  \label{eq:nonlinear-clock-activation}
\end{equation}
and the nonlinear analog of the interaction activation, using the stronger all-projector test,
\begin{equation}
  \mathcal D_d^{\mathrm{nl}}(t)
  :=\sup_{\substack{P=P^\dagger=P^2\\ \lVert\psi\rVert_2=1}}
  \sqrt{\Tr\!\left[
    \bigl(\ket{\psi}\!\!\mkern1mu\bra{\psi}\otimes\varrho_{C,d}(t)\bigr)
    [P\otimes\id_{C,d},\mathbb V_d]^\dagger
    [P\otimes\id_{C,d},\mathbb V_d]
  \right]},
  \label{eq:nonlinear-system-changing-action}
\end{equation}
where $P$ acts on $M_0L$.

For the $j$th gate, set
$\varrho_{M_0L,d}(s):=\Tr_C\rho_d(s)$ and let $P_{+,j,d}$ be the positive spectral
projection of
$\varrho_{M_0L,d}(t_{j,d})-\varrho_{M_0L,d}(t_{j-1,d})$.
Since $H_{Q,d}=0$ in \eqref{eq:finite-clock-total-H}, the paired projector is
$P_{+,j,d}$ at every gate time.  Define
\begin{align}
  C_{j,d}^{\mathrm{nl}}
  &:=[P_{+,j,d}\otimes\id_{C,d},\mathbb V_d],
  \label{eq:nonlinear-paired-commutator}\\
  \mathcal D_{j,d}^{\mathrm{nl,pair}}(t)
  &:=\sup_{s\in[t_{j-1,d},t_{j,d}]}
  \sqrt{\Tr\!\left[
    \bigl(\varrho_{M_0L,d}(s)\otimes\varrho_{C,d}(t)\bigr)
    (C_{j,d}^{\mathrm{nl}})^\dagger C_{j,d}^{\mathrm{nl}}
  \right]}.
  \label{eq:nonlinear-paired-action}
\end{align}
For each $s$, a spectral decomposition of $\varrho_{M_0L,d}(s)$ expresses the
squared norm as a convex combination of the pure-state tests in
\eqref{eq:nonlinear-system-changing-action}.  Taking the supremum over $s$ gives
\begin{equation}
  \mathcal D_{j,d}^{\mathrm{nl,pair}}(t)
  \le\mathcal D_d^{\mathrm{nl}}(t).
  \label{eq:nonlinear-paired-dominance}
\end{equation}
Thus the uniform quiet-set mean-square bounds on $\mathcal J_d^{\mathrm{nl}}$ and
$\mathcal D_d^{\mathrm{nl}}$ below imply the same bounds with
$\mathcal D_{j,d}^{\mathrm{nl,pair}}$ in place of the latter: square
\eqref{eq:nonlinear-paired-dominance} and integrate over $\mathcal W'_{j,d}$.

\begin{description}
  \item[\textbf{B1$'$.}] \phantomsection\label{ass:B1prime}%
  \textbf{Uniform nonlinear quiet-time activation.}
  The family satisfies \assumptionlink[B1$'$]{B1prime} when there are $d$-independent
  constants $K_J^{\mathrm{nl}}<\infty$ and $K_D^{\mathrm{nl}}<\infty$ such that
  \begin{align}
    \frac{T_d^2}{|\mathcal W'_{j,d}|}
    \int_{\mathcal W'_{j,d}}\mathcal J_d^{\mathrm{nl},2}(t)\dd t
    &\le K_J^{\mathrm{nl}},
    \label{eq:B1prime-J}\\
    \frac{T_d^2}{|\mathcal W'_{j,d}|}
    \int_{\mathcal W'_{j,d}}\mathcal D_d^{\mathrm{nl},2}(t)\dd t
    &\le K_D^{\mathrm{nl}}.
    \label{eq:B1prime-D}
  \end{align}
  This is a stronger sufficient multi-channel finite-dimensional oscillator counterpart of
  \assumptionlink{B1}, by \eqref{eq:nonlinear-paired-dominance}: it bounds interaction activation
  across every system partition and unit system vector, thereby including all paired gate-state tests.
  The field state is still the actual reduced oscillator state; an
  expectation that merely cancels is insufficient.

  \item[\textbf{B2$'$.}] \phantomsection\label{ass:B2prime}%
  \textbf{Uniform inverse-frequency compatibility of the free-evolution derivative.}
  Let $\Pi_{n,d}:=\ket{E_n}\!\!\mkern1mu\bra{E_n}$ and, for an oscillator operator $X$, define, for a Bohr frequency $\nu$,  its Bohr component 
  by
  \begin{equation}
    X[\nu]
    :=\sum_{\substack{0\le m,n\le d-1\\(m-n)\omega_{0,d}=\nu}}
      \Pi_{m,d}X\Pi_{n,d}.
    \label{eq:finite-clock-Bohr-component}
  \end{equation}
  Introduce the free-evolution derivative
  \begin{equation}
    \begin{aligned}
      I_{C,d}^{(l)}(t)
      &:=e^{iH_{C,d}t}I_{C,d}^{(l)}e^{-iH_{C,d}t},\\
      \dot I_{C,d}^{(l)}
      &:=\left.\frac{\dd}{\dd t}I_{C,d}^{(l)}(t)\right|_{t=0}
      =i[H_{C,d},I_{C,d}^{(l)}].
    \end{aligned}
  \label{eq:finite-clock-translation-derivative}
  \end{equation}
  Define the normalized Hilbert--Schmidt norm
  $\lVert X\rVert_{\mathrm{HS},d}^2:=d^{-1}\Tr(X^\dagger X)$. (Derivatives of $I_{C,d}^{(l)}(t)$ at later times are the same, up to a unitary transformation for which the Hilbert--Schmidt norm is invariant.)  Its discrete unweighted and
  inverse-frequency moments are
  \begin{align}
    \Gamma_{\partial,d}
    &:=\sum_{l=1}^{N_g(d)}\sum_{\nu\ne0}
      \bigl\lVert\dot I_{C,d}^{(l)}[\nu]\bigr\rVert_{\mathrm{HS},d}^2,
    \label{eq:B2prime-Gamma}\\
    C_{\partial,d}
    &:=\sum_{l=1}^{N_g(d)}\sum_{\nu\ne0}
      \frac{\bigl\lVert\dot I_{C,d}^{(l)}[\nu]\bigr\rVert_{\mathrm{HS},d}^2}{|\nu|}.
    \label{eq:B2prime-C}
  \end{align}
  Note that $\bigl\lVert\dot I_{C,d}^{(l)}[\nu]\bigr\rVert_{\mathrm{HS},d}=0$ for $\nu=0$, hence the summation over $\nu$ in \eqref{eq:B2prime-C} is taken over its support. The summation over $\nu$ in Eqs.~\eqref{eq:B2prime-Gamma} and \eqref{eq:B2prime-C} is analogous to the integration over $k$ in \eqref{eq:quiet-Gamma} and \eqref{eq:family-Cg}, while the summation over $l$ accounts for the $N_g(d)$ localized interaction channels, in contrast to the single coupling channel $B_\lambda$ in the linear bosonic drive model.  The family satisfies \assumptionlink[B2$'$]{B2prime} when a $d$-independent dimensionless
  constant $K_\omega^{\mathrm{nl}}<\infty$ obeys
  \begin{equation}
    C_{\partial,d}
    \le K_\omega^{\mathrm{nl}}T_d \Gamma_{\partial,d}
    \qquad\text{for every }d.
    \label{eq:B2prime-inverse-frequency}
  \end{equation}
  This is a discrete spectral-density condition for the part of the coupling that varies under
  free oscillator evolution.  The definition through the commutator is essential: it annihilates
  the zero-Bohr-frequency component without suggesting that an additive shift of the oscillator
  energy could open a Bohr-frequency gap.
\end{description}

\begin{lemma}[The truncated oscillator with nonlinear interactions satisfies \texorpdfstring{\assumptionlink[B1$'$]{B1prime}
and \assumptionlink[B2$'$]{B2prime}}{B1' and B2'}]
\label{lem:finite-clock-B-prime}
For the localized-potential construction recalled above, the selected nontrivial gate interval
satisfies \assumptionlink[B1$'$]{B1prime}, and the whole family satisfies
\assumptionlink[B2$'$]{B2prime} with
\begin{equation}
  K_\omega^{\mathrm{nl}}=\frac{1}{2\pi}.
  \label{eq:B2prime-constant}
\end{equation}
\end{lemma}

\begin{proof}
We first prove the quiet-time activation bounds.  In the construction of
Ref.~\cite{Woods2024}, the center of the $l$th periodic potential is [its Eq.~(C.26)]
\begin{equation}
  x_{0,d}^{(l)}=\frac{2\pi(l-1/2)}{N_g(d)}.
  \label{eq:finite-clock-potential-centers}
\end{equation}
Free oscillator evolution translates the quasi-ideal wave packet at angular speed $2\pi/T_d$.
Denoting its center by $\theta_{c,d}(t)$, Eq.~\eqref{eq:nonlinear-quiet-windows} therefore gives,
for every $t\in \mathcal W'_{j,d}$,
\begin{equation}
  \min_{1\le l\le N_g(d)}
  \operatorname{dist}_{\mathbb S^1}\!\left(
    \theta_{c,d}(t),x_{0,d}^{(l)}
  \right)
  \ge \frac{3\pi}{4N_g(d)}
  =\frac{3}{8}\frac{2\pi}{N_g(d)},
  \label{eq:finite-clock-quiet-separation}
\end{equation}
where
$\operatorname{dist}_{\mathbb S^1}(\theta,\phi)
:=\min_{m\in\mathbb Z}|\theta-\phi+2\pi m|\in[0,\pi]$
is the shortest distance between the angles modulo $2\pi$.
Appendix~C of
Ref.~\cite{Woods2024} chooses the packet width $\sigma_d$ and potential concentration
$n_d$ so that
\begin{equation}
  \frac{d}{\sigma_dN_g(d)}\longrightarrow\infty,
  \qquad
  \frac{n_d}{N_g(d)}\longrightarrow\infty.
  \label{eq:finite-clock-localization-scales}
\end{equation}
The estimates in Ref.~\cite{Woods2024} require only that the packet center be separated from
each potential center by one quarter of the neighboring-center spacing.  The stronger
three-eighths separation in Eq.~\eqref{eq:finite-clock-quiet-separation} therefore leaves its
Gaussian- and potential-tail estimates [Lemma~C.2 and Eqs.~(C.37)--(C.40)] and the corresponding
near/far decomposition [Eqs.~(C.132)--(C.142)] valid uniformly in $j$ and $d$.
Applying the same decomposition to the squared potential gives
the state-weighted squared norms in \eqref{eq:nonlinear-clock-activation}; its additional
polynomial prefactor is absorbed by the parameter estimates in Eqs.~(C.123)--(C.130).  The
wave-packet propagation estimate in Eq.~(C.51) and the iterative estimate in Lemma~C.1 control the exact
trajectory between times $\{t_j\}_j$.  Consequently, there is a sequence $\delta_d\to0$, independent
of $j$, such that
\begin{equation}
  T_d\sup_{t\in \mathcal W'_{j,d}}
  \mathcal J_d^{\mathrm{nl}}(t)
  \le\delta_d.
  \label{eq:finite-clock-quiet-localization}
\end{equation}
This estimate concerns the sum over all $N_g(d)$ interaction terms; it is not merely a
pointwise statement about a single potential.

For any projection $P$, equation~\eqref{eq:finite-clock-positive-interactions} gives
$\lVert[P,I_{M_0L}^{(l)}]\rVert\le2\lVert I_{M_0L}^{(l)}\rVert\le4\pi$.  Let
$\ket{\Phi_{C,d}(t)}$ be any purification of $\varrho_{C,d}(t)$.  Expanding the commutator gives
\begin{equation*}
  [P\otimes\id_{C,d},\mathbb V_d]
  =\sum_{l=1}^{N_g(d)}[P,I_{M_0L}^{(l)}]\otimes I_{C,d}^{(l)}.
\end{equation*}
Therefore, for every unit vector $\ket{\psi}$, the triangle inequality and the purification
identity yield
\begin{align*}
  &\left\|
    \bigl([P\otimes\id_{C,d},\mathbb V_d]\otimes\id_R\bigr)
    \bigl(\ket{\psi}\otimes\ket{\Phi_{C,d}(t)}\bigr)
  \right\|_2
  \\
  &\quad\le\sum_{l=1}^{N_g(d)}
    \lVert[P,I_{M_0L}^{(l)}]\rVert
    \sqrt{\Tr\!\left[\varrho_{C,d}(t)\bigl(I_{C,d}^{(l)}\bigr)^2\right]}
  \\
  &\quad\le4\pi\sum_{l=1}^{N_g(d)}
    \sqrt{\Tr\!\left[\varrho_{C,d}(t)\bigl(I_{C,d}^{(l)}\bigr)^2\right]}
  =4\pi\mathcal J_d^{\mathrm{nl}}(t).
\end{align*}
The first norm is the quantity under the supremum in
\eqref{eq:nonlinear-system-changing-action}.  Taking that supremum therefore gives
\begin{equation}
  \mathcal D_d^{\mathrm{nl}}(t)
  \le4\pi\mathcal J_d^{\mathrm{nl}}(t).
  \label{eq:finite-clock-D-by-J}
\end{equation}
Equations~\eqref{eq:finite-clock-quiet-localization} and
\eqref{eq:finite-clock-D-by-J} imply
\begin{align}
  \frac{T_d^2}{|\mathcal W'_{j,d}|}
  \int_{\mathcal W'_{j,d}}\mathcal J_d^{\mathrm{nl},2}(t)\dd t
  &\le\delta_d^2,
  \label{eq:finite-clock-B1prime-J-proof}\\
  \frac{T_d^2}{|\mathcal W'_{j,d}|}
  \int_{\mathcal W'_{j,d}}\mathcal D_d^{\mathrm{nl},2}(t)\dd t
  &\le16\pi^2\delta_d^2.
  \label{eq:finite-clock-B1prime-D-proof}
\end{align}
The two bounds above give $d$-independent $K_J^{\mathrm{nl}}$ and $K_D^{\mathrm{nl}}$ for all sufficiently large $d$;
enlarging them covers the finitely many remaining dimensions.  This proves
\assumptionlink[B1$'$]{B1prime}, and \eqref{eq:nonlinear-paired-dominance} then gives the
paired gate-state quietness bounds with the same constants.

For \assumptionlink[B2$'$]{B2prime}, equation~\eqref{eq:finite-clock-free-H} gives
\begin{equation}
  [H_{C,d},X[\nu]]=\nu X[\nu],
  \qquad
  \nu\in\omega_{0,d}\mathbb Z.
  \label{eq:finite-clock-Bohr-commutator}
\end{equation}
Thus every nonzero frequency in \eqref{eq:B2prime-C} satisfies
$|\nu|\ge\omega_{0,d}=2\pi/T_d$.  Termwise comparison of the nonnegative sums gives
\begin{equation}
  C_{\partial,d}
  \le\frac{1}{\omega_{0,d}}\Gamma_{\partial,d}
  =\frac{1}{2\pi}T_d \Gamma_{\partial,d},
  \label{eq:finite-clock-B2prime-proof}
\end{equation}
which proves \eqref{eq:B2prime-constant}.  Moreover,
$\dot I_{C,d}^{(l)}[0]=0$ because the zero-frequency block commutes with $H_{C,d}$.
This is why no division by zero occurs.
\end{proof}

\section{Full-Hamiltonian optimization at fixed assumption bounds}
\label{app:ir-regular-gaussian}

We test whether optimizing the Hamiltonian as well as the incoming state
changes the energy exponent under
Assumptions~\assumptionlink{A1}, \assumptionlink{A2},
\assumptionlink{B1}, and \assumptionlink{B2}.
All results in this section evolve the joint qubit--field state.
The code and data supporting the numerical results in this section
are available in Ref.~\cite{WoodsSquareRootCode}.

\subsection{A family with variable dispersion, coupling, and packet shape}
\label{app:irg-model}

The qubit basis is $\{\ket{0},\ket{1}\}$, with
$H_{Q,\lambda}=\omega_{q,\lambda}\ket{1}\bra{1}$ and
$\sigma_\lambda=\sigma_+:=\ket{1}\bra{0}$, with
$\sigma_-:=\sigma_+^\dagger$.
The one-particle Hamiltonian $h_\lambda$ acts by multiplication,
\begin{equation}
 (h_\lambda f)(k)=\omega_\lambda(k)f(k),
 \qquad
 \omega_\lambda(k)=|k|\,[1+\eta_\lambda(|k|-1)].
 \label{eq:vtest-dispersion}
\end{equation}
Thus $h_\lambda$ is distinct from the scalar field-energy expectation
$h_{F,\lambda}(t):=\Tr[\rho_\lambda^{\picsup{S}}(t)(\id_Q\otimes H_{F,\lambda})]$.
We search $0\le\eta_\lambda\le0.25$.
For the common reference wave number $k_{\mathrm{ref}}=1$, choose
\begin{equation}
 T_\lambda:=\frac{1}{\omega_\lambda(k_{\mathrm{ref}})}=1.
 \label{eq:vtest-reference}
\end{equation}
This rule uses the free dispersion alone and obeys the scale covariance
of Section~\ref{sec:reference-time}.  Fixing $\omega_\lambda(1)=1$
while changing its curvature, the coupling, and the qubit gap does not
uniformly fast-forward the Hamiltonian.

A smooth taper on the positive propagation branch is
\begin{equation}
 w(k):=\boldsymbol 1_{(2,6)}(k)
       \sin^2\!\left[\frac{\pi(k-2)}4\right].
 \label{eq:vtest-taper}
\end{equation}
The coupling and incoming packet profiles are
\begin{align}
 g_\lambda(k)
 &=\sqrt{\Gamma_{g_\lambda}}\,\mathcal N_{g,\lambda}\,
   w(k)\exp\!\left[-\frac{(k-k_{c,\lambda})^2}{2s_{g,\lambda}^2}
                   +i\beta_\lambda(k-k_{c,\lambda})^2\right],
 \label{eq:vtest-g}\\
 \xi_\lambda(k)
 &=\mathcal N_{\xi,\lambda}\,
   w(k)\exp\!\left[-\frac{(k-k_{c,\lambda})^2}{2s_{p,\lambda}^2}\right],
 \qquad
 \omega_{q,\lambda}=\omega_\lambda(k_{c,\lambda})+\Delta_\lambda .
 \label{eq:vtest-packet}
\end{align}
The positive normalizations satisfy
$\int|\xi_\lambda|^2\dd k=1$ and
$\int|g_\lambda|^2\dd k=\Gamma_{g_\lambda}$.
Here $k_{c,\lambda}$ is the spectral centre,
$s_{g,\lambda}$ and $s_{p,\lambda}$ control the coupling and packet widths,
$\beta_\lambda$ is a spectral chirp, and $\Delta_\lambda$ is the qubit
detuning from the dispersion at the profile centre.
The taper vanishes with its first derivative at the spectral edges.
The coupling is identically zero near $k=0$, so its infrared integrals
are finite.  Modes outside $(2,6)$ start in vacuum and factorize
exactly from the coupled dynamics.

With $b_\lambda:=\int\xi_\lambda(k)^*a_k\dd k$, prepare
\begin{align}
 \ket{\Psi_\lambda(0)}
 &=\ket{0}_Q\otimes
 D_{b_\lambda}(\alpha_\lambda)\ket{\mathrm{vac}},
 \qquad
 D_b(\alpha):=\exp(\alpha b^\dagger-\alpha^*b),
 \label{eq:vtest-state}\\
 \overline\omega_{\xi,\lambda}
 &:=\int\omega_\lambda(k)|\xi_\lambda(k)|^2\dd k,
 \qquad
 \alpha_\lambda=\sqrt{E_\lambda/\overline\omega_{\xi,\lambda}}>0 .
 \label{eq:vtest-energy}
\end{align}
The packet is centred at the interaction at preparation.
Every candidate at a given energy budget has the same total initial
energy above the interacting ground energy, not the same displacement.
The searched parameters lie in the fixed box
\begin{align}
 k_c&\in[2.5,5.2],&
 s_g,s_p&\in[0.45,1.2],&
 \eta&\in[0,0.25],
 \notag\\
 \Gamma_g&\in[0.04,0.13],&
 \Delta&\in[-0.3,0.3],&
 \beta&\in[-0.25,0.25].
 \label{eq:vtest-box}
\end{align}
This is a finite-dimensional search over functions, not optimization
over all possible form factors or field states.

\subsection{Full quantum propagation and the acceptance rule}

For numerical propagation, we first discretize the continuum Hamiltonian
using Gauss--Legendre quadrature on $(2,6)$, constructing its nodes and
weights from a symmetric tridiagonal eigenvalue problem
\cite{GolubWelsch1969,WoodsPlenio2016Gauss}.
This gives modes with
$G_j=\sqrt{w_j}\,g_\lambda(k_j)$,
$f_j=\sqrt{w_j}\,\xi_\lambda(k_j)$, and
$\omega_j=\omega_\lambda(k_j)$.
A free coherent displacement with
$\alpha_j(t)=\alpha_\lambda f_j e^{-i\omega_jt}$ gives, up to a scalar,
\begin{align}
 \widetilde H_\lambda(t)
 &=H_{Q,\lambda}+\sum_j\omega_j b_j^\dagger b_j
   +\sigma_+\widetilde B_\lambda
   +\sigma_-\widetilde B_\lambda^\dagger
   +d_\lambda(t)\sigma_++d_\lambda(t)^*\sigma_-,
 \notag\\
 \widetilde B_\lambda&:=\sum_jG_jb_j,\qquad
 d_\lambda(t):=\sum_jG_j\alpha_j(t).
 \label{eq:vtest-displaced-H}
\end{align}
The fluctuation field starts in vacuum.
The residual interaction with the $b_j$ is retained throughout.
Consequently this is an exact change of representation for the
discretized Hamiltonian, including back-action and entanglement.

We truncate only the \emph{total fluctuation occupation},
$\sum_j b_j^\dagger b_j\le N_{\mathrm{fl}}$.
The coherent displacement in \eqref{eq:vtest-state} is represented
analytically and is not restricted to $N_{\mathrm{fl}}$ photons.
Production refinement uses $24$ quadrature modes and $N_{\mathrm{fl}}=4$,
giving joint dimension $2\binom{28}{4}=40950$.
We integrate the full vector with the adaptive Dormand--Prince $5(4)$
Runge--Kutta pair \cite{DormandPrince1980}: the fifth-order formula advances
the state and the embedded fourth-order formula estimates the local error.
We use local vector-error tolerance $2\times10^{-10}$, and maximum step $0.015$.
The checks described below increase the occupation cutoff and mode
resolution.

For every candidate use the same accepted interval
\begin{equation}
 I_\lambda=[0,t_{99,\lambda}],\qquad
 t_{99,\lambda}:=\inf\{t>0:
       \bra{1}\rho_{Q,\lambda}^{\picsup{S}}(t)\ket{1}=0.99\},
 \quad
 \omega_{\mathrm{gate},\lambda}=\frac{2\pi}{t_{99,\lambda}}.
 \label{eq:vtest-gate}
\end{equation}
A cubic Hermite interpolation of the propagated vector locates the
first crossing.  Candidates without a crossing before $t=2$ are
discarded; their rates cannot exceed $\pi$, well below all selected
rates.  The initial target population is zero.
The fixed-Hamiltonian comparison uses precisely this preparation and
gate rule, so its comparison with the optimized family does not mix
different gate definitions.  These tests concern state transfer,
rather than the fidelity of an arbitrary-input qubit channel.

\subsection{Common constants and direct checks of the criteria}
\label{app:irg-assumptions}

We use the common bounds
\begin{equation}
 (C_{\mathrm{fld}},\delta_0,K_J,K_D,K_\omega)
   =(1.03,0.99,0.02,0.12,0.5).
 \label{eq:vtest-budgets}
\end{equation}
These are the $\lambda$-independent bounds entering
\assumptionlink{A1}, \assumptionlink{A2}, \assumptionlink{B1}, and
\assumptionlink{B2}, and hence the gate-speed bound.
The same upper and lower bounds are imposed on every candidate at
every energy.  Their measured values need not saturate these bounds.

\paragraph{\assumptionlink{A0}: scale-covariant uniform lower stability.}
All members use the same modes $\{a_k\}_{k\in\mathbb R}$ on a common
Fock space and hence the same vacuum $\ket{\mathrm{vac}}$.
Since $\omega_\lambda(k)\ge0$ for all $k$, the uncoupled bosonic
Hamiltonians obey the common free-field zero convention
\eqref{eq:common-free-energy-zero}:
\begin{equation}
 H_{F,\lambda}\ket{\mathrm{vac}}=0,
 \qquad \inf\operatorname{spec}(H_{F,\lambda})=0.
 \label{eq:vtest-common-free-zero}
\end{equation}
On the coupling support, $\omega_\lambda(k)\ge2$, so
$C_{g_\lambda}\le\Gamma_{g_\lambda}/2\le0.065$.
The coupling gap also makes $g_\lambda/\omega_\lambda$ square integrable.
With the quotient convention following \eqref{eq:Cg-definition},
square completion gives
\begin{align}
 H_\lambda
 &=\int_{\mathbb R}\dd k\,\omega_\lambda(k)
 \left(a_k^\dagger+\frac{g_\lambda(k)}{\omega_\lambda(k)}\sigma_+\right)
 \left(a_k+\frac{g_\lambda(k)^*}{\omega_\lambda(k)}\sigma_-\right)
 \notag\\
 &\quad +(\omega_{q,\lambda}-C_{g_\lambda})\ket{1}\bra{1}\ge0,
 \label{eq:vtest-square-completion}
\end{align}
because the search box \eqref{eq:vtest-box} has
$\omega_{q,\lambda}\ge2.2>C_{g_\lambda}$.
Moreover, $\ket{0}_Q\otimes\ket{\mathrm{vac}}$ has exactly zero energy,
so $E_{\mathrm{gs},\lambda}=0$.
The system spectral ceiling is
$\lambda_{\max}(H_{Q,\lambda})=\omega_{q,\lambda}$, and the same box gives
\begin{equation}
 \omega_{q,\lambda}
 \le5.2\bigl[1+0.25(5.2-1)\bigr]+0.3=10.96.
 \label{eq:vtest-system-ceiling}
\end{equation}
Together with \eqref{eq:vtest-reference}, this yields
\begin{equation}
 T_\lambda\bigl[\lambda_{\max}(H_{Q,\lambda})-E_{\mathrm{gs},\lambda}\bigr]
 =\omega_{q,\lambda}\le10.96<11=:K_0.
 \label{eq:vtest-A0}
\end{equation}
Thus \assumptionlink{A0}, including the common vacuum and free-field
zero convention and the uniform lower-stability condition
\eqref{eq:A0-reference-time-bound}, holds analytically throughout the
search box with one dimensionless constant $K_0=11$.

\paragraph{\assumptionlink{A1}: uniform field-energy control.}
The preceding identification $E_{\mathrm{gs},\lambda}=0$ and the
vanishing initial interaction and qubit expectations establish the
resource formula \eqref{eq:vtest-energy}.
Energy conservation and
$|\langle V_\lambda\rangle|
 \le2\sqrt{C_{g_\lambda}h_{F,\lambda}(t)}$ imply
\begin{align}
 E_\lambda
 &\ge h_{F,\lambda}(t)
       -2\sqrt{C_{g_\lambda}h_{F,\lambda}(t)},\notag\\
 \frac{h_{F,\lambda}(t)}{E_\lambda}
 &\le
 \left(\sqrt{\frac{C_{g_\lambda}}{E_\lambda}}
       +\sqrt{1+\frac{C_{g_\lambda}}{E_\lambda}}\right)^2
 <1.03\qquad(E_\lambda\ge400).
 \label{eq:vtest-A1}
\end{align}
Thus \assumptionlink{A1} holds analytically throughout the searched box
and energy range, independently of the sampled field-energy checks.

\paragraph{\assumptionlink{A2}.}
Since the target projector commutes with $H_{Q,\lambda}$, its
population agrees in the two pictures.
The variational characterization of trace distance and
\eqref{eq:vtest-gate} give
\begin{equation}
 \delta_\lambda=
 D(\rho_{\mathrm{out},\lambda},\rho_{\mathrm{in},\lambda})
 \ge0.99=\delta_0 .
 \label{eq:vtest-A2}
\end{equation}
We additionally compute $\delta_\lambda$ from
$\rho_{\mathrm{in},\lambda}$ and $\rho_{\mathrm{out},\lambda}$ and use
the positive spectral
projection $P_{+,\lambda}$ to evaluate the interaction activation
$\mathcal D_\lambda(t)$ in \assumptionlink{B1}.

\paragraph{\assumptionlink{B1}: direct evaluation.}
The same physical quiet window
\begin{equation}
 \mathcal W_\lambda=[8,9]
 \label{eq:vtest-window}
\end{equation}
is used for every member.
It is disjoint from the accepted gate and has gaps between it and the
gate, as allowed by \eqref{eq:B1-quiet-intervals}--\eqref{eq:B1-quiet-set}.
At each gate time $s$, let
\[
 A_\lambda(s):=[P_{Q,\lambda}(s),\sigma_+],
 \qquad
 M_{2,\lambda}(t):=
 \Tr[\rho_{F,\lambda}^{\picsup{S}}(t)B_\lambda^2].
\]
Using $[P,\sigma_-]=-[P,\sigma_+]^\dagger$ in
\eqref{eq:B1-disturbance-action} gives the following expression for the
interaction activation $\mathcal D_\lambda(t)$:
\begin{align}
 \mathcal D_\lambda^2(t)
 &=\sup_{s\in I_\lambda}
 \left\{a_\lambda(s)\mathcal J_\lambda^2(t)
       +b_\lambda(s)[\mathcal J_\lambda^2(t)+\Gamma_{g_\lambda}]
       -2\operatorname{Re}[c_\lambda(s)M_{2,\lambda}(t)]\right\},
 \label{eq:vtest-paired-D}\\
 a_\lambda(s)&:=
   \Tr[\rho_{Q,\lambda}^{\picsup{S}}(s)A_\lambda(s)^\dagger A_\lambda(s)],
 \notag\\
 b_\lambda(s)&:=
   \Tr[\rho_{Q,\lambda}^{\picsup{S}}(s)A_\lambda(s)A_\lambda(s)^\dagger],
 \qquad
 c_\lambda(s):=
   \Tr[\rho_{Q,\lambda}^{\picsup{S}}(s)A_\lambda(s)^2].
 \notag
\end{align}
Both the state and projector use the same gate time $s$.
The field moments use the independently chosen quiet time $t$.
The creation contribution is the canonical
$\langle B_\lambda B_\lambda^\dagger\rangle_t
=\mathcal J_\lambda^2(t)+\Gamma_{g_\lambda}$.
We evaluate the supremum over the propagated gate samples, retaining
the initial and final samples defining $\rho_{\mathrm{in},\lambda}$ and
$\rho_{\mathrm{out},\lambda}$, and integrate the quiet-time samples by the trapezoidal
rule, including the exact window boundaries.
A candidate is accepted only if both averages are below the fixed
$K_J=0.02$ and $K_D=0.12$.
These two bounds directly implement \assumptionlink{B1}.
The evaluation of the interaction activation $\mathcal D_\lambda(t)$
involves no additional maximization over unrelated qubit vectors or
projectors.

For the final selected family, the largest quiet averages are
\begin{equation}
 \max_\lambda\overline{\mathcal J_\lambda^2}
   \simeq0.0198616<0.02,\qquad
 \max_\lambda\overline{\mathcal D_\lambda^2}
   \simeq0.1198814<0.12,
 \label{eq:vtest-B1-values}
\end{equation}
with $T_\lambda=|\mathcal W_\lambda|=1$.
These are finite-grid numerical checks of the averages and gate-time
suprema.

\paragraph{\assumptionlink{B2}.}
The fixed coupling gap gives an analytic bound for every searched
dispersion and form factor:
\begin{equation}
 \frac{C_{g_\lambda}}{T_\lambda\Gamma_{g_\lambda}}
 =\frac{\int_2^6 |g_\lambda(k)|^2/\omega_\lambda(k)\,\dd k}
        {\int_2^6|g_\lambda(k)|^2\,\dd k}
 \le\frac12=K_\omega .
 \label{eq:vtest-B2}
\end{equation}
The selected values lie between $0.26502$ and $0.30178$.
The constraint is applied to the functions of each candidate, not
inferred from one reference Hamiltonian.

\subsection{Search, numerical convergence, and selected Hamiltonians}

We maximize $\omega_{\mathrm{gate},\lambda}T_\lambda$ at each of
\begin{equation}
 T_\lambda E_\lambda\in
 \{400,800,1600,3200,6400,12800,25600\}.
 \label{eq:vtest-energy-grid}
\end{equation}
The shared library contains $36$ grid profiles and $30$
seeded samples of the parameters in \eqref{eq:vtest-box}.
At each energy we then make $35$ local parameter proposals around
the best feasible candidate, using coordinate and seeded random
perturbations.  Including re-evaluation of the starting candidates,
this gives $714$ full quantum screening evaluations.
Screening uses $24$ modes and $N_{\mathrm{fl}}=3$.
We re-evaluate the two fastest feasible proposals per energy at
$N_{\mathrm{fl}}=4$ with tighter integration tolerances, together
with the best initial grid candidate, and choose the fastest candidate
that still meets \eqref{eq:vtest-budgets}.
One proposal fails the refined $K_D$ test and is excluded.
All accepted and rejected evaluations are retained in the data tables.

This is a bounded search in a specified seven-parameter family.
It gives attainable rates, hence lower bounds on the unrestricted
constrained optimum. 
The finite search itself is not forced to follow a square-root law.

The lowest- and highest-energy members, together with the member having
the largest quiet-time mean square of the field activation
$\mathcal J_\lambda(t)$, are additionally repeated
at $N_{\mathrm{fl}}=5$, tolerance $5\times10^{-11}$, and maximum
step $0.01$.  Their occupation-refinement changes are below
$0.001\%$ in gate rate and at most $0.583\%$ and $0.099\%$ in the
quiet-time mean squares of the field activation $\mathcal J_\lambda(t)$
and interaction activation $\mathcal D_\lambda(t)$, respectively.
These three finer results replace the
corresponding $N_{\mathrm{fl}}=4$ rows.
An independent mode-resolution check at fixed $N_{\mathrm{fl}}=3$
raises the quadrature from $24$ to $40$ nodes; the quiet averages
change by at most $0.000272\%$ and $0.000042\%$, respectively.
The largest norm drift is below $7\times10^{-9}$, and the reconstructed
total-energy error during the gate is below $2\times10^{-10}$
relative to the initial energy.
A single-mode limit agrees to better than $10^{-10}$ with the exact
coherent-state Jaynes--Cummings population \cite{Narozhny1981}:
\begin{equation}
 P_1(t)=\sum_{n=0}^{\infty}
 e^{-|\alpha|^2}\frac{|\alpha|^{2n}}{n!}
 \sin^2\!\left(\kappa\sqrt n\,t\right).
 \label{eq:vtest-JC-validation}
\end{equation}
Here the qubit starts in $\ket{0}$, the single mode is resonant with it,
and $\kappa$ is the single-mode coupling strength; the Poisson weights
are the coherent state's photon-number probabilities.
A direct Kronecker-product trace independently checks
\eqref{eq:vtest-paired-D}.
These checks assess numerical resolution, rather than constituting
rigorous continuum error bounds.

The selected parameters are listed in Table~\ref{tab:vtest-parameters}.
They are recorded at full precision in the accompanying CSV file.
In particular, the coupling power $\Gamma_{g_\lambda}$ varies from $0.09773$ to $0.11324$,
$\eta_\lambda$ ranges from $0$ to $0.06871$, and the qubit gap ranges
from $2.2771$ to $3.5449$.
These changes cannot arise from a common scalar multiplication of
the Hamiltonian because $\omega_\lambda(1)$ remains one.
The minimum qubit purity over the complete simulated trajectories
is about $0.5031$.
Here $C_{\mathrm{fld,sampled}}$ denotes the largest sampled ratio
$h_{F,\lambda}(s)/E_\lambda$ during the gate.
Figure~\ref{fig:vtest-diagnostics} summarizes the constraint checks and
the variation of the selected Hamiltonians.

\begin{table}[ht]
 \centering
 \small
 \caption{Selected full-Hamiltonian parameters.  The centre is $k_c$,
 the coupling and packet widths are $s_g,s_p$, the dispersion curvature
 is $\eta$, the coupling power is $\Gamma_g$, the qubit detuning is
 $\Delta$, and the coupling chirp is $\beta$.  Values are rounded;
 the full-precision CSV is used for reproduction.}
 \label{tab:vtest-parameters}
\begin{tabular}{r r r r r r r r}
\hline
$E$ & $k_c$ & $s_g$ & $s_p$ & $\eta$ & $\Gamma_g$ & $\Delta$ & $\beta$\\
\hline
400 & 3.336 & 0.789 & 0.723 & 0.0093 & 0.10219 & 0.1360 & 0.1822\\
800 & 2.908 & 0.883 & 0.741 & 0.0000 & 0.10549 & 0.0559 & 0.0899\\
1600 & 2.550 & 0.900 & 0.750 & 0.0000 & 0.10200 & 0.0000 & 0.0000\\
3200 & 2.502 & 0.861 & 0.847 & 0.0687 & 0.11324 & -0.0806 & 0.1380\\
6400 & 2.550 & 0.900 & 0.900 & 0.0000 & 0.10200 & 0.0000 & 0.0000\\
12800 & 2.740 & 0.964 & 0.812 & 0.0000 & 0.09773 & 0.1905 & 0.0856\\
25600 & 2.500 & 1.114 & 0.942 & 0.0000 & 0.10419 & -0.2229 & 0.0275\\
\hline
\end{tabular}
\end{table}

\begin{figure}[ht]
 \centering
 \includegraphics[width=0.98\textwidth]
 {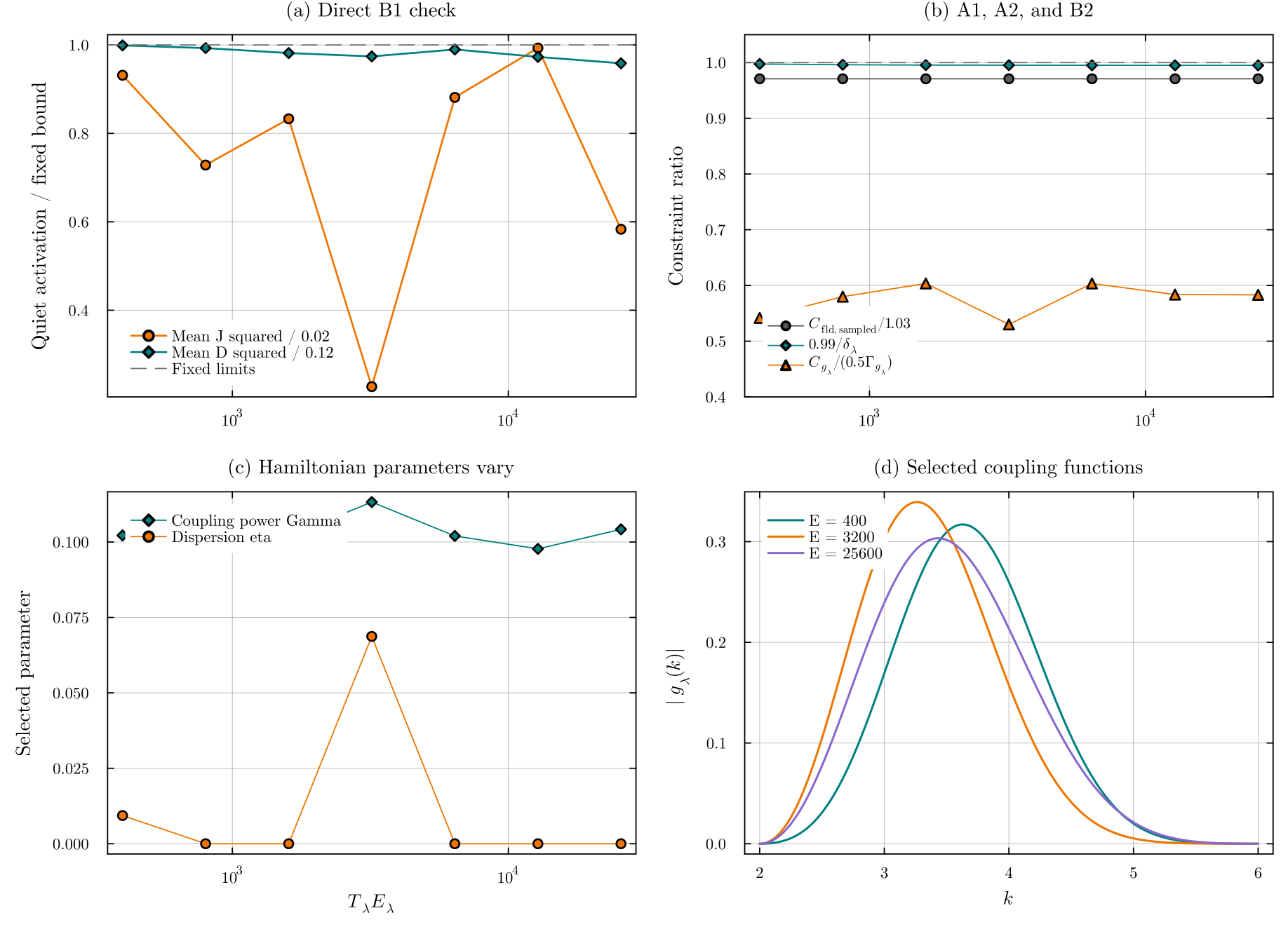}
 \caption{Checks on the optimized full quantum family.
 (a) Direct quiet averages divided by the fixed B1 limits;
 (b) the sampled field-energy ratio, reciprocal separation ratio,
 and B2 inverse-frequency moment, normalized by their common limits.
 All displayed constraint ratios must be at most one.
 (c) Changes in the coupling power and dispersion curvature.
 (d) Three selected coupling magnitudes.
 The different functions and parameters demonstrate genuine Hamiltonian
 variation across the energy family.}
 \label{fig:vtest-diagnostics}
\end{figure}

\subsection{Energy--speed scaling and its scope}

An unweighted ordinary least-squares regression of the seven selected
rates against $\sqrt{T_\lambda E_\lambda}$ gives
\begin{equation}
 \omega_{\mathrm{gate},\lambda}T_\lambda
 =0.741305\sqrt{T_\lambda E_\lambda}-0.703650,
 \qquad R^2=0.99937 .
 \label{eq:vtest-fit}
\end{equation}
The coefficient of determination uses vertical residuals of the rates.
A separate, unconstrained regression
$\log(\omega_{\mathrm{gate},\lambda}T_\lambda)
=p\log(T_\lambda E_\lambda)+\log A$ gives $p=0.5104$.
The normalized rates
$\omega_{\mathrm{gate},\lambda}T_\lambda/\sqrt{T_\lambda E_\lambda}$
range from $0.6880$ to $0.7426$.
The small finite-window deviation of $p$ from $1/2$ does not imply
an asymptotic violation of the upper bound.

For comparison, the fixed family uses
$(k_c,s_g,s_p,\eta,\Gamma_g,\Delta,\beta)
=(4,0.8,0.8,0,0.04,0,0)$
at the same energies, with the same preparation, window, constants,
and gate rule.  Its square-root fit has $R^2=0.99999$.
Optimization increases the attained rates by factors $1.68$--$1.74$.
With \eqref{eq:vtest-budgets},
\eqref{eq:quiet-relative-constant} gives
\begin{equation}
 c_{\mathrm{quiet}}
 =\frac{4\pi}{0.99}\sqrt{1.03(0.5)}
       \bigl(\sqrt{0.12}+\sqrt{0.02}\bigr)
 =4.44373 .
 \label{eq:vtest-theorem-prefactor}
\end{equation}
The upper panel of Figure~\ref{fig:irg-full-dynamics-extended} in the main text shows the
optimized rates, fixed-Hamiltonian comparison, square-root regression,
and common theorem bound.
Its lower panel resolves the squeezing-induced rate shifts discussed
in Section~\ref{app:vtest-squeezed}.
The observations support an approximately square-root dependence
even when the Hamiltonian is optimized and varies between members.

\clearpage
\subsection{Squeezed inputs at equal total energy}
\label{app:vtest-squeezed}

To test the sensitivity to nonclassical input fluctuations, we keep the
seven Hamiltonians and packet profiles of Table~\ref{tab:vtest-parameters}
and replace only the initial field state by
\begin{align}
 \ket{\Psi_{\lambda;r,\phi}(0)}
 &=\ket{0}_Q\otimes
   D_{b_\lambda}(\alpha_{\lambda;r})
   S_{b_\lambda}(re^{i\phi})\ket{\mathrm{vac}},
 \label{eq:vtest-squeezed-state}\\
 S_b(re^{i\phi})
 &:=\exp\!\left[\frac r2
     \left(e^{-i\phi}b^2-e^{i\phi}b^{\dagger2}\right)\right],
 \qquad
 \alpha_{\lambda;r}
 =\sqrt{\frac{E_\lambda}{\overline\omega_{\xi,\lambda}}-\sinh^2r}>0 .
 \label{eq:vtest-squeezed-energy}
\end{align}
The order is displacement after squeezing.  Since the initial qubit
and interaction energies vanish and $E_{\mathrm{gs},\lambda}=0$,
\begin{equation}
 E_\lambda=\overline\omega_{\xi,\lambda}
             \left(\alpha_{\lambda;r}^2+\sinh^2r\right).
 \label{eq:vtest-squeezed-budget}
\end{equation}
Thus squeezing is charged to the same total energy budget, rather than
added as a free resource at fixed displacement.
For the packet quadrature
$X_\theta=(e^{-i\theta}b_\lambda+e^{i\theta}b_\lambda^\dagger)/\sqrt2$,
the initial minimum and maximum variances are $e^{-2r}/2$ and $e^{2r}/2$.
We test $r=0.5,1$, giving variance factors relative to vacuum
$(0.368,2.718)$ and $(0.135,7.389)$, respectively.
At each energy and nonzero $r$, the phase is sampled at
$\phi=j\pi/4$, $j=0,\ldots,7$; $r=0$ is the coherent reference.
This is a phase search on the previously selected, energy-dependent
Hamiltonians, not a new joint global optimization of Hamiltonian and
squeezed input.

The propagation retains the complete joint state.  For numerical
efficiency we use the free qubit--field interaction picture and apply
the constant unitary $D_{b_\lambda}(\alpha_{\lambda;r})S_{b_\lambda}(re^{i\phi})$
to represent the initial residual field as vacuum.  With the quadrature
coefficients $G_j,f_j,\omega_j$ defined above, put
$B_0(t)=\sum_jG_je^{-i\omega_jt}b_j$ and
$u_\lambda(t)=\sum_jG_jf_je^{-i\omega_jt}$.  The transformed field operator is
\begin{equation}
 \widehat B_{\lambda;r,\phi}(t)
 =\alpha_{\lambda;r}u_\lambda(t)+B_0(t)
  +u_\lambda(t)\left[(\cosh r-1)b_\lambda
                     -e^{i\phi}\sinh r\,b_\lambda^\dagger\right],
 \label{eq:vtest-squeezed-transformation}
\end{equation}
where in this discretization $b_\lambda=\sum_jf_j^*b_j$.
The propagated interaction-picture Hamiltonian is
$e^{i\omega_{q,\lambda}t}\sigma_+\widehat B_{\lambda;r,\phi}(t)
+\mathrm{h.c.}$, including both annihilation and creation fluctuations.
This canonical change of representation is exact before the residual
occupation cutoff; it does not replace the field by its expectation.
In particular, the cutoff does not truncate the physical
squeezed-vacuum photon-number distribution.

The gate threshold, reference time, quiet interval $[8,9]$, and all
constants in \eqref{eq:vtest-budgets} are unchanged.
The proofs of \assumptionlink{A0}, \assumptionlink{A1}, and
\assumptionlink{B2} are state independent and still apply; an accepted
$0$--$0.99$ transfer implies \assumptionlink{A2} as before.
For \assumptionlink{B1} we recompute both the field activation
$\mathcal J_\lambda(t)$ and the interaction activation
$\mathcal D_\lambda(t)$ from each new joint trajectory, using
\eqref{eq:vtest-paired-D}; coherent-state values are not reused.

The initial phase library comprises $119$ joint trajectories, including
the seven coherent references, with $24$ modes and residual cutoff
$N_{\mathrm{fl}}=3$.  The fastest threshold-reaching and fastest feasible
phases are refined at $N_{\mathrm{fl}}=4$; if refinement rejects a
candidate, further phases from the library are tested.
The retained nonzero-$r$ phases are $\pi/4$ at $E=400,12800$ and zero at
the other five energies, for both tested squeezing strengths.

\paragraph{Resolution and common bounds.}
Six sensitive trajectories are repeated at $N_{\mathrm{fl}}=5$:
all three inputs at $E=400$, both squeezed inputs at $E=12800$,
and $r=1$ at $E=25600$.
The occupation-refinement changes are below $10^{-6}\%$ in rate
and at most $0.74\%$ and $0.12\%$ in the quiet-time mean squares
of the field activation and interaction activation, respectively.
These finer results replace the corresponding cutoff-$4$ values.
The closest interaction-activation case, $E=400,r=0.5$, is further
repeated at cutoff $6$: the two quiet averages change by $0.0252\%$
and $0.00392\%$, with the latter equal to $0.1199894<0.12$.
This cutoff-$6$ result is retained.
For the $E=12800,r=0.5$ field-activation constraint, an additional
cutoff-$5$--$6$ comparison with $20$ modes gives a $0.00962\%$ change
and a final quiet average $0.0199569<0.02$; the companion $20$- versus
$24$-mode cutoff-$5$ values differ by $0.0159\%$.
An independent joint quadrature/integration check on three $r=1$
cases raises $24$ to $40$ nodes at cutoff $3$ and tightens the
local tolerance from $2\times10^{-9}$ to $2\times10^{-11}$:
the quiet averages change by less than $0.0014\%$ and $0.0002\%$.

All retained points numerically satisfy the same bounds
$\overline{\mathcal J_\lambda^2}\le0.02$ and
$\overline{\mathcal D_\lambda^2}\le0.12$ on the same quiet window.
For the retained data, norm drift and relative reconstructed
total-energy error during the gate are both below $4\times10^{-9}$.
The squeezing extension also agrees with an independently prepared
single-mode squeezed-coherent state and the exact resonant
Jaynes--Cummings photon-number sum to better than $10^{-10}$,
and recovers the coherent propagation in a different representation.

Figure~\ref{fig:vtest-squeezed} compares these new trajectories with
the coherent data at the same total energies.  The upper panel shows
the near-overlap of the rates; the lower panel resolves the signed
percentage changes relative to the matched coherent trajectories.
Ordinary least-squares fits $\omega_{\mathrm{gate},\lambda}T_\lambda
=a\sqrt{T_\lambda E_\lambda}+b$ and separate unconstrained log--log
regressions give Table~\ref{tab:vtest-squeezed-fits}.

\begin{table}[ht]
	\centering
	\caption{Equal-energy squeezed-input comparison on the same seven
		energy-dependent Hamiltonians.  Each fit uses seven points.
		$R^2$ uses vertical rate residuals of the square-root fit, whereas
		$p$ is the slope of the separate log--log regression.}
	\label{tab:vtest-squeezed-fits}
	\begin{tabular}{c r r c c}
		\hline
		$r$ & $a$ & $b$ & $R^2$ & $p$\\
		\hline
		$0$   & $0.741305$ & $-0.703650$ & $0.99937$ & $0.5104$\\
		$0.5$ & $0.740449$ & $-0.571960$ & $0.99935$ & $0.5083$\\
		$1$   & $0.740914$ & $-0.619588$ & $0.99931$ & $0.5104$\\
		\hline
	\end{tabular}
\end{table}

The selected $r=0.5$ rates exceed the coherent rates by
$0.015\%$--$0.803\%$; for $r=1$ the signed changes range from
$-0.683\%$ to $0.899\%$.  At the largest energy the increases are
only $0.0165\%$ and $0.0226\%$, respectively.
Thus squeezing can give small finite-energy rate improvements in this
test: an assertion that it never helps would be incorrect.
The numerical evidence instead supports robustness of the approximately
square-root energy dependence to substantially changed quantum
fluctuations, without an increasing high-energy rate advantage.
The Hamiltonians vary across energies exactly as in the coherent
comparison, and the same assumption bounds remain imposed.

\begin{figure}[ht]
	\centering
	\includegraphics[width=0.97\textwidth]
	{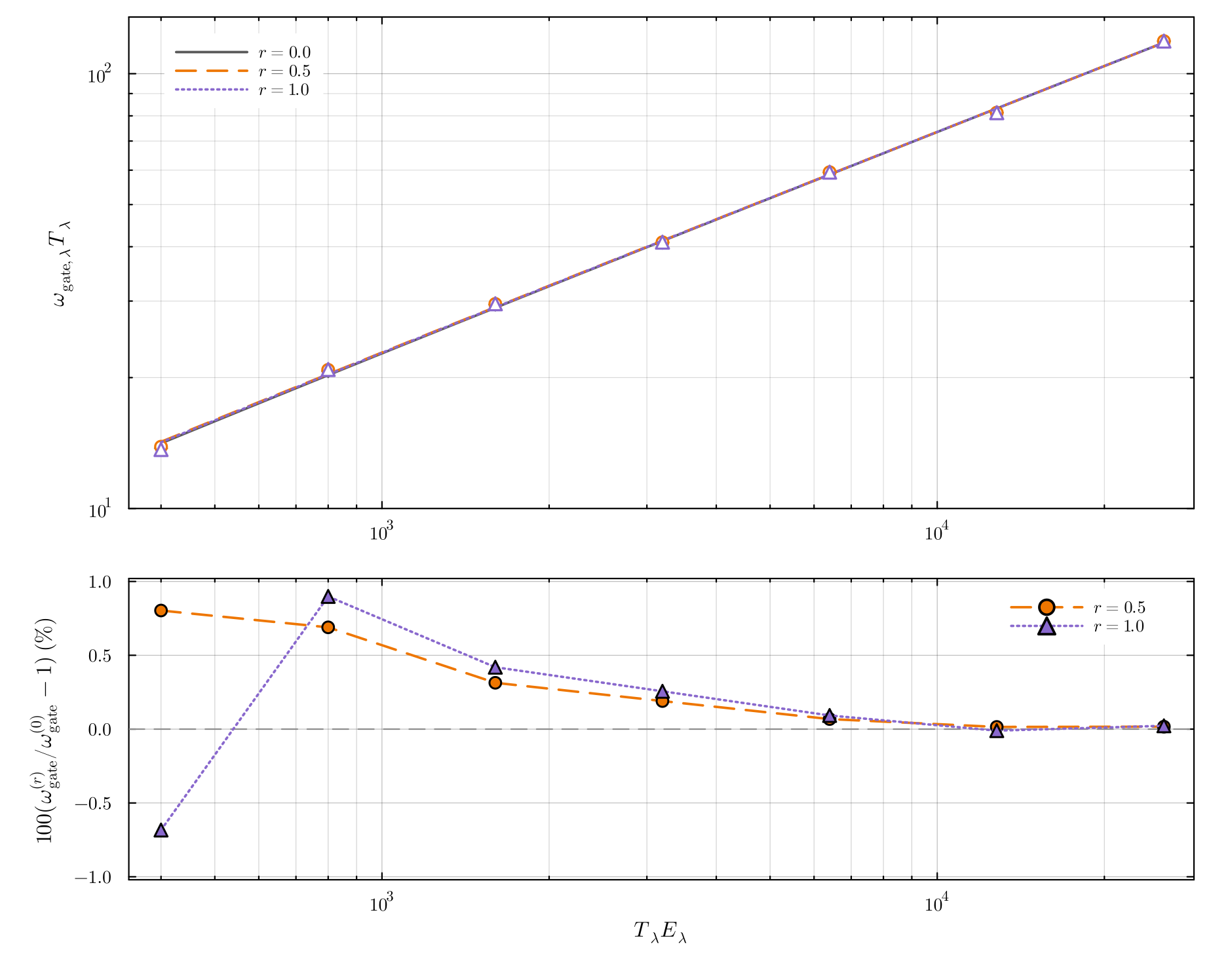}
	\caption{Full quantum dynamics with coherent and squeezed inputs at
		equal total energy, on the same energy-dependent Hamiltonians.
		Upper: new squeezed-state data and square-root fits nearly coincide
		with the coherent reference; $r=0.5,1$ change the initial extreme
		quadrature variances by factors $(0.368,2.718)$ and $(0.135,7.389)$
		relative to vacuum.  Lower: signed rate changes resolve the small
		phase-dependent improvements and losses.  The nonzero-$r$ phases are
		selected from the finite phase library subject to the unchanged
		$99\%$ transfer rule and common assumption bounds, including both
		field and interaction activation in B1.  The lines in the lower
		panel join data and are guides to the eye, not an asymptotic fit.}
	\label{fig:vtest-squeezed}
\end{figure}